\documentclass[acmsmall,screen,noacm, preprint]{acmart}
\setcopyright{none}                    % Removes ACM copyright box
\renewcommand\footnotetextcopyrightpermission[1]{} % Removes footnote clearance text
\AtBeginDocument{\pagestyle{plain}}
\makeatletter
\let\@runningfoot\@empty
\let\@firstfoot\@empty

\makeatother
\AtBeginDocument{%
  }

\usepackage{amsmath}
\usepackage{dsfont}
\usepackage{mathtools}
\usepackage{xcolor}
\usepackage{xspace}
\usepackage{xfrac}
\usepackage{stmaryrd}
\usepackage{enumerate}
\usepackage{algorithm}
\usepackage[noEnd=true, indLines=true]{algpseudocodex}
\algrenewcommand\algorithmicrequire{\textbf{Input:}}
\algrenewcommand\algorithmicensure{\textbf{Output:}}
\algnewcommand{\IfThen}[2]{% \IfThen{<if>}{<then>}
  \State \algorithmicif\ #1\ \algorithmicthen\ #2}
\usepackage{tikz}
\usetikzlibrary{calc,automata,positioning,arrows.meta,decorations}
\usetikzlibrary{decorations.pathreplacing}
\usepackage{wrapfig}
\usepackage{thmtools,thm-restate}

\newtheorem{assumption}{Assumption}
\newtheorem{remark}{Remark}
\newtheorem{problem}{Problem}

\newif\ifCOMMENTS
\COMMENTStrue   
\newcommand{\set}[1]{\left\lbrace #1\right\rbrace}

\newcommand{\pref}{\mathit{pref}}
\newcommand{\compl}[1]{\overline{#1}}
\newcommand{\bettereq}{\succeq}
\newcommand{\better}{\succ}
\renewcommand{\P}{\mathbb{P}}
\newcommand{\toas}{\xrightarrow{\text{a.s.}}}

\newcommand{\cyl}{\mathrm{Cyl}}

\mathchardef\mhyphen="2D % Define a "math hyphen"

\newcommand{\eve}{\mathsf{E}}
\newcommand{\adam}{\mathsf{A}}

\def\G{\mathcal{G}}

\newcommand{\V}{V}
\renewcommand{\v}{v}
\newcommand{\vinit}{v^0}
\newcommand{\VZ}{V_\eve}
\newcommand{\VO}{V_\adam}
\newcommand{\E}{E}

\newcommand{\pathsfin}{\mathit{Paths}_{\mathrm{fin}}}
\newcommand{\pathsinf}{\mathit{Paths}_{\mathrm{inf}}}

\newcommand{\lang}{\mathcal{L}}
\newcommand{\langZ}{\lang_\eve}
\newcommand{\langO}{\lang_\adam}

\newcommand{\pathZO}{\mathit{path}}
\newcommand{\last}{\mathit{last}}

\newcommand{\recurring}{\mathit{inf}}

\newcommand{\PZ}{Eve\xspace}
\newcommand{\PO}{Adam\xspace}

\newcommand{\pol}{\pi}
\newcommand{\polZ}{\pol_\eve}
\newcommand{\polZadaptive}{\polZ^{\mathrm{adapt}}}
\newcommand{\polO}{\pol_\adam}

\newcommand{\Pol}{\Pi}
\newcommand{\PolZ}{\Pol_\eve}
\newcommand{\PolZM}{\PolZ^{\mathrm{M}}}
\newcommand{\PolZFM}{\PolZ^{\mathrm{FM}}}
\newcommand{\PolO}{\Pol_\adam}

\newcommand{\BR}{\Gamma}
\newcommand{\static}{\mathtt{StrLib}}
\newcommand{\rank}{\mathtt{rank}}
\newcommand{\hittingtime}{\mathtt{HittingTime}}

\newcommand{\spec}{\varphi}
\newcommand{\reach}{\mathit{Reach}}
\newcommand{\safe}{\mathit{Safe}}
\newcommand{\buchi}{\mathit{B\ddot{u}chi}}
\newcommand{\cobuchi}{\mathit{Co\mhyphen B\ddot{u}chi}}
\newcommand{\gbuchi}{\mathit{GB\ddot{u}chi}}
\newcommand{\grabin}{\mathit{GRabin}}
\newcommand{\gr}[2]{\mathit{GR1}(#1\triangleright #2)}

\newcommand{\Aset}{\mathbf{A}}
\newcommand{\Gset}{\mathbf{G}}
\newcommand{\Bset}{\mathbf{B}}

\newcommand{\Asetfin}{\Aset_{\mathsf{fin}}}
\newcommand{\Asetinf}{\Aset_{\mathsf{inf}}}

\newcommand{\Gsetinf}{\Gset_{\mathsf{inf}}}

\newcommand{\M}{\mathcal{M}}

\newcommand{\Synt}{\mathsf{Solve\_Obliging\_Game}} 

\begin{document}

\title{Adaptive Strategies for GR(1) Games}

\author{S.\ Krishna}
\affiliation{%
 \institution{IIT Bombay}
 \city{Mumbai}
 \country{India}}
\email{krishnas@cse.iitb.ac.in}

\author{Kaushik Mallik}
\affiliation{%
 \institution{IMDEA Software Institute}
 \city{Madrid}
 \country{Spain}
}
\email{kaushik.mallik@imdea.org}

\author{Abhilasha Sharma Suman}
\affiliation{%
\institution{IIT Bombay}
\city{Mumbai}
\country{India}}
\email{23b1011@iitb.ac.in}


\begin{abstract}
We consider two-player GR(1) games on graphs, where the system player \PZ must satisfy the specification
\[
\Box\Diamond A_1\land\cdots\land\Box\Diamond A_m
\;\implies\;
\Box\Diamond G_1\land\cdots\land\Box\Diamond G_n
\]
against every strategy of the environment player \PO. Here $A_1,\ldots,A_m$ are \textit{assumptions} on the environment, $G_1,\ldots,G_n$ are \textit{guarantees} the system must provide, and $\Box\Diamond S$ denotes ``always eventually $S$'' in linear temporal logic.
Traditional approaches compute \textit{static} strategies that satisfy the specification against all possible adversarial moves.
Such strategies are overly conservative: they may ``cheat'' by actively
violating assumptions to trivially satisfy the implication, or abandon all
guarantees the moment any assumption is violated.
Existing methods to prevent such behaviors incur a doubly exponential
blowup in time and strategy complexity.

We introduce an \textit{adaptive} framework in which \PO is treated as a non-adversarial agent with unknown objectives that may violate some assumptions.
\PZ gives \PO the opportunity to fulfill all assumptions, but continuously monitors which assumptions \PO actually meets and adapts her strategy at runtime to maximize the number of satisfied guarantees.
Central to our approach is a novel algorithm for monitoring liveness properties of the form $\Box\Diamond S$.
Under a suitable model of non-adversarial behavior, this monitoring algorithm enables \PZ to maintain real-time likelihood estimates of which assumptions will be fulfilled.
Building on these estimates, we develop a probabilistic adaptation procedure: \PZ pre-computes a family of strategies, each optimal for a different subset of assumptions, and deploys a probability distribution over them that is dynamically adjusted based on the monitors' outputs.
We prove that when some assumptions are violated, \PZ's randomized adaptive strategy converges asymptotically to the deterministic strategy that fulfills the maximum number of guarantees.
The convergence proof relies on a detailed analysis of the time-varying probabilities governing the adaptive strategy.
Using a prototype implementation, we demonstrate that our approach scales better than state-of-the-art methods---completing where the baseline times out---while correctly adapting to the assumptions actually satisfied by the environment.
\end{abstract}

\begin{CCSXML}
<ccs2012>
   <concept>
       <concept_id>10003752.10003766.10003770</concept_id>
       <concept_desc>Theory of computation~Automata over infinite objects</concept_desc>
       <concept_significance>300</concept_significance>
       </concept>
   <concept>
       <concept_id>10011007.10011074.10011099.10011692</concept_id>
       <concept_desc>Software and its engineering~Formal software verification</concept_desc>
       <concept_significance>300</concept_significance>
       </concept>
 </ccs2012>
\end{CCSXML}

\ccsdesc[300]{Theory of computation~Automata over infinite objects}
\ccsdesc[300]{Software and its engineering~Formal software verification}

\keywords{Graph games, GR(1) synthesis, adaptive strategies, liveness monitoring.}

\received{20 February 2007}
\received[revised]{12 March 2009}
\received[accepted]{5 June 2009}

\maketitle
\thispagestyle{empty}   % arxiv

%!TEX root=main.tex

\section{Introduction}
\label{sec:intro}

Graph games are a core algorithmic problem in formal methods with
a wide range of applications, including automated synthesis of reactive
hardware, software, and cyber-physical systems.
In the turn-based variety, two players, called \PZ and \PO, take
turns moving a token along the edges of a given graph.
\PZ wins if the resulting path fulfills a given temporal
specification, typically expressed in linear temporal logic (LTL) or as an
$\omega$-automaton.
The goal is to compute a winning strategy for \PZ that guarantees victory
against every strategy of \PO.
In applications such as reactive program synthesis, \PZ represents the
program, \PO represents the environment, the game graph captures all
possible interactions between them, and every winning strategy of \PZ
corresponds to a provably correct implementation that satisfies the given
requirements under all possible environment behaviors.

We consider games with generalized reactivity~(1)
specifications~\cite{bloem2012synthesis}, better known as GR(1), where \PZ's
specification is drawn from the following fragment of LTL:
\begin{equation}\label{eq:GR1 formula}
\Box\Diamond A_1\land\ldots\land\Box\Diamond A_m \implies 
\Box\Diamond G_1\land\ldots\land\Box\Diamond G_n,
\end{equation}
where $A_1,\ldots,A_m$ are subsets of vertices known as the
\textit{assumptions}, $G_1,\ldots,G_n$ are subsets of vertices known as the
\textit{guarantees}, and ``$\Box\Diamond S$'' for a given set $S$ of vertices
represents the temporal property ``always eventually $S$ is visited.'' 
Typically, the assumptions specify some restrictions on \PO's moves, and the
GR(1) specification requires \PZ to fulfill the guarantees when the
assumptions are met.
We describe GR(1) games using our first running example.

\smallskip
\noindent\textbf{Example I: Vacuum cleaning robot.}
Consider the interaction between a vacuum cleaning robot and its environment,
where the robot must visit all rooms in the apartment infinitely
often, provided all doors open infinitely often.
This can be modeled as a GR(1) game between the robot and the set of doors,
taking the roles of \PZ and \PO respectively, where each vertex of the game
graph represents the current location of the robot and the status of all doors,
the edges represent how these configurations can change, each assumption $A_i$
is the set of vertices where the $i$-th door is open, and each guarantee $G_i$
is the set of vertices where the robot is in the $i$-th room.

\smallskip 
We identify three major limitations of the standard setting of GR(1) games.

Firstly, the usual approach is to model \PO as an adversary, which is a natural
artifact of the problem statement: find a strategy of \PZ that is winning
against \textit{every possible} strategy of \PO.
This is prohibitively conservative in most use cases.
For instance, in Example~I, an adversarial \PO would open a door when the
robot is far away and close it as soon as the robot tries to pass through.
This would make it impossible for \PZ to win, though due to unrealistic reasons.

Secondly, in the standard setting, if even a small subset of assumptions is
violated by \PO, the implication in Eqn.~\eqref{eq:GR1 formula} becomes true, and
\PZ wins automatically without fulfilling any guarantees.
In practice, we would like \PZ to fulfill as many guarantees as possible, when
some of the assumptions are unmet.
For instance, in Example~I, we would like the robot to visit as many rooms as
possible even when some doors remain permanently closed.
While this issue was considered by Bloem et al.~\cite{DBLP:conf/cav/BloemCGHJ10}
with the help of \textit{robust} strategies for \PZ, \PO was still modeled in
the adversarial fashion, and thereby could be overly conservative.
For Example~I, \PO could open a subset of doors \textit{and} use the
aforementioned adversarial strategy to prevent \PZ from fulfilling any guarantees.

Finally, in the standard setting, \PZ can ``cheat'' and win the game by
violating some of the assumptions herself, which is usually undesirable.
For instance, the robot in Example~I might just stand forever next to a
particular door blocking it to ever be opened, thereby violating the assumptions
and winning the GR(1) game.
This issue alone was resolved in the environmentally friendly GR(1)
synthesis framework~\cite{DBLP:conf/tacas/MajumdarPS19} or more generally in the
obliging games framework~\cite{chatterjee2010obliging}.
These works formalized \textit{graceful} strategies of \PZ that ensure
\PO always has a strategy to fulfill all assumptions from any history.
However, these works neither accommodate non-adversarial environments, nor
support robust solutions in the face of partial violations of the assumptions.
The na\"ive combination of the robust strategies with the obliging games
framework turns out to have doubly exponential time and strategy complexity in
the number of assumptions and guarantees, which is prohibitive in practical
applications.  

We present a framework that combines static reactive synthesis with dynamic adaptation.
A \emph{static} phase precomputes a library of optimal strategies, one for each possible subset of assumptions that \PO might fulfill; a \emph{dynamic} phase monitors the environment at runtime and adaptively mixes these strategies based on observed behavior.
This two-phase design rests on a key insight: in many applications---such as multi-agent systems---other players do not behave as pure adversaries; the adversarial assumption is merely an abstraction for unknown behavior.
Although computing a single robust strategy that supports every possible environment is intractable, maintaining a library of strategies and selecting among them online is feasible.
Indeed, our strategy-library construction enjoys exponentially smaller time and
strategy complexity compared to the na\"ive baseline (combining gracefulness
with robustness) described above.

Let us go into more details.
As the first step, we propose a clean, almost axiomatic definition of
non-adversarial behavior by \PO, which assumes:
(A)~the subset $\Asetinf$ of assumptions that \PO will fulfill infinitely often
does not depend on \PZ's strategy, as long as \PZ uses graceful strategies (allowing \PO to fulfill all assumptions); and
(B)~for every fixed strategy of \PZ and every assumption $A\in \Asetinf$, the
time between successive visits to $A$ is bounded.
In Example~I, this means the set $\Asetinf$ of doors that open infinitely often
does not depend on the robot's movement, as long as the robot uses graceful
strategies and does not block any doors. Moreover, for every fixed strategy of
the robot, the doors in $\Asetinf$ open within bounded intervals.
These restrictions rule out the adversarial behavior described earlier: \PO cannot strategically time door openings based on the robot's position.

%!TEX root=../main.tex
% Policy Pipeline Schematic Diagram

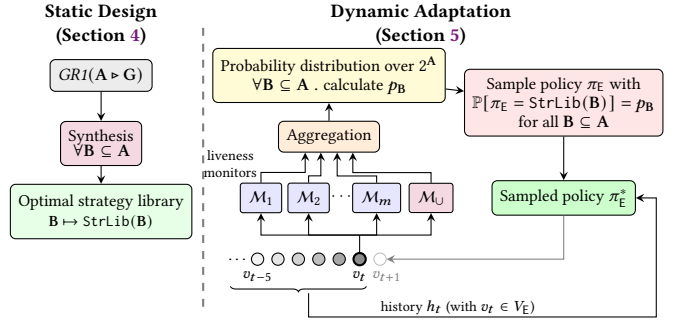
\begin{wrapfigure}{r}{0.6\textwidth}
\begin{tikzpicture}[
    scale=0.9,
    transform shape,
    % Styles
    block/.style={draw, rectangle, rounded corners=2pt, minimum width=1.1cm, minimum height=0.35cm, align=center, font=\scriptsize},
    monitor/.style={draw, rectangle, rounded corners=1pt, minimum width=0.6cm, minimum height=0.28cm, align=center, font=\tiny, fill=blue!10},
    library/.style={draw, rectangle, rounded corners=2pt, minimum width=1.2cm, minimum height=0.8cm, align=center, font=\scriptsize, fill=green!10},
    historynode/.style={circle, draw, minimum size=0.18cm, inner sep=0pt},
    arrow/.style={->, >=stealth, thin},
    phaselabel/.style={font=\footnotesize\bfseries},
]

% Dashed vertical line separating phases
\draw[dashed, thick, black!50] (-2, -2.6) -- (-2, 1.5);

% LEFT HALF: STATIC DESIGN PHASE
\node[phaselabel,align=center] at (-3.5, 1.55) {Static Design\\ (Section~\ref{sec:static design phase})};

% GR(1) spec input
\node[block, fill=gray!15] (spec) at (-3.5, 0.8) {$\gr{\Aset}{\Gset}$};

% Synthesis for all subsets
\node[block, fill=purple!15] (synthesis) at (-3.5, -0.2) {Synthesis\\[-0.5ex]$\forall \Bset\subseteq\Aset$};

% Policy library
\node[library] (library) at (-3.5, -1.2) {Optimal strategy library\\[0.05em]\tiny$\Bset\mapsto\static(\Bset)$};

% Arrows in static phase
\draw[arrow] (spec) -- (synthesis);
\draw[arrow] (synthesis) -- (library);

% RIGHT HALF: DYNAMIC ADAPTATION PHASE
\node[phaselabel,align=center] at (1.2, 1.55) {Dynamic Adaptation\\ (Section~\ref{sec:adaptive strategy})};

% HISTORY SEQUENCE AT BOTTOM
\node[historynode, fill=gray!10] (h1) at (-1.2, -1.9) {};
\node[historynode, fill=gray!20] (h2) at (-0.9, -1.9) {};
\node[historynode, fill=gray!35] (h3) at (-0.6, -1.9) {};
\node[historynode, fill=gray!50] (h4) at (-0.3, -1.9) {};
\node[historynode, fill=gray!70] (h5) at (0.0, -1.9) {};
\node[historynode, fill=gray!90, thick] (hcurr) at (0.3, -1.9) {};
\node[font=\tiny] at (0.3, -2.15) {$v_t$};
% \node[font=\tiny] at (0.0, -2.15) {$v_{t-1}$};
% \node[font=\tiny] at (-0.3, -2.15) {$v_{t-2}$};
% \node[font=\tiny] at (-0.6, -2.15) {$v_{t-3}$};
% \node[font=\tiny] at (-0.9, -2.15) {$v_{t-4}$};
\node[font=\tiny] at (-1.2, -2.15) {$v_{t-5}$};
\node[font=\tiny] at (-1.5, -1.9) {$\ldots$};

% Label for history
\node[font=\tiny, anchor=east] at (3, -2.6) {history $h_t$ (with $v_t\in \VZ$)};

% Brace under history
\draw[decorate, decoration={brace, amplitude=2pt, mirror}]
    (-1.6, -2.3) -- (0.4, -2.3);

% Dotted placeholder for next vertex
\node[historynode, draw=gray!50, fill=none] (hnext) at (0.6, -1.9) {};
\node[font=\tiny, gray] at (0.7, -2.15) {$v_{t+1}$};

% MONITORS - centered under aggregation
\node[monitor] (mon1) at (-1.15, -1.0) {$\M_1$};
\node[monitor] (mon2) at (-0.45, -1.0) {$\M_2$};
\node[font=\tiny] at (0.05, -1.0) {$\cdots$};
\node[monitor] (monm) at (0.55, -1.0) {$\M_m$};
\node[monitor, fill=purple!15] (monunion) at (1.35, -1.0) {$\M_{\cup}$};
\node[font=\tiny,align=center] at    (-1.6,-0.5)   {liveness\\ monitors};

% Arrow from current vertex to monitors
\draw[thin] (hcurr.north) -- ++(0, 0.25) coordinate (vbranch);
\draw[arrow] (vbranch) -| (mon1.south);
\draw[arrow] (vbranch) -| (mon2.south);
\draw[arrow] (vbranch) -| (monm.south);
\draw[arrow] (vbranch) -| (monunion.south);

% AGGREGATION - centered above monitors
\node[block, fill=orange!15, minimum width=1.3cm] (aggregation) at (-0.15, -0.1) {Aggregation};

% Arrows from monitors to aggregation - straight up
\draw[arrow] (mon1.north) -- ++(0, 0.15) -| ([xshift=-0.35cm]aggregation.south);
\draw[arrow] (mon2.north) -- ++(0, 0.1) -| ([xshift=-0.12cm]aggregation.south);
\draw[arrow] (monm.north) -- ++(0, 0.1) -| ([xshift=0.12cm]aggregation.south);
\draw[arrow] (monunion.north) -- ++(0, 0.15) -| ([xshift=0.35cm]aggregation.south);

% Probability distribution
\node[block, fill=yellow!20, minimum width=1.3cm,align=center] (probdist) at
(-0.15, 0.8) {Probability distribution over $2^{\Aset}$\\ $\forall \Bset\subseteq
\Aset\;.\;\text{calculate }p_{\Bset}$};
\draw[arrow] (aggregation) -- (probdist);

\node[block, fill=red!10, minimum width=1.3cm,align=center] (sampler) at (3.3,
0.4) {Sample policy
$\polZ$ with\\ $\P[\polZ=\static(\Bset)]=p_{\Bset}$\\ for all $\Bset\subseteq \Aset$ };
\draw[arrow] (probdist) -- (sampler);

% % CONNECTION: Sample selects from library
% \draw[->, >=stealth, dashed, blue!70!black] (sampler.west) -- (library.north east);

% OUTPUT: Action block on the right
\node[block, fill=green!20] (action) at (3.3, -1) {Sampled policy $\polZ^*$};

% % Library to action
% \draw[arrow] (library.east) -- (action.west)
%     node[midway, above, font=\tiny] {$\static(\Bset)(h_t)$};

% Sampler to action
\draw[arrow]    (sampler)   --  (action);

% History feeds into sampled policy
\draw[arrow] (-0.45, -2.4) -- (-0.45, -2.75) -| ($(action.east)+(0.3,0)$)   -- (action.east);

% Action to next history node
\draw[arrow, gray] (action.south) |- (hnext.east);

\end{tikzpicture}
\caption{Adaptive strategy pipeline. \textbf{Left:} Static design synthesizes $\static(\Bset)$ for each $\Bset\subseteq\Aset$.
\textbf{Right:} At runtime, monitors estimate the likelihood of satisfaction of
each individual assumption, which are aggregated into the likelihood of
satisfaction of subsets of assumptions, which is finally turned into a
probability value $p_{\Bset}$ for each $\Bset\subseteq \Aset$; a sampled
$\polZ^*$ then observes the history to produce the next vertex $v_{t+1}$.}
\label{fig:adaptation schematic}
\end{wrapfigure}
Adaptive strategies for \PZ comprise several components spread across
offline and online phases; see Figure~\ref{fig:adaptation schematic}.
In the offline design phase, we create a \textit{static strategy library}
$\static$ for \PZ, which, for every subset $\Aset^*$ of assumptions $\Aset$,
stores a graceful strategy $\polZ^*=\static(\Aset^*)$ of \PZ that fulfills as
many guarantees as possible when only $\Aset^*$ is fulfilled.
For Example~I, for every subset of doors, $\static$ would store one graceful
strategy that visits every room accessible through those particular doors infinitely often.
% Simple examples show that the number of guarantees actually fulfilled by
% $\polZ^*$ against $\Aset^*$ depends on the current game state,
% but regardless, this number is maximal.
If \PZ knew which assumptions (a non-adversarial) \PO will fulfill, i.e., the
identity of $\Asetinf$, she could simply deploy the corresponding
optimal strategy $\static(\Asetinf)$ from the library. 

In the online deployment phase, the identity of $\Asetinf$ is \textit{learned}
from past observations, and the statically computed optimal
strategies are accordingly mixed using a time-varying probability distribution.
Assumption learning is facilitated by a novel monitoring algorithm for
liveness (B\"uchi) properties of the form ``$\Box\Diamond S$,'' where $S$ is a set of states.
It is well-known that liveness properties like these are not
monitorable~\cite{diekert2015note}, because given any history of observations,
there would be some continuations fulfilling the property and some violating it,
making it impossible to ever reach a conclusive verdict.
Our monitor emits a quantitative verdict in $(0,1)$ indicating the
likelihood that $\Box\Diamond S$ is satisfied.
The \textit{hitting time} of $S$ is the time between successive visits to $S$.
Under the assumption of a constant (unknown) bound on the hitting time,%
\footnote{A constant bound on the hitting time does not mean that $S$ needs to
be visited after a fixed (constant) interval. Rather, it means that there is a
constant $k$ such that every two successive visits to $S$ are \textit{at most} $k$ time steps apart.}
as time goes to $\infty$, the monitor's verdict converges to $1$ when the trace
satisfies $\Box\Diamond S$, and to $0$ otherwise.%
\footnote{In Appendix~\ref{sec:appendix:monitor generalization}, we show that
the claim holds for a significantly more general case, namely hitting times
bounded by any sub-exponential function.
However, for the main technical development, the constant hitting time assumption suffices.}
This is why we call our monitor \textit{asymptotic}; it is of independent
interest from a runtime verification perspective.

We use one asymptotic monitor $\M_i$ for each assumption $\Box\Diamond A_i$ in
$\Aset = \{A_1,\ldots,A_m\}$.
These monitors output the likelihood of the fulfillment of each assumption
separately, and these individual likelihood estimates get combined to form an
aggregated verdict $w_{\Aset^*}$ about how likely it is that a given
\textit{nonempty set} $\Aset^*$ of assumptions is actually the set $\Asetinf$ of
infinitely visited assumptions.
We also track the likelihood $w_\emptyset$ that $\Asetinf=\emptyset$, but for
that we need to use a separate monitor $\M_\cup$ for $\Box\Diamond \cup_{i\in
[1;m]} A_i$.
Next, these likelihood estimates for sets of assumptions are turned into a
probability distribution over the subsets of $\Aset$, where the relative
strengths of the estimates are proportional to the relative probability values.
For Example~I, for every given (empty or nonempty) subset of doors, we would
obtain a probability that those are the only doors that will open infinitely
often. 
As the verdicts of the monitors evolve after each step, so does this probability
distribution.

The obtained probability distribution then gives rise to a probabilistic
adaptive strategy for \PZ: at each step, \PZ samples a strategy from the library
according to the current distribution, and uses the sampled strategy to decide
the next action. 
For every arbitrary strategy of \PO, the adaptive strategy of \PZ gives rise to a Markov chain whose states represent the history of vertices and
the edge probabilities represent the sampling probability distribution over the static strategies.
We prove that, under the assumption that \PO is non-adversarial and as time goes
to infinity, the Markov chain almost surely concentrates on the component where
the static strategy $\static(\Asetinf)$ corresponding to the set $\Asetinf$ is
always sampled with probability $1$.

The elegance of our adaptive strategy is in its simplicity: both the monitor and
the probabilistic strategy are surprisingly simple to implement.
Yet, the convergence proof is highly nontrivial.
Part of the difficulty stems from a circular dependency between the monitor's
verdicts and the strategy output: the output of the adaptive strategy influences the
path taken on the game graph, which affects the frequency of the assumptions
visited along the path, which in turn changes the output of the monitor.
We break this circularity by establishing a lower bound on the probability of
selecting any given static strategy for a given number of steps in a row,
regardless of the history of vertices seen so far.
This result is at the heart of our technical proofs of convergence.

We implemented our algorithm in a prototype tool called \textit{Grace}, written in C++.
As baseline, we implement the straightforward combination of robust strategies~\cite{DBLP:conf/cav/BloemCGHJ10} with obliging games~\cite{chatterjee2010obliging}.
On a suite of benchmarks drawn from the literature, \textit{Grace} outperforms the baseline as soon as the number of assumptions and guarantees grows, consistent with the exponential savings predicted by our complexity analysis.
Moreover, while the adaptive strategy requires exponential time in theory to identify the true set $\Asetinf$ of fulfilled assumptions, in practice it converges within a few hundred time steps across all benchmarks.

We now introduce a second motivating example, which is slightly more technical
than Example~I, but will be used throughout the paper to precisely illustrate
different concepts.

\smallskip
\noindent\textbf{Example II: Repeated matching pennies.}
In this game, \PZ and \PO each have a penny, \PZ's heads and tails are called
$h_\eve$ and $t_\eve$, and the same for \PO are called $h_\adam$ and $t_\adam$,
respectively; throughout the paper, $\eve$ will indicate Eve and $\adam$ will
indicate Adam.
In each round of the game, first \PZ picks a side of her penny, which \PO
observes, and then either picks a side of his own penny, or withdraws from that
round, which is represented as the action $\bot$.
There is no restriction on how the players pick their sides---they may do so deterministically
or randomly.
The outcome of a round is a \textit{match} if \PO decides to play and both players
choose the same side; otherwise, it is not a match.
This finishes one round, and the game repeats indefinitely.
\PZ wins if she can \textit{guarantee} infinitely many matches, under the
\textit{assumption} that \PO picks both sides of his coin infinitely often.
The game graph is depicted in Figure~\ref{fig:matching pennies game graph}, and
\PZ's GR(1) specification is given as:
\begin{equation}\label{eq:matching pennies:GR(1) specification}
\Box\Diamond\, H_\adam \land \Box\Diamond\, T_\adam \implies \Box\Diamond\, \set{h_\eve h_\adam} \land \Box\Diamond\, \set{t_\eve t_\adam},
\end{equation}
where the assumptions $H_\adam \coloneqq \set{h_\eve h_\adam,t_\eve h_\adam}$ and $T_\adam \coloneqq
\set{h_\eve t_\adam,t_\eve t_\adam}$ represent the states where \PO has picked $h_\adam$ and $t_\adam$,
respectively, and guarantees $h_\eve h_\adam$ and $t_\eve t_\adam$ represent
matches.
\begin{wrapfigure}{r}{0.5\textwidth}
\centering

\centering
\begin{tikzpicture}[
    scale=0.8,
    transform shape,
    state/.style={circle, draw, minimum size=0.5cm, inner sep=1pt, font=\scriptsize},
    rectstate/.style={rectangle, draw, minimum size=0.5cm, inner sep=1pt, font=\scriptsize},
    every edge/.style={draw, ->, >=stealth, thin},
    node distance=0.8cm
]
    \node[state,initial above] (s) at (0,0) {$s$};
    \node[rectstate] (t0) [right=of s] {$t_\eve$};
    \node[rectstate] (h0) [left=of s] {$h_\eve$};
    \node[state] (t0h1) [above right=0.5cm and 0.6cm of t0] {$t_\eve h_\adam$};
    \node[state] (t0t1) [below right=0.5cm and 0.6cm of t0] {$t_\eve t_\adam$};
    \node[state] (h0h1) [above left=0.5cm and 0.6cm of h0] {$h_\eve h_\adam$};
    \node[state] (h0t1) [below left=0.5cm and 0.6cm of h0] {$h_\eve t_\adam$};
    \node[state] (bot) [below=of s] {$\bot$};

    \path[->]
        (s) edge (h0)
            edge (t0)
        (h0) edge (h0h1)
             edge (h0t1)
             edge (bot)
        (t0) edge (t0h1)
             edge (t0t1)
             edge (bot)
        (h0h1) edge (s)
        (h0t1) edge[bend right=15] (s)
        (t0h1) edge (s)
        (t0t1) edge[bend left=15] (s)
        (bot) edge (s);
\end{tikzpicture}
\caption{Game graph for repeated matching pennies. Circular vertices are controlled by \PZ; square vertices by \PO.}
\label{fig:matching pennies game graph}
\end{wrapfigure}
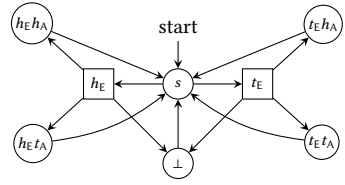

It is worthwhile to note that every strategy of \PZ in this game is graceful,
since she cannot prevent \PO from fulfilling the assumptions; this will allow us
to focus on the core idea of strategy adaptation.
Furthermore, the game is winning for \PZ, and every winning strategy of \PZ must
\textit{eventually always} use the following strategy:
keep picking $h_\eve $ until \PO picks $h_\adam$, then switch to $t_\eve $ and keep picking
$t_\eve $ until \PO picks $t_\adam$, then switch back to $h_\eve $ and keep picking $h_\eve $
until \PO picks $h_\adam$, and so on; let $\PolZ^{(1)}$ be the set of all \PZ
strategies that use this tactic.
Clearly, if \PO eventually stops picking one of the two sides, \PZ wins due to the violation of at least one of the assumptions.
On the other hand, if \PO picks both sides always eventually, \PZ will fulfill
both guarantees $(h_\eve ,h_\adam)$ and $(t_\eve ,t_\adam)$.
In other words, regardless of \PO's actions, the GR(1) specification will be
fulfilled.
It can be shown that strategies of \PZ that do not use this tactic will lose
the GR(1) game, and therefore are not optimal.
To see this, take any strategy of \PZ that always eventually switches from
$t_\eve $ to $h_\eve $ before \PO produces a $t_\adam$, and then switches back
before \PO produces $h_\adam$; let $\PolZ^{(2)}$ be the set of strategies of \PZ
implementing this tactic, and note that $\PolZ^{(2)}$ is the complement of $\PolZ^{(1)}$.
When \PZ uses a strategy in $\PolZ^{(2)}$, \PO can punish \PZ as follows: \PO
selects $t_\adam$ as soon as \PZ switches from $t_\eve$ to $h_\eve $, selects
$h_\adam$ as soon as \PZ switches from $h_\eve$ to $t_\eve $, and selects
$\bot$ if no change happened in the last two moves of \PZ.
Since policies in $\PolZ^{(2)}$ will always eventually make the switch
regardless of \PO's moves, the assumptions will be fulfilled but the guarantees
will be violated, spoiling the GR(1) specification.

A non-adversarial \PO must have decided a partition $\Asetinf$ and $\Asetfin$ of
the set $\Aset$ of all assumptions---regardless of \PZ's actions, and must make
sure that against every fixed strategy of \PZ, each assumption  $A\in\Asetinf$
is visited within a bounded delay from the previous visit.
For instance, a valid non-adversarial strategy of \PO eventually always
alternates between $\bot$ and $t_\adam$; let $\PolO^{(1)}$ denote strategies of \PO fulfilling this pattern.
Here, $\Asetinf = \set{T_\adam}$ and $\Asetfin=\set{H_\adam}$, and $T_\adam\in \Asetinf$ is
visited within every two steps---regardless of what \PZ picks.
On the other hand, the set of strategies of \PO that eventually always picks the
opposite of \PZ's choices---call it $\PolO^{(2)}$---does not meet our non-adversarial criterion:
$\Asetinf$ and $\Asetfin$ will depend on \PZ's strategies, namely
$\Asetinf=\set{H_\adam,T_\adam}$ and $\Asetfin=\emptyset$ when \PZ uses a strategy from
$\PolZ^{(2)}$, and either $\Asetinf=\set{T_\adam}$ and $\Asetfin=\set{H_\adam}$ or
$\Asetinf=\set{H_\adam}$ and $\Asetfin=\set{T_\adam}$ if \PZ uses a strategy from $\PolZ^{(1)}$. 
Furthermore, if \PZ uses a strategy $\polZ\in \PolZ^{(2)}$ such that the $i$-th
switch---either from $h_\eve $ to $t_\eve $ or from $t_\eve $ to $h_\eve $---happens at time
$2^i$, then \PO's response strategy will fulfill the assumptions after increasing
amounts of delay, which is not permitted according to our bounded delay
criterion of non-adversarial strategies of \PO.

Assuming \PO uses only non-adversarial strategies, our goal is to obtain a
strategy of \PZ that \textit{learns} which assumptions \PO will fulfill
via runtime monitoring, and accordingly \textit{adapts} her strategy to maximize the number of guarantees.
For example, if \PO uses a non-adversarial strategy $\polO^{*}$ fulfilling only $T_\adam$ but
violating $H_\adam$, i.e., $\Asetinf=\set{T_\adam}$ and $\Asetfin=\set{H_\adam}$, then \PZ
must learn $\Asetinf$ and $\Asetfin$, and adapt her strategy so that she only
picks $t_\eve $ in the long run.
(Note that, even if \PO fulfills $T_\adam$, he need not constantly pick
$t_\adam$ as he can also choose $\bot$ occasionally.)

\subsubsection*{Contributions}
The highlights of our contributions are as follows:
\begin{itemize}
\item A formal definition of non-adversarial environment behavior, and a
formalization of the problem of computing the optimal graceful strategies of \PZ
(Section~\ref{sec:practical relaxations}).

\item A novel \textit{asymptotic liveness monitor} for $\Box\Diamond S$ properties whose verdict converges to the true satisfaction value under mild hitting time assumptions (Section~\ref{sec:liveness monitors}).

\item An \textit{adaptive strategy framework} combining offline computation of a
graceful strategy library (Section~\ref{sec:static design phase}) with online
probabilistic mixing guided by monitor outputs (Section~\ref{sec:randomized mixing}).

\item A convergence proof showing that our adaptive strategy asymptotically
concentrates on the optimal deterministic strategy (Theorem~\ref{sec:proof of convergence}).

\item A prototype implementation \emph{Grace} showing exponential speedups over the baseline and fast convergence in practice (Section~\ref{sec:experiments}).
\end{itemize}

\subsection{Related Work}

GR(1) synthesis~\cite{bloem2012synthesis} provides an exponential improvement over classical LTL synthesis~\cite{pnueli1989synthesis}, and has been extended to richer specifications~\cite{ehlers2024fully,almagor2017quantitative}, infinite-state systems~\cite{maderbacher2026parameterized}, and asynchronous systems~\cite{klein2012effective}.
We build on the robust synthesis problem of Bloem et al.~\cite{DBLP:conf/cav/BloemCGHJ10}, but differ fundamentally: their work, along with others~\cite{bloem2009better,topcu2012synthesizing,ehlers2014resilience}, computes static strategies that make worst-case assumptions about the environment; in contrast, we introduce adaptive strategies that start optimistically and adapt based on observed runtime behavior, while being exponentially more efficient.
Works on environment-friendly strategies~\cite{bloem2015cooperative,DBLP:conf/tacas/MajumdarPS19,chatterjee2010obliging} do not support the robustness criterion.

We propose a novel framework for adaptive GR(1) strategies that adjust at runtime based on observed environment behavior.
While adaptive design is common in reinforcement learning~\cite{khetarpal2022towards} and adaptive control~\cite{aastrom1995adaptive}, its use in formal methods is rare.
Casting our problem as a learning problem would not provide the strong formal guarantees we achieve through verified static strategies and their adaptive mixing.

We also introduce a novel class of liveness monitors, which is of independent interest.
While monitors are integral to runtime verification~\cite{bartocci2018introduction,medhat2015runtime,grieser2020assuring,henzinger2023monitoring}, liveness properties are not monitorable~\cite{diekert2015note}.
Previous approaches approximated liveness via safety properties using domain-specific tricks~\cite{yao2024mostly}.
In contrast, our monitors' verdicts converge to the true satisfaction value under mild frequency assumptions that cannot be encoded by any safety relaxation.

Few works in formal methods consider strategies that adapt to non-adversarial environments.
Nayak et al.~\cite{nayak2022robustness} proposed adaptive strategies for rLTL objectives that adapt against sub-optimal environment actions, which differs from our adaptation against assumption violations.
Multi-environment MDPs~\cite{raskin2014multiple} learn the true environment model from a candidate set while fulfilling temporal objectives; our setting similarly learns which assumptions hold, but deploys corresponding optimal strategies.

Our technique of mixing strategies resembles other works, though context and analysis differ.
Anand et al.~\cite{anand2026follow} use strategy templates for liveness shields that probabilistically nudge systems towards target states.
Avni et al.~\cite{avni2026decoupled} blend single-objective strategies to achieve multiple objectives.
While these works also mix strategies, our setting differs and poses unique challenges like the circular dependency between the probability distribution over strategies and the frequency of assumption occurrences.

%!TEX root=main.tex

\section{Preliminaries}
\label{sec:prelims}

\noindent\textbf{Notation.}
Given an alphabet $\Sigma$, we will write $\Sigma^*$ and $\Sigma^\omega$ to
denote the set of all finite and infinite words over $\Sigma$, respectively, and
will write $\Sigma^\infty$ to denote the set $\Sigma^*\cup \Sigma^\omega$.

\smallskip
\noindent\textbf{Game graphs.}
A game graph $\G$ is a tuple $(\V,\VZ,\VO,\E)$, where 
$\V$ is a finite set of vertices,
$\VZ$ and $\VO$ form a partition over $\V$, i.e., $\VZ\cup \VO = \V$ and $\VZ\cap \VO = \emptyset$,
and $\E \subseteq \V\times \V$ is a set of (directed) edges.
We will interchangeably interpret $\E$ as follows: for every $\v\in \V$, $\E(\v) \coloneqq \set{\v'\in \V\mid (\v,\v')\in \E}$.
Moreover, we will assume that game graphs do not have dead-ends, i.e., $\E(\v)\neq \emptyset$ for every $\v$.

The game is played between two players, called \PZ and \PO, who \emph{own} the vertices in $\VZ$ and $\VO$, respectively.
Initially, a token is placed on one of the vertices, and at each subsequent
step, the owner of the current vertex moves the token to one of the successor vertices.
This gives rise to an infinite sequence of vertices, called a \textit{path}, and \PZ wins the game if this path belongs to a specified set of ``good'' paths.
Below, we formalize this mechanics of the game.

\smallskip
\noindent\textbf{Paths and strategies.}
A \emph{path} in the game graph $\G$ starting at a given initial vertex $\v^0\in
\V$ is a sequence $\v^0\ldots \v^k\in \V^\infty$, where for every $i\geq 0$,
$\v^{i+1}\in \E(\v^i)$.%
\footnote{For a variable $x$ with time-varying valuations, we write $x^t$ to
denote the value of $x$ at the time $t$, and write $(x^t)^k$ to denote the value
of $x^t$ raised to the power $k$. For a constant $c$, we simply write $c^k$ to
denote the value of $c$ raised to the power $k$. These usages will always be
clear from the context. The subscripts are reserved for associating some
identities to the symbols, like in $\VZ$, $\VO$, and different assumptions $A_1,\ldots,A_m$ in Eqn.~\eqref{eq:GR1 formula}.}
A path can be either finite or infinite, and we will write $\pathsfin(\G)$ and $\pathsinf(\G)$ to denote, respectively, the sets of all finite and infinite paths in $\G$.
The last vertex of every finite path $\rho = \v^0\ldots \v^k$ will be denoted as $\last(\rho)\coloneqq\v^k$.

A \emph{strategy} $\polZ$ of \PZ in the game graph $\G$ maps every finite path
ending at an \PZ-owned vertex $\v\in \VZ$ to one of the successors of $\v$;
formally, $\polZ\colon\V^*\VZ\cap \pathsfin(\G)\to \V$ with the constraint $\polZ(\v^0\ldots \v^k)\in \E(\v^k)$.
The strategy $\polZ$ is called \textit{memoryless} if it only depends on the
current vertex instead of the path seen so far, i.e., for every pair of finite
paths $\pi v,\pi' v\in \pathsfin(\G)$ ending at the same vertex $v$, $\polZ(\rho
v)=\polZ(\rho'v)$. 
%; in such cases we simply write $\polZ(v)$ to denote the proposed action at $v$.
The strategy $\polZ$ is called a \textit{finite}-memory strategy if it uses a
finite automaton to store $\polZ$.
Formally, for a given game graph $\G=(\V,\VZ,\VO,\E)$, a finite-memory strategy
for \PZ can be represented as a tuple $\Gamma=(Q,q^0,F,H)$, where $Q$ is a
finite set of memory states, $q^0$ is the initial state, $F\colon Q\times \V\to
Q$ is the memory update function, and $H\colon Q\times \VZ\to V$ is the
output function.
Every finite path $\rho=\v^0\ldots\v^t\in\pathsfin(\G)$ in $\G$ induces a
sequence of memory states $q^0\ldots q^t$ as follows: for every $i\in [0;t)$,
$q^{i+1} = F(q^i,\v^i)$.
Then, given $\rho$ and if $\v^t\in \VZ$, we obtain $\polZ(\rho) = H(q^t,v^t)$.
We will say \textit{the memory size of $\polZ$ is $|Q|$}.
An \textit{infinite}-memory strategy is one whose memory cannot be encoded using
a finite automaton as described above.
We will write $\PolZ$, $\PolZM$, and $\PolZFM$ to represent the sets of all,
memoryless, and finite-memory strategies of \PZ; clearly, $\PolZM\subset
\PolZFM\subset \PolZ$, and the set of infinite-memory strategies of \PZ is $\PolZ\setminus \PolZFM$.
Analogously, the strategy of \PO in $\G$ is a function $\polO\colon \V^*\VO\cap
\pathsfin(\G)\to \V$; the notions of memoryless, finite-memory, and
infinite-memory strategies extend naturally.

A path $\v^0\v^1\ldots\in \V^\infty$ is said to be \emph{compliant} with a given
strategy $\polZ$ of \PZ if for every $i\geq 0$ with $\v^i\in \VZ$, $\v^{i+1} = \polZ(\v^0\ldots \v^i)$; compliance of paths with strategies of \PO is defined analogously. 
We will write $\langZ(\G,\polZ,\v^0)$ and $\langO(\G,\polO,\v^0)$ to denote the sets of infinite paths in the game graph $\G$ that start at the vertex $\v^0$ and are compliant with the strategy $\polZ$ of \PZ and the strategy $\polO$ of \PO, respectively.
On the other hand, $\pathZO(\G,\polZ,\polO,\v^0)$ will denote the unique
infinite path compliant with both $\polZ$ and $\polO$ and starting at $\v^0$.
We will drop $\G$ as an argument of $\langZ,\langO$, and $\pathZO$ if it is
clear from the context.

Given a finite or infinite path $\rho$, we will say $u$ is a prefix of $\rho$ if
there exists a path $w$ such that $\rho = uw$, and $\pref(\rho)$ denotes the
set of all prefixes of $\rho$.
This is extended to sets of paths in the natural manner: for a given set $L$ of
paths, $\pref(L) \coloneqq \cup_{\rho\in L} \pref(\rho)$.

\smallskip
\noindent\textbf{Specifications.}
A \emph{specification} $\spec$ is a set of infinite paths in $\G$, and we will write $\compl{\spec}$ to denote the complement of $\spec$ in $\G$, i.e., $\compl{\spec} = \pathsinf(\G)\setminus \spec$.
We will consider the following types of specifications:
\begin{description}
	\item[Reachability.] Given $T\subseteq \V$, called the \emph{target vertices}, the reachability specification $\reach(T)$ is the set of all infinite paths that visit $T$ at least once;\\ formally, $\reach(T)\coloneqq \set{\v^0\v^1\ldots\in\pathsinf(\G)\mid \exists i \geq 0\;.\;\v^i\in T}$.
	\item[Safety.] Given $S\subseteq \V$, called the \emph{safe vertices}, the safety specification $\safe(S)$ is the set of all infinite paths that stay in $S$ forever; formally, $\safe(S)\coloneqq\set{\v^0\v^1\ldots\in\pathsinf(\G)\mid \forall i \geq 0\;.\;\v^i\in S}$.
	\item[B\"uchi.] Given $B\subseteq \V$, called the \emph{B\"uchi vertices}, the B\"uchi specification $\buchi(B)$ is the set of all infinite paths that visit $B$ infinitely many times;\\ formally, $\buchi(B)\coloneqq \set{\v^0\v^1\ldots\in\pathsinf(\G)\mid \forall i\geq 0\;.\;\exists j \geq i\;.\; \v^j\in B}$.
	\item[Co-B\"uchi.] Given $B\subseteq \V$, called the \emph{co-B\"uchi
	vertices}, the co-B\"uchi specification $\cobuchi(B)$ is the set of all
	infinite paths that visit $B$ only finitely many times; formally,
	$\cobuchi(B)\coloneqq \set{\v^0\v^1\ldots\in\pathsinf(\G)\mid \exists i\geq
	0\;.\;\forall j \geq i\;.\; \v^j\notin B}$.
	\item[Generalized B\"uchi.] Given $\Bset = \set{B_1,\ldots,B_p}$ with $B_i\subseteq \V$ for each $i\in [1;p]$, the generalized B\"uchi specification $\gbuchi(\Bset)$ is the set of all infinite paths that visit each $B_i$ set infinitely many times, i.e., $\gbuchi(\Bset)\coloneqq \bigcap_{i\in [1;p]} \buchi(B_i)$.
	\item[Generalized Rabin (single pair).] Given a pair $\left( \Bset =
	\set{B_1,\ldots,B_p},\mathbf{C}=\set{C_1,\ldots,C_q} \right)$, where
	$B_i\subseteq \V$ for each $i\in [1;p]$ and $C_j\subseteq \V$ for each $j\in
	[1;q]$, the (single pair) generalized Rabin specification
	$\grabin(\Bset,\mathbf{C})$ is the set of all infinite paths that visit
	every $B_i\in \Bset$ infinitely often and every $C_j$ only finitely often,
	i.e., $\grabin(\Bset,\mathbf{C})\coloneqq\gbuchi(\Bset)\cap \cap_{j\in [1;q]}\cobuchi(C_j)$.%
	\footnote{Our (single pair) generalized Rabin specification is actually a
	\textit{single generalized Rabin pair} like in the work of Chatterjee et
	al.~\cite{chatterjee2013automata}, which considered specifications
	consisting of disjunctions of multiple such generalized Rabin pairs.
	This should not be confused with the generalized Rabin condition of
	Ehlers~\cite{ehlers2011generalized}, which considers a set of assumptions
	and guarantees where each assumption and guarantee consists of a single
	\textit{regular} Rabin pair.}
	Since $\cap_{j\in [1;q]}\cobuchi(C_j) = \cobuchi(\cup_{j\in [1;q]} C_j)$, we
	could instead define $\grabin(\cdot,\cdot)$ to contain a single set of
	co-B\"uchi vertices in the second component. 
	While this would be semantically the same, it would slightly complicate the notation afterwards.
	\item[GR(1).] Given $\Aset = \set{A_1,\ldots,A_m}$ with $A_i\subseteq\V$ for each $i\in [1;m]$, called the (environment) \emph{assumptions}, and $\Gset = \set{G_1,\ldots,G_n}$ with $G_j\subseteq \V$ for each $j\in [1;n]$, called the (system) \emph{guarantees}, the GR(1) specification $\gr{\Aset}{\Gset}$ is the set of all infinite paths that either visits \emph{some} $A_i$ only \emph{finitely} often, or visits \emph{all} $G_j$ \emph{infinitely} often. 
	Formally, $
		\gr{\Aset}{\Gset}\coloneqq  \compl{\gbuchi(\Aset)} \cup \gbuchi(\Gset)
	$.
	By ``\PO meeting the assumptions'' or ``\PZ fulfilling the guarantees,''
	we mean $\gbuchi(\Aset)$ or $\gbuchi(\Gset)$ being satisfied, respectively.
	Simply put, GR(1) requires \PZ to fulfill all
	guarantees if \PO meets all assumptions.
\end{description}
Generally, a specification $\spec$ in $\G$ is called a \textit{liveness} specification if
every finite path in $\G$ can be extended to an infinite path belonging to
$\spec$; formally,  $\pref(\spec) = \pathsfin(\G)$.

Of special interest to us are paths that fulfill a partial set of assumptions and a partial set of guarantees, for a given GR(1) specification $\gr{\Aset}{\Gset}$.
As a convention, we will write $\Asetinf$ to denote the largest set of assumptions that are fulfilled by a given path $\rho$, while the complement $\Aset\setminus\Asetinf$ contains those that are violated; in other words, $\rho\in \grabin(\Asetinf,\Aset\setminus\Asetinf)$.
Similarly, we will write $\Gsetinf\subseteq \Gset$ to denote the largest set of guarantees that are fulfilled.

\smallskip
\noindent\textbf{Two-player games.}
A \textit{two-player game} is a tuple $(\G,\spec)$, where $\G$ is a game graph,
and $\spec$ is a specification.
In this game, a strategy $\polZ$ of \PZ is called a \emph{winning strategy} of
\PZ from a given initial vertex $\v\in \V$, written as $(\polZ,\v)\models
\spec$, if for every strategy of \PO, the generated infinite path belongs to $\spec$; in
other words, $\langZ(\polZ,\v)\subseteq \spec$.
The \emph{winning region} of \PZ is the set of all vertices from which she has a
winning strategy.
It is known that \PZ has a uniform winning strategy that is
winning from all vertices in the winning
region~\cite{zielonka1998infinite,piterman2006faster}; unless otherwise
mentioned, whenever we use winning strategies, they will refer to uniform
winning strategies.

\smallskip
\noindent\textbf{The obliging variant of two-player games.}
We consider a second variant of games, known as obliging games.
An \textit{obliging game}~\cite{chatterjee2010obliging} is a triple
$(\G,\spec,\psi)$, where $\G$ is a game graph, and $\spec$ and $\psi$ are two specifications, referred to as the \text{strong} and \text{weak} specifications, respectively.
\PZ's objective is to guarantee the fulfillment of $\spec$, while not
``obstructing'' \PO from fulfilling $\psi$.
Formally, in the obliging game, a strategy $\polZ$ of \PZ is called a \textit{winning
strategy} from the initial vertex $v^0\in \V$, if 
\begin{enumerate}[(i)]
	\item $\polZ$ is a winning strategy of \PZ in the ordinary game $(\G,\spec)$
	from the initial vertex $v^0$, and
	\item for every finite path $v^0\ldots v^k$ compliant with $\polZ$, there
	exists a strategy $\polO$ of \PO such that the infinite path starting at $v^k$
	and compliant with $\polZ$ and $\polO$ is in $\psi$, i.e., $\pref\left(\langZ(\polZ,\v^0)\right)\subseteq \pref\left( \langZ(\polZ,\v^0) \cap \psi \right)$.
\end{enumerate}
The concept of winning region is analogous to ordinary games. 

We present a basic fact about winning strategies in obliging games (can be
adapted for ordinary games too) that will be used later.
\begin{proposition}[Prefix independence of winning strategies]
	\label{prop:prefix-independence of graceful strategies}
	Let $(\G,\spec,\psi)$ be an obliging game, where $\spec$ and $\psi$ are
	liveness specifications.
	Suppose $\polZ$ is an arbitrary strategy of \PZ, and $\polZ'$ is a winning
	strategy of \PZ in $(\G,\spec,\psi)$.
	Let $t>0$ be a time step.
	Define the composite strategy $\polZ^+$ of \PZ that uses $\polZ$ for
	the first $t$ steps, and then switches to $\polZ'$; formally, for every
	$\rho = v^0\ldots v^k\in \pathsfin(\G)$ with $v^k\in \VZ$, if $k <t$ then
	$v^{k+1} = \polZ(\rho)$, else $v^{k+1} = \polZ'(\rho)$.
	If for every strategy of \PO, the resulting vertex $v^t$ reached at time $t$
	is in the winning region, then $\polZ^+$ is also winning.
\end{proposition}

\begin{proof}
	As $\spec$ and $\psi$ are both liveness specifications, every finite
	path in $\G$, including paths after $t$ time steps can be extended into
	paths that are in $\spec$ and $\psi$.
	Moreover, since the vertices $v^t$ obtained by applying arbitrary \PO
	strategies are in the winning region, it immediately follows (from
	the definition of winning region) that $\polZ^+$ is winning from the starting vertex $v^0$.
\end{proof}

% Although Proposition~\ref{prop:prefix-independence of graceful strategies} is a
% simple observation, it serves as the enabler of our adaptive strategy framework,
% where multiple different graceful strategies will be combined via random
% interleaving, and Proposition~\ref{}
%!TEX root=main.tex

\section{Gracefully Robust Strategies for GR(1) Specifications}
\label{sec:gracefully robust strategies}

Standard winning strategies for GR(1) games suffer from two limitations:
(A)~\PZ can actively violate assumptions, and
(B)~\PZ may abandon all guarantees if any assumption is violated.
Prior work addresses these separately: graceful strategies~\cite{DBLP:conf/tacas/MajumdarPS19,chatterjee2010obliging} address~(A),
robust strategies~\cite{DBLP:conf/cav/BloemCGHJ10} address~(B), but neither addresses both.

We formalize \emph{gracefully robust} strategies that address both limitations (Section~\ref{sec:formalizing graceful robustness}), show that their static computation is hard (Section~\ref{sec:hardness of graceful robustness}), and propose practical relaxations enabling the dynamic computation developed in the rest of the paper (Section~\ref{sec:practical relaxations}).

\subsection{Formalizing Gracefully Robust Strategies}
\label{sec:formalizing graceful robustness}

A strategy of \PZ is \emph{graceful} if it wins the GR(1) specification while allowing \PO to visit all assumptions infinitely often.
We formalize this using obliging games~\cite{chatterjee2010obliging} (introduced in Section~\ref{sec:prelims}), presenting a general definition where \PZ's objective is a fragment of the original specification.

\begin{definition}[Graceful strategies in GR(1) games]\label{def:graceful strategies}
	Consider the GR(1) game $(\G,\gr{\Aset}{\Gset})$ where $\G=(\V,\VZ,\VO,\E)$,
	let the initial vertex be $\v\in \V$, and let $\gr{\Aset'}{\Gset'}$ be a
	sub-specification of the original specification $\gr{\Aset}{\Gset}$ such
	that $\Aset'\subseteq \Aset$ and $\Gset'\subseteq \Gset$.
	A strategy $\polZ$ of \PZ that is winning for the sub-specification
	$\gr{\Aset'}{\Gset'}$ is called \emph{graceful} from $\v$, with respect to 
	$\Aset$, if it is winning for the obliging game $(\G,\gr{\Aset'}{\Gset'},\gbuchi(\Aset))$ from the initial vertex $\v$.
\end{definition}
When $\Aset'=\Aset$ and $\Gset'=\Gset$, our graceful strategies coincide with the \textit{environmentally friendly} strategies of Majumdar et al.~\cite{DBLP:conf/tacas/MajumdarPS19}.
The more general setting is needed for our adaptive strategies.

\begin{example}[Graceful strategies of \PZ]\label{ex:gracefulness}
	Consider the game in Figure~\ref{fig:GR(1) game:illustrating best-effort and
	limit-best-effort} with $\Aset'=\Aset=\set{A_1,A_2}$ and $\Gset'=\Gset=\set{G_1,G_2}$.
	The strategy $\pol_{\eve,a}$ choosing $0\to 6$ and eventually $8\to 10$ is
	\textit{graceless}: once vertex $10$ is reached, \PO can no longer fulfill
	$\buchi(A_1)$.

	The strategy $\pol_{\eve,b}$ choosing $0\to 6$ and $8\to 9$ is graceful: \PO can fulfill $\gbuchi(\set{A_1,A_2})$ by visiting vertices $7$ and $9$ infinitely often.

	Both $\pol_{\eve,a}$ and $\pol_{\eve,b}$ are GR(1) winning:
	for $\pol_{\eve,a}$, the game either stays at vertex $7$ or reaches $10$, violating an assumption in both cases;
	for $\pol_{\eve,b}$, if \PO visits $7\to 8$ and $9\to 6$ infinitely often, both assumptions and guarantees are met, otherwise some assumption is violated.

	Finally, a strategy eventually choosing $8\to 12$ is not GR(1) winning (both assumptions met but guarantees violated), hence not graceful (according to Definition~\ref{def:graceful strategies}) despite enabling \PO to fulfill all assumptions.
	\qed
\end{example}
	%!TEX root=../main.tex
% Terminally gracefully robust schematic diagram

\begin{wrapfigure}{r}{0.5\textwidth}
	\centering
	\tikzset{every state/.style={minimum size=20pt}}
	\begin{tikzpicture}[
		scale=0.7,
		transform shape,
		state/.style={circle, draw, minimum size=0.5cm, inner sep=1pt, font=\scriptsize},
		rectstate/.style={rectangle, draw, minimum size=0.5cm, inner sep=1pt, font=\scriptsize},
		every edge/.style={draw, ->, >=stealth, thin},
		node distance=0.8cm]
		\node[state,initial above]	(a)	at	(0,0)	{$0$};
		\node[state,rectangle]	(b)	[right=of a]	{$1$};
		\node[state,rectangle,label={left:{\color{red}$A_1$}}]	(c)	[above=of b]	{$2$};
		\node[state,rectangle,label={left:{\color{red}$A_2$}}]	(d)	[below=of b]	{$4$};
		\node[state,rectangle,label={right:{\color{blue}$G_1$}}]	(e)	[right=of c]	{$3$};
		\node[state,rectangle,label={right:{\color{blue}$G_2$}}]	(f)	[right=of d]	{$5$};
		
		\node[state,rectangle,label={above:{\color{blue}$G_1$}}]	(g)	[left=of a]		{$6$};
		\node[state,label={below:{\color{blue}$G_2$}}] (g1) [left=of g] {$8$};
         
		% \node[state,rectangle,label={right:{\color{blue}$G_1$}}]	(h)	[above=of g]	{$7$};
		\node[state,rectangle,label={above:{\color{red}$A_1$}}]	(i)	[above=of g1]		{$7$};
		\node[state,rectangle,label={above right:{\color{red}$A_2$}}]	(j)	[below=of g]	{$9$};
		% \node[state,rectangle, label={left:{\color{blue}$G_2$}}]	(k)	[left=of j]		{$9$};
		\node[state,rectangle,label={below:{\color{blue}$G_2$}}] (g2) [left=of g1] {$10$};
          \node[state,rectangle, label={above :{\color{red}$A_2$}}] (g3)
          [left=of i] {$11$};
          \node[state, rectangle,label={above:{\color{red}$A_1,A_2$}}] (k)
          [below=of g2] {$12$};

		\path[->]	(a)		edge	(b)
							edge	(g)
					(b)		edge	(c)
							edge	(d)
							% edge[loop right]	()
					(c)		edge	(e)
					(d)		edge	(f)
					(e)		edge	(b)
					(f)		edge	(b)
					(g)		edge	(i)
							% edge[loop left]	()
					% (h)		edge	(i)
					(i)		edge	[bend right=0](g1)
					         edge[loop right]	()
					(g1)		edge	(j)
					         edge (g2)
							 edge	(k)
					  (g2) edge[bend left](g3)
					  (g3) edge[bend left](g2)  
					  		edge[loop right]	()     
							(j)		edge	(g)
							edge[loop right]	()
					 (k)		edge[loop right]	();
	\end{tikzpicture}
	\caption{A game graph with the GR(1) specification $\gr{\set{A_1,A_2}}{\set{G_1,G_2}}$. The circle vertices belong to \PZ and the square ones belong to \PO.}
	\label{fig:GR(1) game:illustrating best-effort and limit-best-effort}
\end{wrapfigure}
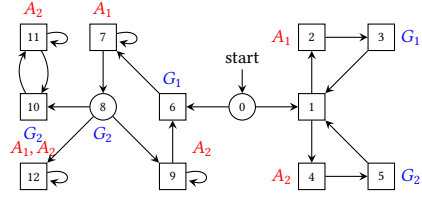

We now define robustness in terms of fulfilled guarantees.
For finite sets $S$ and $S'$, write $S \sqsubseteq S'$ if $|S|\leq |S'|$, and use the strict variant ``$\sqsubset$'' and inverse variants ``$\sqsupseteq$'' and ``$\sqsupset$'' accordingly.
This induces a total order over subsets of $\Aset$ and $\Gset$, yielding a preference order over \PZ's strategies.

Given initial vertex $\v^0$, strategies $\pol_{\eve,1}$ and $\pol_{\eve,2}$ of \PZ, and a strategy $\polO$ of \PO:
$\pol_{\eve,1}$ is at least as good as $\pol_{\eve,2}$ if it meets at \textit{most} as many assumptions but at \textit{least} as many guarantees; strictly better if one inequality is strict.

For $i\in \set{1,2}$, let $\rho_i = \pathZO(\pol_{\eve,i},\polO,\vinit)$, and let $\Aset_{\mathrm{inf},i}$ and $\Gset_{\mathrm{inf},i}$ be sets of assumptions and guarantees fulfilled by $\rho_{i}$, i.e., $\rho_{i}\in \grabin(\Aset_{\mathrm{inf},i},\Aset \setminus
\Aset_{\mathrm{inf},i})$ and $\rho_i \in \grabin(\Gset_{\mathrm{inf},i},\Gset\setminus\Gset_{\mathrm{inf},i})$.
\begin{description}
	\item[Weak comparison:] $(\pol_{\eve,1},\polO) \bettereq
	(\pol_{\eve,2},\polO)$ if $\Aset_{\mathrm{inf},1} \sqsubseteq
	\Aset_{\mathrm{inf},2}$ and $\Gset_{\mathrm{inf},1} \sqsupseteq \Gset_{\mathrm{inf},2}$.
	\item[Strong comparison:] $(\pol_{\eve,1},\polO) \better
	(\pol_{\eve,2},\polO)$ if either $\Aset_{\mathrm{inf},1} \sqsubset
	\Aset_{\mathrm{inf},2}$ and $\Gset_{\mathrm{inf},1} \sqsupseteq
	\Gset_{\mathrm{inf},2}$, or $\Aset_{\mathrm{inf},1} \sqsubseteq  \Aset_{\mathrm{inf},2}$ and $\Gset_{\mathrm{inf},1} \sqsupset \Gset_{\mathrm{inf},2}$.
\end{description}
This compares \PZ's strategies for a fixed \PO strategy; we extend this by quantifying over \PO.
\begin{definition}[Dominated strategies]\label{def:dominated strategies}
	Consider the GR(1) game $(\G,\gr{\Aset}{\Gset})$.
	Suppose $\polZ$ and $\polZ'$ are two strategies of \PZ.
	From a given initial vertex $\v^0$, we will say $\polZ$ dominates $\polZ'$ if:
	\begin{itemize}
		\item $\forall$ strategies $\polO$ of \PO, $\polZ$ is at least as good as
		$\polZ'$, i.e., $(\polZ,\polO)\bettereq (\polZ',\polO)$, and
		\item $\exists$ a strategy $\polO$ of \PO, for which $\polZ$ is
		strictly better than $\polZ'$, i.e., $(\polZ,\polO) \better (\polZ',\polO)$.
	\end{itemize}
\end{definition}
Finally, we are ready to introduce \emph{gracefully robust} strategies in GR(1) games:
\begin{definition}[Gracefully robust strategies]\label{def:globally gracefully robust}
	Consider the GR(1) game $(\G,\gr{\Aset}{\Gset})$, and let $\vinit$ be a
	given initial vertex.
	A strategy $\polZ$ of \PZ is a gracefully robust strategy from $\vinit$ if:
	\begin{enumerate}[(A)]
		\item $\polZ$ is a \textit{graceful} strategy from $\vinit$, and
		\label{def:globally gracefully robust:gracefulness}
		\item $\polZ$ is not dominated by any other graceful
		strategy from $\vinit$.\label{def:globally gracefully robust:non-dominance}
	\end{enumerate}
\end{definition}
We illustrate graceful robustness, continuing Example~\ref{ex:gracefulness}.

\begin{example}[Graceful robustness]\label{ex:graceful robustness}
	Consider again Example~\ref{ex:gracefulness}.
	Let $\pol_{\eve,c}$ choose edge $0\to 1$ (the choice at vertex $8$ is irrelevant).
	This strategy is graceful since \PO can fulfill $\gbuchi(\set{A_1,A_2})$ by visiting vertices $2$ and $4$ infinitely often.
	To see that $\pol_{\eve,c}$ dominates $\pol_{\eve,b}$, consider \PO strategy $\polO'$ selecting $7\to 7$ and $1\to 2$.
	Both $\pathZO(0,\pol_{\eve,c},\polO')$ and $\pathZO(0,\pol_{\eve,b},\polO')$ fulfill only $\buchi(A_1)$, but the latter fulfills no guarantees while the former fulfills $\buchi(G_1)$.
	Thus $(\pol_{\eve,c},\polO')\better (\pol_{\eve,b},\polO')$, and one can verify $(\pol_{\eve,c},\polO)\bettereq (\pol_{\eve,b},\polO)$ for all $\polO$, so $\pol_{\eve,c}$ dominates $\pol_{\eve,b}$.
	Since $\pol_{\eve,c}$ is not dominated by any other strategy, it is gracefully robust.
	\qed
\end{example}

Graceful strategies from the literature~\cite{DBLP:conf/tacas/MajumdarPS19,chatterjee2010obliging} implement only~\eqref{def:globally gracefully robust:gracefulness}
from Definition~\ref{def:globally gracefully robust}
\footnote{Majumdar et al.~\cite{DBLP:conf/tacas/MajumdarPS19} use the term \emph{environmentally friendly}, equivalent to our graceful strategies. The term ``graceful'' is from Chatterjee et al.~\cite{chatterjee2010obliging}, who considered gracefulness beyond GR(1).}
while robust strategies~\cite{DBLP:conf/cav/BloemCGHJ10} implement only~\eqref{def:globally gracefully robust:non-dominance}; Definition~\ref{def:globally gracefully robust} combines both.

\begin{problem}[Computing gracefully robust
strategies]\label{prob:globally gracefully robust}
	Given the GR(1) game $(\G,\gr{\Aset}{\Gset})$, compute a gracefully
	robust strategy for \PZ.
\end{problem}

To compute (non-graceful) robust strategies, Bloem et al.~\cite{DBLP:conf/cav/BloemCGHJ10} encode robustness as a single LTL formula $\spec^{\mathsf{robust}}$ such that every winning strategy for $(\G,\spec^{\mathsf{robust}})$ is robust for $(\G,\gr{\Aset}{\Gset})$.
A na\"ive approach to graceful robustness is to solve the obliging game $(\G,\spec^{\mathsf{robust}},\gbuchi(\Aset))$.
We show this has prohibitive complexity.

\subsection{Synthesizing Gracefully Robust Strategies: The Blowup in the Na\"ive Approach}
\label{sec:hardness of graceful robustness}

% The na\"ive approach for computing globally gracefully robust strategies is as
% follows. 
% First, use the technique of Bloem et al.~\cite{DBLP:conf/cav/BloemCGHJ10} to
% transform the given GR(1) specification into a
We first summarize the complexity of computing (non-graceful) robust strategies~\cite{DBLP:conf/cav/BloemCGHJ10}.

\begin{restatable}[Upper complexity bound for finding robust
strategies~\cite{DBLP:conf/cav/BloemCGHJ10}]{theorem}{thmComplexityRobustStrategies}\label{thm:complexity of robust strategies}
	Let $\G = (\V,\VZ,\VO,\E)$ be a game graph, $\v^0$ be the initial vertex,
	$\gr{\Aset}{\Gset}$ be the GR(1) specification that has $m$ assumptions and $n$
	guarantees, and let $k= \min(m,n)$.
	A robust strategy of \PZ can be computed in time
	$\mathcal{O}\big(2^{m+n}\cdot |\V|^k \cdot |\E|\cdot (Y\,Z)^{k(k+1)}\cdot k!\big)$,
	where $Y = 2^{\mathcal{O}(m\,2^m)}$ and $Z = 2^{\mathcal{O}(n\,2^n)}$; in
	particular the bound is \emph{doubly} exponential in $m$ and $n$.
The memory size of the resulting policy is $\big((Y+1)(Z+1)\big)^{k}\cdot k!\cdot |\V|$.
\end{restatable}

\begin{proof}[Proof sketch, details in Appendix~\ref{appendix:proofs of the problem statement section}]
	Bloem et al.\ encode robustness using ``multiple counting requirements,'' each a set $s = \{(a_1,b_1),\ldots,(a_l,b_l)\}$ of threshold pairs stating: if at least $a_i$ assumptions hold then at least $b_i$ guarantees must hold.
	This conjunction of $l$ implications is exactly a generalized reactivity specification of rank $l$, i.e., a game of the form $\bigwedge_{p\in[1;l]}(\psi_p \Rightarrow \varphi_p)$.
	The antecedent ``at least $a_p$ of $m$ assumptions'' is a DCB (disjunction of conjunctions of Büchi), and converting each DCB to generalized Büchi via~\cite[Corollary~1]{DBLP:conf/cav/BloemCGHJ10} yields $Y = 2^{\mathcal{O}(m\,2^m)}$ conjuncts (symmetrically $Z = 2^{\mathcal{O}(n\,2^n)}$ for consequents).
	To find (non-graceful) robust strategies, they iterate over each of the $\mathcal{O}(2^{m+n})$ requirements~\cite[Theorem~7]{DBLP:conf/cav/BloemCGHJ10}, solve its generalized reactivity game (rank bounded by $l \le k$), and take the best solution; solving each game costs $\mathcal{O}(|\V|^k|\E|(YZ)^{k(k+1)}k!)$ time, giving the stated bound.
\end{proof}

While the lower complexity bound for the problem is open, Bloem et al.\ prove that
games with DCB objectives are coNP-complete and argue that the exponential blowup
in the DCB-to-generalized-B\"uchi translation is ``probably
inevitable''~\cite[Section~4]{DBLP:conf/cav/BloemCGHJ10}.
Recall that non-graceful robust strategies considered in
Theorem~\ref{thm:complexity of robust strategies} above only implements the
requirement~\eqref{def:globally gracefully robust:non-dominance} of
Definition~\ref{def:globally gracefully robust}.
Now we add the overhead of including the requirement~\eqref{def:globally
gracefully robust:gracefulness} in the following corollary.

\begin{restatable}{corollary}{corGracefulComplexity}\label{cor:graceful complexity}
	Under the hypotheses and notation of Theorem~\ref{thm:complexity of robust
	strategies}, a gracefully robust strategy of \PZ can be computed in time
	\[
		\mathcal{O}\big(2^{m+n}\cdot |\E|\cdot |\V|^{\,k+m}\cdot
		Y^{\,k(k+m+1)}\cdot Z^{\,k(k+m+1)}\cdot (k+m)!\big),
	\]
	and the memory size of the resulting strategy will be
	$\mathcal{O}(Y^kZ^k(k+m)!)$, where $Y=2^{\mathcal{O}(m\,2^m)}$ and $Z=2^{\mathcal{O}(n\,2^n)}$; both bounds are
	doubly exponential in $m$ and $n$.
\end{restatable}

\begin{proof}[Proof sketch, details in Appendix~\ref{appendix:proofs of the problem statement section}]
	For gracefulness, we solve an obliging game $(\G,\spec,\gbuchi(\Aset))$ for each requirement's specification $\spec$.
	Each requirement's generalized reactivity game first reduces~\cite[proof of Theorem~3]{DBLP:conf/cav/BloemCGHJ10} to a Streett game with $\mathcal{O}(|\V|Y^l Z^l)$ vertices, $\mathcal{O}(|\E|Y^l Z^l)$ edges, and $l$ Streett pairs.
	By Chatterjee et al.~\cite[Lemma~2, Theorem~6]{chatterjee2010obliging}, combining an $l$-pair Streett game with an $m$-generalized Büchi weak objective yields a regular Streett game with $(l+m)$ pairs, with constant blowup in the vertex and edge set; the time and memory complexities then follow from known results on Streett games~\cite{piterman2006faster}.
\end{proof}

\subsection{Practical Relaxations}\label{sec:practical relaxations}

Corollary~\ref{cor:graceful complexity} shows that computing gracefully robust strategies is prohibitively expensive, and Bloem et al.~\cite{DBLP:conf/cav/BloemCGHJ10} provide evidence that improvement is unlikely.
We therefore consider two practical relaxations that form the basis of our efficient \textit{adaptive} strategies in Sections~\ref{sec:static design phase} and~\ref{sec:adaptive strategy}.

\subsubsection{\textit{Terminal} Graceful Robustness in Limit Subgames}
Our first relaxation allows \PZ to be dominated in the transient part of the game (vertices visited only finitely often): ``mistakes'' in the transient part cannot be reverted and are forgiven.
However, we require robustness in every subgame where the game settles
long-term.
We formalize this below.

Given an infinite path $\rho=\v^0\v^1\ldots$, let $\recurring(\rho) \coloneqq \set{u\in \V\mid\forall i\geq 0\;.\;\exists j\geq i\;.\;\v^j = u}$ be the vertices visited infinitely often.
For $U\subseteq \V$, define $\G[U] \coloneqq (\V\cap U,\VZ\cap U,\VO\cap U, \E\cap (U\times U))$; this is a \emph{subgame} if no vertex in $\V\cap U$ is a dead-end.
The \emph{reachable vertices} from $U$ are those $\v\in \V$ reachable by some path from $U$.

\begin{definition}[Limit subgames]
	For a given strategy $\polZ$ of \PZ and a strategy $\polO$ of \PO, the \emph{limit subgame} of $\polZ$ and $\polO$ is the subgame induced by the set of vertices reachable from the set $\recurring(\pathZO(\polZ,\polO,\vinit))$.
\end{definition}

Limit subgames are well-defined: since $\G$ has no dead-ends, neither does any reachable subset.
Vertices outside all limit subgames are transient and unreachable in the long run.
If $\G$ is strongly connected, every limit subgame is the entire graph.

\begin{definition}[Terminally gracefully robust strategies]\label{def:terminally
gracefully robust strategies}
	A strategy $\polZ$ of \PZ is called \emph{terminally} gracefully robust if the
	following hold:
	\begin{enumerate}[(a)]
		\item the strategy $\polZ$ is winning for the original GR(1)
		specification and is graceful with respect to the entire set of assumptions $\Aset$, i.e., it is winning for the obliging game $(\G,\gr{\Aset}{\Gset},\Aset)$; \label{item:def:terminally gracefully robust:gracefulness}
		\item for every strategy $\polO$ of \PO, if $\G[U]$ is the
		limit subgame of $\polZ$ and $\polO$, then within $\G[U]$, $\polZ$ is
		not dominated by any other graceful strategy of \PZ. \label{item:def:terminally gracefully robust:robustness}
	\end{enumerate}
\end{definition}
\begin{remark}\label{rem:dominated within subgame}
	In item~\eqref{item:def:terminally gracefully robust:robustness} above,
	\emph{graceful} and \emph{dominated} refer to the subgame $\G[U]$, using
	the same assumption and guarantee sets $\Aset,\Gset$ (restricted implicitly
	to $U$, since any play confined to $U$ can only ever visit $\Aset\cap U$
	and $\Gset\cap U$ anyway), and \emph{not dominated} means not dominated
	\emph{from any vertex of $U$}: since the transient history determines
	exactly which vertex of $U$ is first entered, we require $\polZ$ to be
	undominated uniformly across all of $U$, not merely from one particular
	entry point.
\end{remark}

In Example~\ref{ex:graceful robustness}, $\pol_{\eve,c}$ dominates $\pol_{\eve,b}$, so $\pol_{\eve,b}$ is not robust.
But if \PZ ``mistakenly'' chooses $0\to 6$, this cannot be undone, and we end up
in the limit subgame that is the set of all vertices to the left of $0$.
We argue that in this limit subgame, $\pol_{\eve,b}$ is optimal, which would
imply that $\pol_{\eve,b}$ is \textit{terminally} gracefully robust, though not
gracefully robust.
To see why $\pol_{\eve,b}$ is optimal in the left half of the game, notice that
the alternative strategies of \PZ in this subgame are to eventually pick either
$8\to 10$ or $8\to 11$. 
We already established (in Example~\ref{ex:gracefulness}) that picking $8\to 10$
is not graceful, since \PO can never fulfill $A_2$.
On the other hand, picking $8\to 11$ would lead to violating the GR(1)
specification, since both $A_1$ and $A_2$ will be met but none of the
guarantees will be fulfilled: therefore, this option is not GR(1) winning and
therefore is also not graceful.
We conclude that $\pol_{\eve,b}$ is not dominated by any other \PZ strategy in
this limit subgame.

The following claim shows that graceful robustness is stronger than terminally gracefully robust.

\begin{restatable}{proposition}{propGlobalImpliesTerminal}\label{prop:global implies terminal}
	Every gracefully robust strategy is terminally gracefully robust.
\label{prop:grace}
\end{restatable}

\begin{proof}[Proof sketch, details in Appendix~\ref{appendix:proofs of the problem statement section}]
	Suppose $\polZ$ is gracefully robust but not terminally so: then some graceful $\polZ''$ dominates $\polZ$ within a limit subgame $\G[U]$.
	Since $U$ is closed under edges (once entered, plays cannot leave), we construct a composite strategy $\polZ'$ that follows $\polZ$ until entering $U$, then switches to $\polZ''$.
	The strategy $\polZ'$ is graceful (by case analysis on whether $U$ is entered) and dominates $\polZ$ globally (since the $\bettereq$ comparisons depend only on infinitely visited vertices, not finite prefixes), contradicting that $\polZ$ is gracefully robust.
\end{proof}

Intuitively, terminally gracefully robust strategies start optimistically, assuming \textit{all} assumptions will be fulfilled.
As the game narrows to a specific limit subgame and some assumptions fail, the strategy takes undominated actions in the remainder.
This is more reasonable for fault-tolerant systems than existing approaches~\cite{DBLP:conf/cav/BloemCGHJ10} that assume worst-case environment behavior from the start.
We formalize this non-adversarial view in our second relaxation.

\begin{problem}[Computing terminally gracefully robust strategies]\label{prob:terminally gracefully robust}
	Given the GR(1) game $(\G,\gr{\Aset}{\Gset})$, compute a terminally
	gracefully robust strategy for \PZ.
\end{problem}

\subsubsection{Non-Adversarial Environment}
Our second relaxation assumes \PO is not a pure adversary: \PO's choice of which assumptions to meet does not depend on \PZ's particular strategy (provided \PZ is graceful), and assumptions visited infinitely often appear at bounded intervals.
For instance, in Example~II (Section~\ref{sec:intro}), a \PO strategy from $\PolO^{(1)}$ alternating between $h_1$ and $t_1$ is non-adversarial, whereas strategies from $\PolO^{(2)}$ picking the opposite of \PZ's choice are adversarial (which subset of assumptions will be fulfilled depends on \PZ's actions and may have unbounded frequency).

We formalize bounded visits using hitting times.
For an infinite path $\rho = v^0v^1\ldots$ and vertex set $X$, the \textit{hitting time} of $X$ in $\rho$ is bounded by $g\colon \mathbb{N}\to \mathbb{N}$ if for every $i$ with $v^i \in X$, there exists $j \leq i+g(i)$ with $v^j\in X$.
We say $X$ has \textit{constant} hitting time if $g(\cdot)\equiv n$ for some constant $n$.

\begin{assumption}[Non-adversarial environment]\label{assump:finite-memory environment}
	Fix a graph $\G$.
	\PO uses strategies from a set $\PolO'\subset \PolO$ that fulfills the following.
	There exists $\Asetfin$ and $\Asetinf$ that partition $\Aset$ (i.e., $\Aset
	= \Asetfin\cup \Asetinf$ and $\Asetfin\cap\Asetinf = \emptyset$), and there
	exists a function $\hittingtime\colon \PolZ\times\Asetinf\to \mathbb{N}$,
	such that $\forall$ strategies $\polO\in \PolO'$, $\forall$ graceful strategies
	$\polZ$ of \PZ, and for every initial vertex $v^0$, the following hold:
	\begin{enumerate}[(i)]
		\item $\pathZO(\G,\polZ,\polO,\v^0) \in \grabin(\Asetinf,\Asetfin)$, and 
		\item every $A\in \Asetinf$ has a constant hitting time in
		$\pathZO(\G,\polZ,\polO,\v^0)$ given by $\hittingtime(\polZ,A)$.\label{item:assump:finite-memory
		environment:bounded hitting time}
	\end{enumerate}
\end{assumption}

Since \PZ cannot know in advance which subset of assumptions the environment
(\PO) will actually fulfill, Section~\ref{sec:static design phase} precomputes
an optimal response to \emph{every} possible assumption subset; 
Section~\ref{sec:adaptive strategy} then shows how our adaptive strategy
identifies, and increasingly commits to the right one of these statically
computed optimal responses as evidence accumulates online.
\section{Static Design of a Library of Individual Optimal Strategies}
\label{sec:static design phase}

During the static design phase, we build a \textit{static strategy library}
formalized as the function $\static\colon 2^{\Aset}\to \PolZFM$, mapping every
subset $\Aset'\subseteq \Aset$ of assumptions to the corresponding
(finite-memory) optimal strategy when only the assumptions in $\Aset'$ are met.
% In addition, the optimal strategies must be graceful, as well as always remain
% within the graceful GR(1) winning region.
Later in Section~\ref{sec:adaptive strategy}, our adaptive strategy will use a
randomized mixture of these statically computed strategies.

\begin{definition}[Optimal strategy library]
\label{def:optimal strategy for assumption subset}
	Let $\G = (\V,\VZ,\VO,\E)$ be a given game graph, $\gr{\Aset}{\Gset}$ be
	a given GR(1) specification, and $\Aset'\subseteq \Aset$ be a subset of assumptions.
	Suppose, $W\subseteq \V$ is the winning region for the obliging game $(\G,\gr{\Aset}{\Gset},\gbuchi(\Aset))$.
	The strategy $\polZ$ of \PZ is called \emph{optimal} against $\Aset'$ if for every initial vertex $\vinit\in W$, all of the following hold:
	\begin{enumerate}[(a)]
		\item the strategy $\polZ$ is graceful from $\vinit$ with respect to the
		entire set of assumptions $\Aset$; \label{item:def:optimal strategy:a}
		\item there exists a strategy $\polO$ of \PO such that every $A\in \Aset$ has constant hitting time in $\pathZO(\G,\polZ,\polO,\vinit)$;\label{item:def:optimal strategy:a2}
		\item the strategy $\polZ$ never leaves the set $W$, i.e., 
		$\langZ(\polZ,\vinit)\subseteq \safe(W)$; and \label{item:def:optimal strategy:b}
		\item in the subgame induced by the set of vertices reachable from
		$\vinit$, no other strategy of \PZ that fulfills (a),  (b), and (c)  can
		fulfill a larger number of guarantees against $\Aset'$ than $\polZ$.
		\label{item:def:optimal strategy:c}
	\end{enumerate}	
	The optimal strategy library is the function $\static\colon 2^{\Aset}\to
	\PolZFM$ such that for every $\Aset'\in 2^\Aset$, $\static(\Aset')$ is the optimal strategy against $\Aset'$. 	
\end{definition}
In the above definition, Conditions~\eqref{item:def:optimal strategy:a} and
\eqref{item:def:optimal strategy:a2} allow \PO to visit the assumption vertices
within bounded time intervals, making sure that 
Condition~\eqref{item:assump:finite-memory environment:bounded hitting time} of
Assumption~\ref{assump:finite-memory environment} remains achievable.
Condition~\eqref{item:def:optimal strategy:b} of Definition~\ref{def:optimal strategy for assumption subset} is necessary to prevent leaving
the original graceful GR(1) winning region, which is not automatically
guaranteed because the optimal strategies concern a different GR(1)
specification than the original one.
For instance, in the context of Example~\ref{ex:gracefulness} and
\ref{ex:graceful robustness}, $W = \set{0,\ldots,9}$. 
For the initial vertex $\vinit =8\in W$ and for $\Aset'=A_1$, i.e., when only
$A_1$ is fulfilled and $A_2$ is violated, both the edges $8\to 9$ and $8\to 12$
fulfill Conditions~\eqref{item:def:optimal strategy:a}, \eqref{item:def:optimal
strategy:a2}, and \eqref{item:def:optimal strategy:c} in
Definition~\ref{def:optimal strategy for assumption subset};
Condition~\eqref{item:def:optimal strategy:c} holds because in both cases, all
assumptions can be reached with bounded hitting time and none of the guarantees will be fulfilled for both options.
However, the option $8\to 12$ is clearly detrimental because the original GR(1)
specification will be violated.
Finally, Condition~\eqref{item:def:optimal strategy:c} in Definition~\ref{def:optimal strategy for assumption subset} ensures that the obtained
strategy actually fulfills the largest possible number of guarantees in the
remaining subgame.
In particular, for Example~\ref{ex:gracefulness} and \ref{ex:graceful
robustness}, if the initial vertex is $\v^0=0$ and $\Aset'=\set{A_1}$,
Condition~\eqref{item:def:optimal strategy:c} forces \PZ to select $0\to 1$ over
$0\to 6$, because the former would ensure one guarantee (namely $G_1$) while the
latter would ensure none.

% \todo{Add some discussion about the difference between Assumption 1 and the
% bounded visit requirement in optimal strategies, also mention which bound is used
% where in the proofs of different Lemmas.}
% \new{The straightforward use of the definition of graceful strategies from the obliging
% games framework does not impose the requirement of bounded visits of the
% assumptions by \PO.
% However, since the strategies obtained from obliging games are finite memory, it
% can be shown that this extra requirement is automatically fulfilled.\todo{add a proposition}
% We will need this for the technical parts later in the paper.}

Algorithm~\ref{alg:static computation of strategy library} describes the
computation procedure for the function $\Gamma\colon
\V\times \Aset \to \PolZFM$, which we call the \textit{intermediate} optimal
strategy library.
Intuitively, $\Gamma$ maps every vertex $v$ and every $\Aset'\subseteq \Aset$ to
an optimal policy $\Gamma(v,\Aset')$ against $\Aset'$, which is more
fine-grained than $\static$ since $\static$'s output is not allowed to depend on vertices.
The function $\Gamma(\cdot,\cdot)$ can be turned into the function $\static$ as
follows (Theorem~\ref{thm:static computation soundness} formally proves the
soundness of this step): for every $\Aset'\in 2^\Aset$ and for every $\rho \in
\pathsfin(\G')$,
if $\last(\rho)\in \VZ$, then 
\begin{equation}\label{eq:going from Gamma to static}
\static(\Aset')(\rho) \coloneqq \Gamma(\last(\rho),\Aset')(\rho).
\end{equation}
To understand Algorithm~\ref{alg:static computation of strategy library}, recall that for finite sets $S$ and $S'$, $S \sqsubseteq S' \iff |S| \leq |S'|$.
The algorithm uses the subroutine 
$\Synt(\G,\gr{\Aset'}{\Gset'},\gbuchi(\Aset))$, which returns the winning region and the winning strategy for the obliging game
$(\G,\gr{\Aset'}{\Gset'},\gbuchi(\Aset))$.
% Only for the special case $\Aset = \Aset'$, we could implement
% $\Synt(\G,\gr{\Aset}{\Gset'},\gbuchi(\Aset))$ using the environmentally friendly GR(1)
% synthesis algorithm of Majumdar et al.~\cite{DBLP:conf/tacas/MajumdarPS19}.

The algorithm works as follows.

First, it computes the winning region $W$ of the obliging game with the original
GR(1) specification, which is not to be left by the optimal strategies as per
Condition~\eqref{item:def:optimal strategy:b} of Definition~\ref{def:optimal
strategy for assumption subset}.
% From now on, we will consider only the subgame $\G' = \G[W]$, because the
% vertices outside of $\G'$ will surely not be chosen by \PZ strategies, whereas \PO
% will not have any edges going outside of $\G'$ (follows from the solution of
% obliging games~\cite{chatterjee2010obliging}).

The map $\BR(\cdot,\cdot)$ is initialized as $\BR\equiv \bot$, and is
incrementally built as follows.
We use a nested for loop over the subsets of $\Gset$ and subsets of $\Aset$,
shown in Lines~\ref{line:alg:static:for loop:G} and \ref{line:alg:static:for
loop:A}, respectively.
In the outer loop, we iterate over the subsets of $\Gset$ in the
decreasing order (according to ``$\sqsubseteq$''), and in the inner loop we
iterate over subsets of $\Aset$ in the increasing order (according to
``$\sqsubseteq$'').
For each pair $\tilde{\Gset}\subseteq \Gset$ and $\tilde{\Aset}\subseteq \Aset$,
we solve the obliging game (Line~\ref{line:alg:static:solve game}) where \PZ aims to fulfill the fragment
$\gr{\tilde{\Aset}}{\tilde{\Gset}}$ of the original GR(1) specification
$\gr{\Aset}{\Gset}$ as well as the safety condition $\safe(W)$ (to stay within $W$), while allowing \PO to meet the original set $\Aset$ of
assumptions. 
Solving this obliging game provides us a winning region $Z\subseteq W \subseteq
\V$ as well as a finite-memory winning strategy $\polZ \in \PolZFM$ for \PZ.
Now for every vertex $\v\in Z$ and every $\Aset'\supseteq \tilde{\Aset}$, if we
have not yet found the optimal strategy, i.e., if $\BR(\v,\Aset')=\bot$, then
$\BR(\v,\Aset')$ is set to $\polZ$ (Line~\ref{line:alg:static:populate map
Gamma}).

Intuitively, for each $\v\in Z$, since $\polZ$ fulfills $\tilde{\Gset}$ against
$\tilde{\Aset}$, $\polZ$ will also fulfill $\tilde{\Gset}$ against \textit{every
superset} of $\tilde{\Aset}$ (more assumptions = easier for \PZ); this explains the ``every $\Aset'\supseteq
\tilde{\Aset}$'' part.
Moreover, as we are iterating over $\Gset$ in the \textit{decreasing} order,
$\tilde{\Gset}$ will remain the largest subset of guarantees that can be
fulfilled against every $\Aset'\supseteq \tilde{\Aset}$ from $\v\in Z$.
This means, we should prevent $\Gamma(\v,\Aset')$ to be overwritten in future
iterations; this explains the ``if $\BR(\v,\Aset')=\bot$'' part.

Since the last iteration over guarantees uses $\tilde{\Gset}=\emptyset$, this
effectively provides us a graceful strategy that is not going to achieve any
guarantees; formally, this involves solving the obliging game
$(\G,\top,\gbuchi(\Aset))$, where $\top$ indicates that all paths are in the
specification.
It then follows from the definition of $W$ that for every vertex $\v\in W$ and
every $\tilde{\Aset}\in 2^{\Aset}$, the set $\BR(\v,\tilde{\Aset})$ is nonempty.

The final step to obtain $\static$ from $\BR(\cdot,\cdot)$ is described in
Eqn.~\eqref{eq:going from Gamma to static}, where we essentially ``stitch up''
different optimal strategies of \PZ stored for different vertices stored in $\BR(\cdot,\cdot)$.
Intuitively, for every assumption subset $\Aset'\in 2^{\Aset}$, we require the
strategy $\static(\Aset')$ of \PZ to do the following:
for every finite path $\rho$ ending at a \PZ vertex $\v$, use $\BR(\Aset',\v)$
to decide which vertex to select next.
The optimality of such stitched up strategies returned by $\static(\cdot)$ is far
from trivial, because different strategies of \PZ might interfere with each other in
an adverse manner.
In the following theorem, we show that this is not the case, where we use a
ranking argument over the iteration indices of the two for loops to prove that
the ``stitched up'' strategies of \PZ returned by $\static(\cdot)$ are actually
optimal against every subset of assumptions.

\begin{algorithm}
	\caption{Static computation of the optimal strategy library}
	\label{alg:static computation of strategy library}
	\begin{algorithmic}[1]
		\Require Game graph $\G = (\V,\VZ,\VO,\E)$, GR(1) specification $\gr{\Aset}{\Gset}$
		\Ensure The \textit{intermediate} optimal strategy library $\Gamma\colon
		\V\times 2^\Aset \to \PolZFM$
		\State $(W,\cdot)\gets \Synt(\G,\gr{\Aset}{\Gset},\gbuchi(\Aset))$
		% the winning region of environmentally friendly GR(1) synthesis for the original specification $\gr{\Aset}{\Gset}$ in the original game graph $\G$
		% \State $\G'\gets \G[W]$
		\State $\forall v\in \V\;.\;\forall \tilde{\Aset}\subseteq
		\Aset\;.\;\BR(v,\tilde{\Aset}) \gets \bot$ \Comment{initialization of $\Gamma$}
		\For{every $\tilde{\Gset}\in 2^{\Gset}$ in the
		decreasing sequence $\Gset \sqsupseteq \Gset' \sqsupseteq \Gset'' \ldots
		\sqsupseteq \emptyset$}\label{line:alg:static:for loop:G}
		\For{every $\tilde{\Aset}\in 2^{\Aset}$ in the increasing sequence $\emptyset\sqsubseteq \Aset' \sqsubseteq \Aset'' \ldots \sqsubseteq \Aset$}\label{line:alg:static:for loop:A}
		\State $(Z,\polZ)\gets
		\Synt(\G,\gr{\tilde{\Aset}}{\tilde{\Gset}}\cap \safe(W),\gbuchi(\Aset))$
		\label{line:alg:static:solve game}
				\For{every $\v\in Z$ \textit{and} every $\Aset' \supseteq \tilde{\Aset}$}
					\If{$\BR(\v,\Aset') = \bot$} \Comment{to ensure the optimal policy (most guarantees) is stored}
					\State $\BR(v,\Aset') \gets
					\polZ$ \label{line:alg:static:populate map Gamma}
					\EndIf
				\EndFor
			\EndFor
		\EndFor
		% \For{every $\Aset'\in 2^{\Aset}$}
		% 			\State $\forall \rho\in
		% 			\pathsfin(\G')\;.\;(\last(\rho)\in \VZ) \implies
		% 			\static(\Aset')(\rho) \coloneqq
		% 			\BR(\last(\rho),\Aset')$\label{line:alg:static:stitching strategies}
		% \EndFor
	\end{algorithmic}
\end{algorithm}

\begin{restatable}[Soundness of statically computed optimal
strategies]{theorem}{thmStaticComputationSoundness}\label{thm:static computation soundness}
Suppose $\gr{\Aset}{\Gset}$ is the given GR(1) specification.
For every $\Aset^*\in 2^{\Aset}$, and for every initial vertex that is in the
domain of the graceful winning region for $\gr{\Aset}{\Gset}$, the strategy
$\static(\Aset^*)$ defined in Eqn.~\eqref{eq:going from Gamma to static} is optimal against $\Aset^*$.
\end{restatable}

\begin{proof}[Proof sketch, details in Appendix~\ref{appendix:proofs of the static design section}]
	The proof uses three observations: (i) winning regions of obliging games are traps for \PO, (ii) winning strategies and cooperative \PO strategies (that fulfill all assumptions) are finite-memory with bounded hitting times, and (iii) winning regions are monotonic in the specification (fewer assumptions or more guarantees yield smaller winning regions).
	We rank each vertex $\v$ by the iteration index when $\BR(\v,\Aset^*)$ is first changed from $\bot$, and prove by induction on ranks that $\static(\Aset^*)$ is optimal: the base case uses direct construction from the algorithm, and the inductive step stitches strategies across rank boundaries using Proposition~\ref{prop:prefix-independence of graceful strategies}.
\end{proof}

We now present the upper complexity bound of Algorithm~\ref{alg:static
computation of strategy library}.
Recall that the approach based on the work of Bloem et
al.~\cite{DBLP:conf/cav/BloemCGHJ10}, presented in Section~\ref{sec:hardness of
graceful robustness}, had a doubly exponential time and memory complexity with
respect to $m$ (number of assumptions) and $n$ (number of guarantees)
(Corollary~\ref{cor:graceful complexity}), whereas we show that
Algorithm~\ref{alg:static computation of strategy library} has only a single
exponential time and memory complexity.
In other words, the static (offline design) part of our approach is
exponentially faster and uses exponentially smaller memory than the baseline.
However, this gain comes at a cost: while Bloem et
al.~\cite{DBLP:conf/cav/BloemCGHJ10} computes gracefully robust strategies,
our technique computes \textit{terminally} gracefully robust strategies
under Assumption~\ref{assump:finite-memory environment}, which is a practically
motivated relaxation of the problem.

\begin{restatable}[Upper complexity bound]{theorem}{thmStaticComplexity}\label{thm:static complexity}
	Let $m=|\Aset|$ and $n=|\Gset|$.
	Then Algorithm~\ref{alg:static computation of strategy library} computes the map
	$\BR$, and hence the library $\static$, in time
	\[
		\mathcal{O}\big(2^{m+n}\,|V|^{m+1}\cdot|E|\cdot (mn)^{m+2}\cdot (m+1)! + |\V|\cdot 2^{2m+n}\big),
	\]
	and needs memory of size
	$\mathcal{O}\big(2^{m+n}\cdot (m+1)!\big)$.
\end{restatable}

\begin{proof}[Proof sketch, details in Appendix~\ref{appendix:proofs of the static design section}]
	Each obliging game
	$\Synt(\G,\gr{\tilde{\Aset}}{\tilde{\Gset}}\cap\safe(W),\gbuchi(\Aset))$
	reduces to a Streett game with $(m+1)$ pairs, $\mathcal{O}(|V|\cdot mn)$
	vertices, and $\mathcal{O}(|E|\cdot mn)$ edges (the safety condition does not significantly increase complexity,
	handled by pruning unsafe edges~\cite{bloem2015shield}).
	Using known Streett-game complexity~\cite{piterman2006faster}, each $\Synt$ call costs $\mathcal{O}(|V|^{m+1}|E|(mn)^{m+2}(m+1)!)$ time and $(m+1)!$ memory.
	The nested loops iterate $2^{m+n}$ times; summing these costs and the $\BR$ updates yields the stated bounds.
\end{proof}

To summarize this section, Algorithm~\ref{alg:static computation of strategy
library} provides us $\static$, which maps every subset of assumptions to an
optimal strategy, and we established that the process of obtaining $\static$
needs exponentially smaller runtime and memory as compared to the baseline.
Now we proceed to the dynamic adaptation part, where the strategies in $\static$
are combined at runtime to obtain a terminally gracefully robust strategy for \PZ.

\section{Adaptive Design of Terminally Gracefully Robust Policies}
\label{sec:adaptive strategy}

The adaptive strategy of \PZ uses a time-varying probabilistic mixture of the
statically computed optimal strategies.
Before presenting this probabilistic approach, we briefly describe a
\textit{deterministic} alternative, discuss its pitfalls, and motivate the use
of randomization.

We describe the deterministic strategy using the repeated matching penny game
presented in Example~II in Section~\ref{sec:intro}.
Consider the strategy $\polZ^{(d)}$ of \PZ which operates in ``phases,'' where
the $i$-th phase lasts for $i$ time steps, and $\polZ^{(d)}$ picks $h_\eve $ in
even phases and $t_\eve $ in odd phases.
We explain why this is robust:
by assumption, \PO maintains a partition $\Asetinf$ and
$\Asetfin$ such that against every fixed strategy of \PZ, every assumption $A\in
\Asetinf$ is visited within $N$ steps for some $N\geq 0$.
Let $\Asetinf = \set{T_\adam}$, but importantly, both $N$ and $\Asetinf$,
$\Asetfin$ are unknown to \PZ.
The strategy $\polZ^{(d)}$ of \PZ creates windows of increasing lengths where
only one of her actions is executed, which guarantees that after the $N$-th
phase, there will be infinitely many time windows of lengths larger than $N$
during which \PZ will constantly choose $t_\eve $, guaranteeing that there are
infinitely many matches (i.e., the state $t_\eve t_\adam$ is visited infinitely often).
It can be easily verified that no matter what $\Asetinf$ and $\Asetfin$ are, \PZ
will fulfill the largest possible number of guarantees using this strategy.

Other than the infinite memory requirement, another drawback of the strategy
$\polZ^{(d)}$ is that in the limit, matches will become increasingly rarer,
because phases where \PZ picks the right actions get increasingly distant.
Even though this does not hurt from the standpoint of the LTL semantics of the
GR(1) specification, such behaviors could be undesirable in many practical
scenarios like Example~I.
In contrast, our adaptive strategies guarantee that in the long run,
almost surely, they converge to the static strategy corresponding to the
true set of assumptions fulfilled by \PO.

We now describe our adaptive construction of terminally gracefully robust
strategies.
The key components are \textit{monitors}, whose job is to provide a real-time
estimate of the relative likelihood of the fulfillment of a given set of
assumptions; the monitoring algorithm is presented in Section~\ref{sec:liveness monitors}.
Based on the estimates from the monitors,  we present a probabilistic strategy of \PZ
that adaptively selects a mixture of the statically computed optimal strategies in
accordance with the likelihood estimate obtained from the monitor.

\subsection{Monitoring Environment Assumptions}
\label{sec:liveness monitors}

In formal verification, monitors are software components that passively observe
a sequence of events, and after each new observation, outputs a verdict on
whether the given specification is still satisfied or violated.
Since the assumptions in GR(1) specifications are all \textit{liveness}
properties---always eventually some assumption state should be visited, it is
well-known that they cannot be monitored in the usual sense.
This is because, for  \textit{every} finite path, there will be some infinite
extensions on which all assumptions will hold true but some where all will be
false.
Since the monitor needs to make verdicts based on finite histories only, it will
always remain undetermined about whether the liveness property holds or not.

In the following, we present our novel liveness monitors for B\"uchi
specifications, which, at each time point, outputs a quantitative estimate of the satisfaction of the B\"uchi specification.
Under mild assumptions on the frequency of appearance of the B\"uchi vertices,
as time goes to infinity, the outputs of the monitor converge to $1$ if the
specification holds, and they converge to $0$ otherwise.
To the best of our knowledge, such liveness monitors have not been considered before,
and they could be of independent interest.
Therefore, we make the presentation of liveness monitors in a self-contained
manner.

We imagine that we are given some arbitrary B\"uchi specification $\buchi(S)$
for some arbitrary set $S$ of vertices; in the end, for a given GR(1)
specification $\gr{\Aset}{\Gset}$, we will build a local liveness monitor for
$\buchi(A)$ for each $A\in \Aset$, as well as a global liveness monitor for
$\buchi(\cup_{A\in \Aset}A)$.

\begin{algorithm}
	\caption{Asymptotic liveness monitor for $\buchi(A)$}
	\label{alg:liveness monitor}
	\algrenewcommand\algorithmicrequire{\textbf{Input parameters:}}
	\algrenewcommand\algorithmicensure{\textbf{Output at each step:}}
	\begin{algorithmic}[1]
		\Require initial decay rate $\alpha^0\in (0,1)$, attenuation rate $\lambda \in (0,1)$
		\Ensure current likelihood score $w$
		\Function{Init}{}\Comment{initialize the monitor}
		\State $w \gets 1$ \label{line:alg:liveness
		monitor:a}\Comment{initialize the likelihood estimate variable}
		\State $\alpha \gets \alpha^0$ \label{line:alg:liveness
		monitor:b}\Comment{initialize the decay rate variable}
		\EndFunction
	\end{algorithmic}
	\begin{algorithmic}[1]
		\Function{Next}{$\v$}\Comment{updated output of the monitor upon
		receiving the next vertex}
			\If{$\v \in A$} 
		\State $w \gets 1$ \label{line:alg:liveness monitor:c}\Comment{reset the
		likelihood estimate to the highest value becaues $A$ was seen}
		\State $\alpha \gets  \lambda\alpha$ \label{line:alg:liveness
		monitor:d}\Comment{reduce the decay rate (increased ``confidence'' about fulfilling $\buchi(A)$)}
			\Else
		\State $w \gets  (1-\alpha)w$ \label{line:alg:liveness monitor:e}
		\Comment{reduce the likelihood estimate because $A$ was not seen}
			\EndIf
			\State \Return $w$
		\EndFunction
	\end{algorithmic}
\end{algorithm}

\begin{definition}[Asymptotic liveness monitors]\label{def:liveness monitor}
	An asymptotic liveness monitor, or simply a liveness monitor, for
	$\buchi(A)$ is a function $\M\colon \pathsfin(\G)\to (0,1)$ such that for every infinite path $\rho=v^0v^1\ldots\in \pathsinf(\G)$, the following hold:
	\begin{enumerate}[(I)]
		\item if $\rho\in \buchi(A)$, the sequence
		$\M(v^0),\M(v^0v^1),\M(v^0v^1v^2)\ldots$ converges to $1$; \label{list
	item:assumption monitor:convergence to 1}
		\item if $\rho\notin \buchi(A)$, the sequence
		$\M(v^0),\M(v^0v^1),\M(v^0v^1v^2)\ldots$ converges to $0$. \label{list
	item:assumption monitor:convergence to 0}
	\end{enumerate}
\end{definition}

Here, convergence of an infinite sequence $z^0z^1\ldots$ of real numbers to
a point $z^*$ would mean: for every real number $\varepsilon>0$, there exists a
natural number $N$ such that for every $n\geq N$, $\|z^n-z^*\| < \varepsilon$.
We will denote this by writing ``$z^n\to z^*$ as $n\to \infty$.''

The liveness monitor is presented in Algorithm~\ref{alg:liveness monitor}, and
is explained below.
The monitor introduces the variables $w$ and $\alpha$, called the
\textit{likelihood score}---or sometimes simply called \textit{score}---for the given B\"uchi specification and the \textit{decay rate}, respectively.
Also, the monitor uses a constant $\lambda \in (0,1)$, called the
\textit{attenuation rate}.

Intuitively, the monitor works as follows.
The values of $w$ and $\alpha$ changes over time: the value of $w$ ranges
within $(0,1]$, with its initial value being $1$ (Line~\ref{line:alg:liveness
monitor:a}), whereas the value of $\alpha$ ranges within $(0,1)$, with its
initial value being set arbitrarily in $(0,1)$ (Line~\ref{line:alg:liveness monitor:b}). 
Every time a vertex from $A$ is seen, $w$ is reset to its maximum value $1$
(Line~\ref{line:alg:liveness monitor:c}), from which point it starts decaying
with the rate $(1-\alpha)$ (Line~\ref{line:alg:liveness monitor:e}) until the
next reset upon the next visit to $A$; this way, if $A$ is seen only finitely
often, then $w$ would converge to $0$ as time approaches infinity.
Clearly, the smaller $\alpha$ is, the slower is the decay of $w$. 
Now the decay rate $\alpha$ itself gets discounted by $\lambda$ every time $A$
is visited (Line~\ref{line:alg:liveness monitor:d}).
This captures the situation that as $A$ is visited more and more times, we are
getting more confident about the fulfillment of $\buchi(A)$, and the monitor is
becoming more ``patient'' in receiving the next $A$, which is reflected in the
slower decay of $w$.
Figure~\ref{fig:matching pennies strategy illustration} provides a visual
illustration of how the liveness monitors would behave in response to the
frequency of visits to $H_1$ and $T_1$ in the matching pennies example from
Section~\ref{sec:intro}; for the moment, ignore the bottom three rows.

\input{FIGURES/matching_pennies_policy_illustration}

\begin{restatable}[Soundness of Algorithm~\ref{alg:liveness
monitor}]{theorem}{thmMonitorSoundness}\label{thm:monitor soundness}
	Let $W \subseteq \pathsinf(\G)$ be the set of infinite paths such that for
	every $\rho\in W$, either $\rho\notin \buchi(A)$ or $A$ has a
	constant hitting time in $\rho$.
	(The exact hitting time need not be available.)
	Then the monitor $\M\colon W\to (0,1)$ implemented by Algorithm~\ref{alg:liveness monitor} is an asymptotic liveness monitor for the specification $\buchi(A)$.
\end{restatable}

\begin{proof}[Proof sketch, see Appendix~\ref{appendix:proofs of the adaptive section}]
	For convergence to $1$ when $\rho\in\buchi(A)$: after $k$ visits to $A$, the decay rate reduces to $\lambda^k\alpha^0$, and with constant hitting time bounded by some $n$, between consecutive visits $w$ remains at least $(1-\lambda^k\alpha^0)^n {\to} 1$ as $k{\to}\infty$.
	For convergence to $0$ when $\rho{\notin}\buchi(A)$: after the last visit to $A$, the weight decays by a constant factor $(1-\alpha)$ each step, so $w\to 0$.
\end{proof}

While Theorem~\ref{thm:monitor soundness} proves soundness of liveness monitors
for constant hitting time of the live vertices, the claim can be significantly
strengthened by showing that the same monitors would work for when the hitting
time is bounded by any sub-exponential function. 
This result could be of independent interest in runtime verification, and
therefore, we have included the general statement and its proof in
Appendix~\ref{sec:appendix:monitor generalization}.

\subsection{Randomized Mixing of Policies via Assumption Monitoring}
\label{sec:randomized mixing}

Algorithm~\ref{alg:adaptive strategy} presents our construction of the adaptive
strategy for \PZ.
We describe the algorithm in the following.
In our description and subsequent analysis, we explicitly include the time $t$
to several quantities, and write $w_i^t$, $w_{\mathbf{B}}^t$,
$p_{\mathbf{B}}^t$, etc. 
Whereas in Algorithm~\ref{alg:adaptive strategy}, the superscript ``$t$'' is
omitted and the same quantities are treated as variables whose values get
overwritten with a new value at each time step.

For each individual assumption $A_i\in \set{A_1,\ldots,A_m}= \Aset$, an
asymptotic liveness monitor $\M_i$ is instantiated to keep track of the
likelihood score of $\buchi(A_i)$.
These likelihood scores for the individual assumptions are aggregated to
form likelihood scores for $\grabin(\Bset,\Aset\setminus \Bset)$ for every
$\Bset\in 2^{\Aset}\setminus \emptyset$, indicating the likelihood that
$\Bset$ is the \textit{only} set of assumptions that are being fulfilled.
Intuitively, the likelihood score for $\grabin(\Bset,\Aset\setminus \Bset)$
should be high iff every $A\in \Bset$ has a high likelihood score \textit{and}
every $A\notin\Bset$ has a low likelihood score.
Formally, let at any given time $t$ the liveness score computed by $\M_i$ be
$w_i^t$ for each $i$.
Then for every $\Bset\in 2^{\Aset}\setminus \emptyset$, define the
\textit{intermediate} likelihood score for every $t=0,1,\ldots$:
\begin{align}\label{eq:definition of w tilde for nonempty sets}
\tilde{w}_{\mathbf{B}}^t\coloneqq \min_{A\in
\mathbf{B}} w_A^t - \max\left(\max_{A\notin \mathbf{B}} w_A^t, \gamma^t\right), && \text{with } \gamma^t = \frac{1}{\ln(3+t)}.
\end{align}
Since $0< w_A^t \leq 1$ and $\gamma^t>0$ for all $t\geq 0$, we obtain $-1 < \tilde{w}_{\Bset}^t < 1$.
We will refer to $\gamma^t$ as the \textit{bias} in intermediate weights, and is
a regularization term that prevents degenerate cases at finite times; it decays
slowly enough ($\gamma^t \to 0$) to preserve the correct asymptotic behavior;
the technical justification of bias is difficult to provide at this point, and
it can be found in Remark~\ref{rem:dominated within
subgame} in Appendix~\ref{appendix:proofs of the adaptive section} after the
proof of soundness of our adaptive strategies (Theorem~\ref{thm:overall soundness}).

% \KS{
% Since $0< w_A^t \le 1$ at all $t$, every $\Bset$ with
% $\emptyset\neq\Bset\subsetneq\Aset$ satisfies $-1 < \tilde{w}_{\Bset}^t < 1$.
% For the full set $\Bset=\Aset$, the maximum is taken over the empty set
% (and equals $0$), so $\tilde{w}_{\Aset}^t=\min_{A\in\Aset} w_A^t\le 1$. So value 1 is attainable, right?  But then $w_{\Bset}^t$ below can become $\infty$.}

While we can now estimate the likelihood of satisfaction of nonempty sets of
assumptions, we need a separate monitor for tracking if none of the assumptions
will be fulfilled, i.e., the likelihood of the satisfaction of $\grabin(\emptyset,\Aset) = \cobuchi(\cup_{A\in \Aset} A)$.
For this, we instantiate a separate liveness monitor $M_\cup$ to keep track of
the \textit{dual} likelihood score $w_\cup^t$ of $\buchi(\cup_{A\in \Aset}A)$ at
time $t$.
Clearly, the higher $w_\cup^t$ is, the lower is the likelihood
$\tilde{w}_\emptyset^t$ that the set of fulfilled assumptions is empty.

%\KS{Formally, we define the \textit{intermediate} likelihood estimate of the empty
%set of assumptions as:

%\[
%\tilde{w}_{\emptyset}^t \coloneqq 1-2w_\cup^t,
%\]

%We scale $w_\cup^t$ by $2$, before subtracting from $1$, so that
%$\tilde{w}_{\emptyset}^t$ takes values in the interval $[-1,1)$ as the
%$\tilde{w}_{\Bset}^t$ for nonempty $\Bset$: since $0<w_\cup^t\leq 1$, we have
%$-1\leq \tilde{w}_{\emptyset}^t<1$. 

%Finally, for every $\Bset$ and $t$, we define the \textit{likelihood score} as:
%\begin{equation}\label{eq:def:w for sets of assumptions}
%	w_{\Bset}^t \coloneqq \frac{1+\tilde{w}_{\Bset}^t}{1-\tilde{w}_{\Bset}^t}.
%\end{equation}

%Consequently $w_{\emptyset}^t=\frac{1+\tilde{w}_{\emptyset}^t}{1-\tilde{w}_{\emptyset}^t}\in[0,\infty)$;
%in particular $w_{\emptyset}^t=0$ exactly when $w_\cup^t=1$ (a vertex of
%$\cup_{A\in\Aset}A$ was just visited), and $w_{\emptyset}^t\to\infty$ when
%$w_\cup^t\to 0$.  Note that $w_\cup^t >0$ so 
%$w_{\emptyset}^t < \infty$. Thus, 
% $w_B >0$ for all non empty $B$ and $w_{\emptyset} \geq 0$.
%This transformation maps $-1<\tilde{w}_{\Bset}^t < 1$ to $0<w_{\Bset}^t<\infty$ for $\Bset \neq \emptyset$, and  $-1 \leq \tilde{w}_{\emptyset}^t < 1$ to $0 \leq w_{\emptyset}^t<\infty$ 
%and strongly punishes negative values of $\tilde{w}_{\Bset}^t$ while amplifying
%the positive positive values.
%}

Formally, we define the \textit{intermediate} likelihood estimate of the empty
set of assumptions as:
\begin{equation}\label{eq:definition of w tilde for empty sets}
\tilde{w}_{\emptyset}^t \coloneqq \min\!\big(1-2w_\cup^t,\; 1-\gamma^t\big),
\end{equation}
where $\gamma^t$ is the bias from Eqn.~\eqref{eq:definition of w tilde for nonempty sets}.
We scale $w_\cup^t$ by $2$ before subtracting from $1$ so that, like the nonempty
sets, $\tilde{w}_\emptyset^t$ ranges over $[-1,1)$: since $0<w_\cup^t\leq 1$ we
have $1-2w_\cup^t\in[-1,1)$, and the cap at $1-\gamma^t$ keeps
$\tilde{w}_\emptyset^t$ bounded away from $1$, giving
\[
-1\ \leq\ \tilde{w}_\emptyset^t\ \leq\ 1-\gamma^t\ <\ 1 .
\]
Note that nonempty sets $\Bset\in
2^{\Aset}$ also  satisfies $\tilde{w}_\Bset^t\leq 1-\gamma^t$, equivalently
$1-\tilde{w}_\Bset^t\geq\gamma^t$.

Finally, for every $\Bset$ and $t$, we define the \textit{likelihood score} as:
\begin{equation}\label{eq:def:w for sets of assumptions}
	w_{\Bset}^t \coloneqq \frac{1+\tilde{w}_{\Bset}^t}{1-\tilde{w}_{\Bset}^t}.
\end{equation}
Because $1-\tilde{w}_\Bset^t\geq\gamma^t>0$ for every $\Bset$, each $w_\Bset^t$ is
finite, with $0<w_\Bset^t\leq \tfrac{2}{\gamma^t}$ for nonempty $\Bset$ and
$0\leq w_\emptyset^t\leq \tfrac{2}{\gamma^t}$. In particular $w_\emptyset^t=0$
exactly when $w_\cup^t=1$ (a vertex of $\cup_{A\in\Aset}A$ was just visited), and
$w_\emptyset^t\to\infty$ when $w_\cup^t\to 0$: as
$w_\cup^t\to0$, $\tilde{w}_\emptyset^t=1-\gamma^t\to1$ since
$\gamma^t\to0$, and hence $w_\emptyset^t\to\infty$. Summarizing,
$w_\Bset^t>0$ for every nonempty $\Bset$ and $w_\emptyset^t\geq 0$, so the
denominator in the normalization \eqref{eq:probability from weights} is always
positive.

The transformation~\eqref{eq:def:w for sets of assumptions} maps
$-1<\tilde{w}_\Bset^t\leq 1-\gamma^t$ to $0<w_\Bset^t\leq \tfrac2{\gamma^t}$ for
nonempty $\Bset$, and $-1\leq\tilde{w}_\emptyset^t\leq 1-\gamma^t$ to
$0\leq w_\emptyset^t\leq \tfrac2{\gamma^t}$; it strongly punishes negative values
of $\tilde{w}_\Bset^t$ while amplifying the positive ones. As $t\to\infty$ the
bound $\tfrac2{\gamma^t}\to\infty$.

%\KM{Probably a better explanation is needed on why this particular
%transofmration. Think. the honest answer is that the math worked! Also, add a
%remark on why sigmoid would not work (bounded probability range, does not
%converge to zero or one) and why clipping $\tilde{w}_{\Bset}^t$ to zero wouldn't
%work (the probability of one strategy could suddenly drop to zero, which would
%mess up Lemma~\ref{claim:lower bound on N-step consecutive execution of any one
%strategy}). However, I have not considered the simple transformation
%$\tilde{w}_{\Bset}^t\mapsto 1+ \tilde{w}_{\Bset}^t$. I can imagine, this would
%not punish the negative values as strongly, otherwise I can't think of any
%reason. Maybe an example is required here. } 

Summarizing, using $m+1$ separate liveness monitors $\M_1,\ldots,\M_m,\M_\cup$,
for every $\Bset\in 2^\Aset$, we obtain an estimate of the likelihood of that
$\Bset$ is the precise set of assumptions that is going to be fulfilled in the long run.

We now proceed to turn the likelihood estimates to probability distributions.
Intuitively, the adaptive strategy of \PZ is a probabilistic mixture of the
statically computed optimal strategies represented using $\static(\cdot)$, where,
for every $\Bset\in 2^{\Aset}$, the probability of using the strategy
$\static(\Bset)$ at each time step $t$ will be correlated to the likelihood
estimate $\tilde{w}_{\Bset}^t$ of $\Bset$.

We turn the likelihood estimates into a \textit{probability distribution} via normalization.
At each time $t$ and given any subset $\mathbf{B}\in 2^{\Aset}$, the current probability estimate that the only $\mathbf{B}$ will
be met in the long run, i.e., $\grabin(\mathbf{B},\Aset\setminus\mathbf{B})$
will be satisfied, is given by:
\begin{align}\label{eq:probability from weights}
p_{\mathbf{B}}^t \coloneqq \frac{w_{\mathbf{B}}^t}{\sum_{\mathbf{B}'\in 2^{\Aset}}w_{\mathbf{B}'}^t}
\end{align}
Since $w_{\Bset}^t > 0 $ for all $\Bset\neq \emptyset$ and all $t$, the denominator of \eqref{eq:probability from weights} is always
nonzero, and therefore $p_{\mathbf{B}}^t$ is well-defined.
Moreover, 
\[
\sum_{\Bset\in 2^\Aset} p_{\Bset}^t = \sum_{\Bset\in 2^\Aset} \frac{w_{\mathbf{B}}^t}{\sum_{\mathbf{B}'\in 2^{\Aset}}w_{\mathbf{B}'}^t} = 1,
\]
implying that $(p_\Bset^t)_{\Bset\in 2^\Aset}$ is a valid probability
distribution.
Figure~\ref{fig:matching pennies strategy illustration} shows an illustration of how the quantities $\tilde{w}_{\Bset}^t$, $w_{\Bset}^t$, and
$p_{\Bset}^t$ would vary over time for a sample run of the matching pennies
example (from Section~\ref{sec:intro}).

\begin{algorithm}
	\caption{Adaptive strategy for $\gr{\Aset}{\Gset}$}
	\label{alg:adaptive strategy}
	\begin{algorithmic}[1]
		\algrenewcommand\algorithmicrequire{\textbf{Input parameters:}}
		\algrenewcommand\algorithmicensure{\textbf{Output at each step:}}
		\Require The intermediate strategy library $\Gamma$ computed by
		Algorithm~\ref{alg:static computation of strategy library}, monitor
		parameters $\alpha^0$ and $\lambda$
		\Ensure Next vertex selected by the adaptive strategy to extend the
		current path
		\Function{Init}{$\vinit$} \Comment{initialize}
			\State $t\gets 0$ \Comment{variable recording the current time}
			\State $h\gets \vinit$ \Comment{variable recording the path seen so far}
			\For{$A_i=A_1,\ldots,A_m\in \Aset$}
			\State $\M_i\gets$ asymptotic liveness monitor for $\buchi(A_i)$
			with parameters $\alpha^0$, $\lambda$
		\State $\M_i.\Call{Init}{}$ \Comment{monitors for estimating the
		likelihood of the assumptions}
			\EndFor
			\State $\M_\cup \gets$ asymptotic liveness monitor for
			$\buchi(\cup_{A_i\in \Aset} A_i)$ with parameters $\alpha^0$, $\lambda$
		\State $\M_\cup.\Call{Init}{}$ \Comment{monitor for estimating the
		likelihood of the empty set of assumptions}
		\EndFunction
	\end{algorithmic}
	\begin{algorithmic}[1]
		\Function{Next}{$\v$} \Comment{update upon receiving the next vertex}
			\State $t\gets t+1$ \Comment{progress time}
			\State $\gamma^t \gets 1/\ln(3+t)$ \Comment{current bias}
			\State $h\gets h\v$ \Comment{update path }
			\For{$A_i=A_1,\ldots,A_m\in \Aset$}
				\State $w_{A_i}\gets \M_i.\Call{Next}{v}$
			\EndFor
			\State $w_{\cup}\gets \M_\cup.\Call{Next}{v}$
			\If{$\v\in \VO$} \Comment{the opponent picks the next vertex, so
			output nothing}
				\State \Return $\bot$
			\Else \Comment{compute the next vertex by mixing the static optimal strategies}
				\For{every $\mathbf{B}\in 2^\Aset\setminus \emptyset$}
			% \State $w_{\mathbf{B}} \gets \max\left\lbrace \min_{A\in \mathbf{B}}
			% w_A - \max_{A\notin \mathbf{B}} w_A, 0\right\rbrace$
			\State $\tilde{w}_{\mathbf{B}} \gets \min_{A\in \mathbf{B}}
			w_A - \max\set{\max_{A\notin \mathbf{B}} w_A, \gamma^t}$
			\Comment{see Eqn.~\eqref{eq:definition of w tilde for nonempty sets}}
			\EndFor
			\State $\tilde{w}_\emptyset \gets \min\!\big(1-2w_\cup^t,\;
			1-\gamma^t\big)$ \Comment{see Eqn.~\eqref{eq:definition of w tilde for empty sets}}
			\For{every $\mathbf{B}\in 2^\Aset$}
			\State $w_{\mathbf{B}} \gets
			(1+\tilde{w}_{\mathbf{B}})/(1-\tilde{w}_{\mathbf{B}})$ \Comment{see Eqn.~\eqref{eq:def:w for sets of assumptions}}
			\EndFor
				% \State $N\gets w_\emptyset + \sum_{\mathbf{B}\in 2^\Aset\setminus
				% \emptyset}w_{\mathbf{B}}$
				\For{every $\mathbf{B}\in 2^\Aset$}
					\State $p_{\mathbf{B}} \gets w_{
					\mathbf{B}}/\sum_{\mathbf{B}'\in 2^{\Aset}} w_{\mathbf{B}'}$
					\Comment{see Eqn.~\eqref{eq:probability from weights}}
					\label{line:alg:adaptive strategy:probability distribution}
					% \State $p_{\mathbf{B}} \gets \frac{\exp(w_{
					% \mathbf{B}})}{\sum_{\mathbf{B}'\in 2^{\Aset}} \exp(w_{\mathbf{B}'})}$
					% \label{line:alg:adaptive strategy:probability distribution}
				\EndFor
				% \State $p_\emptyset \gets w_\emptyset/N$
				\State Sample $\pi\sim \widetilde{\Pi}$ where
					$P(\widetilde{\Pi}=\BR(\v,\mathbf{B}))=p_{\mathbf{B}}$ for every
				$\mathbf{B}\in 2^{\Aset}$ \Comment{strategy mixing}
				% \parbox[t]{0.8\linewidth}{
				% 	$\pi\sim P(\widetilde{\Pi})$ where\\
				% 	% $P(\widetilde{\Pi}=\BR(\v,\emptyset))=p_\emptyset$ and\\
				% 	$P(\widetilde{\Pi}=\BR(\v,\mathbf{B}))=p_{\mathbf{B}}$ for every
				% $\mathbf{B}\in 2^{\Aset}$}
				\State \Return $\pi(h)$
			\EndIf
		\EndFunction
	\end{algorithmic}
\end{algorithm}

We will show that as $t\to \infty$, 
$w_{\mathbf{B}}^t\to \infty$ for exactly one $\mathbf{B}\in 2^{\Aset}$---the one
for which $\grabin(\mathbf{B},\Aset\setminus \mathbf{B})$ holds, while for
every $\mathbf{B}'\neq \mathbf{B}$, $w_{\mathbf{B}'}^t\to 0$, implying that
$p_{\mathbf{B}}^t \to 1$ while $p_{\mathbf{B}'}^t\to 0$.

Given the probability distribution computed using \eqref{eq:probability from
weights}, at each time $t$, our \textit{adaptive strategy} $\polZadaptive$ selects
a probabilistic mixture of the statically computed optimal strategies represented
by $\static$ (described in Section~\ref{sec:static design phase}).
In particular, for every $\mathbf{B}\in 2^{\Aset}$ and at each vertex, the
corresponding optimal strategy $\static(\mathbf{B})$, call it $\polZ^t$, is
selected with probability $p_{\mathbf{B}}^t$.
If the path seen so far is $h=\v^0,\ldots,\v^t$, then the next vertex is obtained as $\v^{t+1}=\polZ^t(h)$.
Although $h$ is the history produced by the adaptive strategy---which switches between different $\static(\Bset)$'s---this does not violate staying within $W$: every $\Gamma(\v,\Aset')$ is built on a graph where all unsafe edges leaving $W$ are pruned (cf.\ proof of Theorem~\ref{thm:static complexity}), so no strategy can recommend an unsafe move regardless of history.

In the rest of this section, we show that, under Assumption~\ref{assump:finite-memory environment}, the adaptive strategy produced by Algorithm~\ref{alg:adaptive strategy} is terminally gracefully robust, almost surely.
However, graceful robustness does not hold in general: with positive probability, the adaptive strategy may select optimal strategies that concern only fragments of the GR(1) specification.
For instance, in Examples~\ref{ex:gracefulness} and \ref{ex:graceful robustness}, we have $\static(\set{A_1}) = \static(\set{A_2}) = \pol_{\eve,c}$, but $\static(\set{A_1,A_2})$ can be either $\pol_{\eve,b}$ or $\pol_{\eve,c}$.
If $\static(\set{A_1,A_2}) = \pol_{\eve,b}$, then at the initial vertex $v^0=0$, the strategy $\pol_{\eve,b}$ could be selected with positive probability; we have already argued that this yields terminal graceful robustness but not graceful robustness.

\subsection{Theoretical Analysis: Convergence of the Adaptive Strategies}
\label{sec:proof of convergence}

For every arbitrary strategy $\polO$ of \PO, the adaptive strategy $\polZadaptive$ of \PZ
gives rise to a probability space over the sample space of infinite paths in $\G$.
This is formalized as follows.
Suppose $h = v^0\ldots v^t$ is a finite path.
The cylinder set generated by $h$ is the subset of infinite paths, compliant
with $\polO$, whose prefix is $h$, i.e., $\cyl(h)\coloneqq
\set{\rho\in\pathsinf(\G)\mid \exists \rho'\in \langO(\G,\polO,\last(h))\;.\;\rho = h\rho'}$.
Let $\widetilde{\Pi}(h)$ be the probability distribution over the set
$\set{\static(\Bset)}_{\Bset\in 2^\Aset}$ for the given history $h$, as computed
by $\polZadaptive$ in Algorithm~\ref{alg:adaptive strategy}.
Define:
\[
P(v^{t+1}\mid h=v^0\ldots v^t) \coloneqq 
		\sum_{\Bset\in 2^\Aset}\widetilde{\Pi}(h)(\static(\Bset))\cdot \mathds{1}(\static(\Bset)(h)=v^{t+1}).
\]
The operator $P$ defined above is lifted to cylinder sets by inductively
applying it on finite paths:
\[
\P(Cyl(h)) \coloneqq \prod_{s=0}^{t-1} P(v^{s+1}\mid v^0\ldots v^s).
\]
By the standard procedure of using Carath\'eodory's extension
theorem~\cite[Chapter~10]{baier2008principles}, the function $\P$ can be
extended to a unique probability measure---also denoted as $\P$---over the
sample space of infinite paths compliant with the strategy $\polO$ of \PO, i.e., the paths in $\langO(\G,\polO,v^0)$.

We extend the concepts of satisfaction of temporal properties to this
probabilistic setting in the natural way.
For example, the adaptive strategy $\polZadaptive$ of \PZ is graceful $\P$-almost
surely if it is a winning strategy for the GR(1) game $(\G,\gr{\Aset}{\Gset})$,
and for every finite path $\rho$ compliant with $\polZadaptive$, there exists a
strategy $\polO$ of \PO such that 
\begin{equation}\label{eq:nonblockingness}
\P[ \gbuchi(\Aset)| \rho \text{ is the history
}\land \polZadaptive \text{ and } \polO \text{ are used after }\rho] = 1.
\end{equation}

We now proceed to state the main result of this paper.

\begin{restatable}[Soundness of adaptive strategies of \PZ]{theorem}{thmOverallSoundness}\label{thm:overall soundness}
	Suppose Assumption~\ref{assump:finite-memory environment} holds.
	The randomized strategy of \PZ obtained from Algorithm~\ref{alg:adaptive strategy}
	is a terminally gracefully robust strategy, $\P$-almost surely.
\end{restatable}

The proof of Theorem~\ref{thm:overall
soundness} uses a series of intermediate lemmas that we describe one by one.

The first lemma is an auxiliary result that will be used in multiple technical
steps afterwards.
Intuitively, this lemma establishes a (positive) lower bound on the probability
with which any one optimal strategy will be sampled for a given number $N>0$ of
consecutive time steps in any given interval of length $2N$. %, conditioned on $\polZ^*$ being sampled at the current step.
% The parameter $2N$ is called the \textit{horizon} in the following text.
% If all vertices in the next $N$ steps are \PO vertices, then this probability equals $1$---the trivial case.
% In general, the more often we see a \PZ vertex, the less are the chances of
% continuing with the same strategy: an analogy is that the more number of
% steps we roll a dice, the less is the probability of seeing any one side at each roll.
%
The proof of the lemma needs to account for an important criticality.
Namely, for every strategy $\polZ^*$ of \PZ that is optimal for a given set of
assumptions $\Aset^*\in 2^{\Aset}$, the probability of $\polZ^*$ being sampled
may vary continuously over the window of time $2N$, in response to the frequency
of seeing different assumptions on the run.
To give an analogy, imagine we are throwing a dice whose bias is constantly
changing, and our goal is to find a lower bound on the probability with which
any one side of the dice will be seen for $N$ times in a row in any given sequence of $2N$ tosses. 
Our insight is that, even though the probability distribution is shifting, the
net shift in the worst case over the horizon can be bounded.

\begin{restatable}{lemma}{lemLowerBoundConsecutiveExecution}\label{claim:lower bound on N-step consecutive execution of any one strategy}
	Let $N \geq 1$ be an arbitrary integer, and $h=v^0\ldots \v^\tau$ be an
	arbitrary (random) path seen until a given time $\tau\geq 1$.
	In the adaptive strategy, the probability that any one of the statically
	computed optimal strategies will be scheduled for $N$ times in a row
	within the interval $[\tau+1;\tau+2N]$ is at least
	\begin{align}\label{eq:probability lower bound of seeing a streak of one strategy}
		(\gamma^{\tau+2N})^N\cdot\frac{(1-\alpha^0)^{N(N-1)/2}}{2^{(m+5)N}},
	\end{align}
	where $\gamma^t$ is the bias term in intermediate weights at time $t$,
	$\alpha^0\in (0,1)$ is the initial decay rate (given parameter for the
	monitor),	and $m$ is the total number of assumptions in the GR(1)
	specification, i.e., $m=|\Aset|$.
\end{restatable}

\begin{proof}[Proof sketch, details in Appendix~\ref{appendix:proofs of the adaptive section}]
	For a suitably chosen assumption $B\in\Aset$, we track the singleton $\{B\}$ and lower-bound $p_{\{B\}}^t$ at each step.
	Since $w_{\Bset}^t = (1+\tilde{w}_{\Bset}^t)/(1-\tilde{w}_{\Bset}^t)$, the selection probability $p_{\{B\}}^t = w_{\{B\}}^t / \sum_{\Bset} w_{\Bset}^t$ can be bounded by controlling the ratio $w_{\Bset}^t / w_{\{B\}}^t$ for all $\Bset$: since numerators are at most $2$ and denominators are at least $\gamma^{\tau+t}$ (the bias term), we get $p_{\{B\}}^t \geq \gamma^{\tau+t} L_t / 2^{m+3}$, where $L_t$ is a lower bound on $1+\tilde{w}_{\{B\}}^t$.
	Two cases determine $L_t$: if some assumption reaches score $\geq 1/2$ during the first $N$ steps, $L_t$ decays as $(1-\alpha^0)^{t-t^*}$; otherwise, all intermediate weights stay bounded away from $-1$, giving $L_t \geq 1 - \gamma^1 > 0$.
	Multiplying over $N$ consecutive steps and using $\sum_{j=0}^{N-1} j = N(N-1)/2$ yields the stated bound.
\end{proof}

%}			
Now towards proving the gracefulness of adaptive strategies, we show that the
adaptive strategy allows \PO to fulfill all the assumptions, almost surely,
i.e., \eqref{eq:nonblockingness} holds.

% -------------- OLD STATEMENT -----------------
% \begin{lemma}\label{claim:random scheduling preserves environmentally friendliness}
% 	The adaptive strategy is graceful with respect to the entire set $\Aset$ of
% 	assumptions, $\P$-almost surely.
% \end{lemma}
% -------------------------------------------------

\begin{restatable}{lemma}{lemRandomSchedulingPreservesEnvFriendliness}\label{claim:random scheduling preserves environmentally friendliness}
	The adaptive strategy fulfills Eqn.~\eqref{eq:nonblockingness}.
\end{restatable}

\begin{proof}[Proof sketch, details in Appendix~\ref{appendix:proofs of the adaptive section}]
	Each statically computed strategy $\static(\Aset^*)$ is graceful (from Theorem~\ref{thm:static computation soundness}), so against each strategy $\polZ^* = \static(\Aset^*)$ of \PZ, \PO has a cooperative response that fulfills all assumptions with constant hitting time.
	Let $N$ be the maximum hitting time across all static strategies.
	Divide time into blocks of $2N$ steps; by Lemma~\ref{claim:lower bound on N-step consecutive execution of any one strategy}, each block has conditional probability at least $(\gamma^{\tau+2N})^N \cdot \varepsilon$ of containing $N$ consecutive steps of some static strategy.
	Since $\sum_i 1/(\ln(3+2iN+2N))^N = \infty$ (by integral test), L\'evy's extension of the Borel-Cantelli lemma (Lemma~\ref{lem:borel-cantelli}) implies infinitely many such blocks occur, $\P$-almost surely.
	By pigeonhole, at least one static strategy $\polZ^*$ is active for $N$ consecutive steps in infinitely many blocks; \PO can use the cooperative response to $\polZ^*$ and fulfill all assumptions in $\Aset$ infinitely often.
\end{proof}

In the following two lemmas, for each $A\in \Aset$, we extend the $0$--$1$
convergence result of $w_A^t$ as stated in Theorem~\ref{thm:monitor soundness}
to the probabilistic setting.

\begin{restatable}{lemma}{lemConvergenceToZero}\label{claim:convergence of w to 0 for finitely visited assumptions}
	For every $A\in \Aset$, if $A$ is visited only finitely many times, i.e., if
	$\cobuchi(A)$ holds, $\P$-almost surely, then $w_{A}^t\toas 0$ as $t\to \infty$.
\end{restatable}
	
\begin{proof}
	Consider any infinite run $\rho\in \cobuchi(A)$.
	From Theorem~\ref{thm:monitor soundness} it follows that on $\rho$,
	$w_A^t\to 0$ as $t\to \infty$.
	Now since $\P[\cobuchi(A)]=1$ by assumption, and since $w_A^t\to 0$ as $t\to
	\infty$ on every $\rho\in \cobuchi(A)$, it follows that $\P[\lim_{t\to \infty} w_A^t= 0] = 1$.
\end{proof}

\begin{restatable}{lemma}{lemConvergenceToOne}\label{claim:convergence of w to 1 for infinitely visited
assumptions}
	Suppose Assumption~\ref{assump:finite-memory environment} holds.
	For every $A\in \Aset$, if $A$ is visited infinitely many times, i.e., if
	$\buchi(A)$ holds, $\P$-almost surely, then $w_{A}^t\toas 1$ as $t\to \infty$.
\end{restatable}

\begin{proof}[Proof sketch, details in Appendix~\ref{appendix:proofs of the adaptive section}]
	From the monitor construction, $w_A^t \geq (1 - \lambda^{\mathit{visits}_t} \alpha^0)^t$, where $\mathit{visits}_t$ counts visits to $A$ up to time $t$.
	We show $\mathit{visits}_t$ grows fast enough: using Lemma~\ref{claim:lower bound on N-step consecutive execution of any one strategy}, we prove that $\sum_t \P[\mathit{visits}_t < (\ln t)^2]$ converges, and then by the first Borel-Cantelli lemma (Lemma~\ref{lem:borel-cantelli}), $\mathit{visits}_t \geq (\ln t)^2$ for all large $t$, $\P$-almost surely.
	Writing $\beta = -\ln(\lambda) > 0$, we have $t \lambda^{\mathit{visits}_t} \leq t e^{-\beta (\ln t)^2} = e^{-\ln t(\beta \ln t - 1)} \to 0$ as $t \to \infty$, implying $w_A^t \to 1$.
\end{proof}

We now put together the technical results to obtain the proof of
Theorem~\ref{thm:overall soundness}.

\begin{proof}[Proof sketch of Theorem~\ref{thm:overall soundness}, details in Appendix~\ref{appendix:proofs of the adaptive section}]
	We verify the two conditions of Definition~\ref{def:terminally gracefully robust strategies}: robustness in the limit subgame, and GR(1) winning with gracefulness.\\
	\textbf{Robustness in the limit subgame:} By Lemmas~\ref{claim:convergence of w to 0 for finitely visited assumptions} and~\ref{claim:convergence of w to 1 for infinitely visited assumptions}, as $t\to\infty$, $w_{\Asetinf}^t \to \infty$ while $w_{\Aset^*}^t$ remains bounded for $\Aset^* \neq \Asetinf$, so $p_{\Asetinf}^t \to 1$.
	Thus the adaptive strategy converges to $\static(\Asetinf)$, which is optimal by Theorem~\ref{thm:static computation soundness}, so the adaptive strategy is undominated in every limit subgame.\\
	\textbf{GR(1) winning and gracefulness:} From Lemma~\ref{claim:random scheduling preserves environmentally friendliness}, \PO can fulfill all assumptions $\P$-almost surely.
	If some assumption is violated, the GR(1) specification is satisfied automatically; otherwise $\Asetinf = \Aset$, and since $p_{\Aset}^t \to 1$ and all static strategies keep the play within the winning region $W$, the adaptive strategy fulfills all guarantees always eventually.
\end{proof}

\section{Experiments}
\label{sec:experiments}

We implemented our proposed synthesis algorithm as a prototype tool, \textit{Grace}. 
Theorem~\ref{thm:static complexity} establishes that our static synthesis pipeline enjoys a better asymptotic time complexity than the na\"ive combination of the approach of Bloem et al.~\cite{DBLP:conf/cav/BloemCGHJ10} with the obliging-game reduction (Corollary~\ref{cor:graceful complexity}); we now ask whether this advantage also holds in practice.
Since the technique of Bloem et al. alone does not support gracefulness, their original tool cannot be compared against directly: we instead reimplement their technique combined with the obliging-game reduction on the same front end and back end as \textit{Grace}, so that any observed gap is attributable to the encoding differences under study.
We call this the \emph{Bloem et al.\ baseline} and ask:
\begin{description}
	\item[RQ1 (Comparison with Bloem et al.).] Can our pipeline be run to completion on GR(1) specifications, and how does it compare against the Bloem et al.\ baseline? 
	\item[RQ2 (Convergence).] How quickly does the adaptive strategy of Section~\ref{sec:adaptive strategy} identify the true set $\Asetinf$ of fulfilled assumptions in practice?
\end{description}

% \subsection{Experimental Setup and Evaluation}
\smallskip\noindent\textbf{Experimental setup.}
\textit{Grace} is implemented in C++ as an extension of the FairSyn~\cite{Majumdar2023AFT} framework, which uses Owl~\cite{DBLP:conf/atva/KretinskyMS18}; both \textit{Grace} and the Bloem et al.\ baseline translate the LTL specification into a deterministic Rabin automaton via Owl before symbolic synthesis, and reported runtimes exclude this translation step.
\textit{Grace}'s richer intermediate encoding, needed for the reduction of Theorem~\ref{thm:static complexity}, produces larger deterministic Rabin automata; this is the main source of overhead on small instances, before \textit{Grace}'s better asymptotic scaling takes over.

Our suite spans four parameterized families at two sizes each (\textsc{Buffer}, \textsc{Lift}, \textsc{Scheduler}, \textsc{Network}), single instances of \textsc{Dining Philosophers} and \textsc{Bus Arbiter}, and two classical examples (Matching Pennies, and our running example); full per-instance descriptions are given in Appendix~\ref{appendix:experiments}.
Across these, the underlying game graphs range from $12$ to $54$ vertices and $29$ to $288$ edges, with $2$ to $3$ assumptions and guarantees each.

For RQ2, we simulate the adaptive strategy of Algorithm~\ref{alg:adaptive strategy} against a fixed environment strategy that behaves arbitrarily until some point, after which it consistently fulfills $\Asetinf$; we record the number of steps $t$ until $p_{\Asetinf}^t$ first exceeds $0.99$ and remains above this threshold for the rest of the run; we report this value averaged over $100$ independent simulated runs per benchmark, using monitor parameters $\alpha^0=0.5$, 
$\lambda=0.9$.

All experiments ran on a Linux server with two Intel Xeon E5-2660 v2 CPUs (2.20\,GHz, 20 physical cores) and 62\,GB RAM; LTL-to-Rabin translation had no timeout, while synthesis was capped at a one-hour timeout and a 5\,GB memory limit.

\begin{table}[ht]
\centering
\caption{Runtime comparison between \textit{Grace} and the Bloem et al.\ baseline, and steps until the adaptive strategy's probability of selecting $\static(\Asetinf)$ first exceeds $0.99$ and remains there (averaged over 100 simulated runs). For parameterized families, $(m,n)$ denotes the number of assumptions and guarantees, respectively. Dining Philosophers and Bus Arbiter have $m=n$, so we list this value directly. The ratio column reports the Bloem et al.\ baseline's time divided by \textit{Grace}'s time; bold values indicate \textit{Grace} is faster. $\mathit{TO}$ denotes the one-hour timeout.}
\label{tab:experiments}
\begin{tabular}{lrrrr}
\toprule
\textbf{Benchmark} & \textbf{\textit{Grace} (s)} & \textbf{Bloem et al.\ (s)} & \textbf{Ratio} & \textbf{Steps} \\
\midrule
Dining Philosophers ($n=3$) & 634.7  & 992.3  & \textbf{1.56} & 218 \\
Arbiter ($n=2$)              & 3.90   & 1.82   & 0.47          & 72  \\
Buffer $(2,2)$                & 6.51   & 5.06   & 0.78          & 83  \\
Buffer $(3,3)$                & 245.2  & 280.98 & \textbf{1.15} & 206 \\
Lift $(2,2)$                   & 6.28   & 2.41   & 0.38          & 97  \\
Lift $(3,3)$                   & 274.66 & 390.32 & \textbf{1.42} & 198 \\
Scheduler $(2,2)$             & 7.02   & 13.37  & \textbf{1.90} & 124 \\
Scheduler $(3,3)$             & \textbf{2786.50} & $\mathit{TO}$ & \textbf{$>$1.29} & 274 \\
Network $(2,2)$              & 23.8   & 30.22  & \textbf{1.27} & 94  \\
Network $(2,3)$              & 52.13  & 98.01  & \textbf{1.88} & 161 \\
Matching Pennies             & 6.55   & 1.36   & 0.21          & 193 \\
Running example (Figure~\ref{fig:GR(1) game:illustrating best-effort and limit-best-effort}) & 41.78 & 61.34 & \textbf{1.47} & 116 \\
\bottomrule
\end{tabular}
\end{table}

\smallskip\noindent\textbf{Experimental outcome of RQ1:}
\textit{Grace} completes every benchmark within the time and memory limits, including \textsc{Scheduler}$(3,3)$, where the Bloem et al.\ baseline times out entirely: the pipeline of Section~\ref{sec:static design phase} is tractable in practice, not only in theory.
\textit{Grace} is faster on 8 of 12 benchmarks overall.  
On \textsc{Scheduler}, \textit{Grace} leads by $1.90\times$ at $(2,2)$, and at $(3,3)$ the baseline exceeds the timeout while \textit{Grace} finishes in $2786.50$s, the asymptotic gap of Corollary~\ref{cor:graceful complexity} manifesting as a literal wall.
\textsc{Buffer} and \textsc{Lift} show the same pattern where both tools still complete: from size $2$ to $3$, \textit{Grace}'s runtime grows $37.7\times$ and $43.7\times$ respectively, versus $55.5\times$ and $162.0\times$ for the baseline; \textsc{Network} shows the same trend more mildly ($1.27\times\to1.88\times$).
In every family, \textit{Grace} starts comparable to or slower than the baseline but overtakes it as size grows, as its larger constant-factor overhead (from the richer encoding) is outweighed by better asymptotic scaling.

\smallskip\noindent\textbf{Experimental outcome of RQ2: convergence.}
The adaptive strategy converges within a few hundred steps across all benchmarks.
Convergence time grows with the number of assumptions to distinguish between: e.g., \textsc{Lift}$(2,2)\to(3,3)$ roughly doubles ($97\to198$ steps), consistent with the adaptive strategy needing to rule out larger  candidate subsets $\Bset\in2^{\Aset}$ before its probability mass concentrates on $\Asetinf$.

In practice, floating-point precision can round monitor outputs very close to $0$ or $1$ to those exact values; we prevent this by clipping outputs to $[\epsilon,1-\epsilon]$ for a small constant $\epsilon>0$, so $p_{\Bset}^t$ never exactly reaches $0$ or $1$ and incorrect strategies remain (rarely) samplable. Our experiments confirm this rarely matters in practice: the correct strategy is deployed almost always.

Summarizing, these results support the claim we set out to check: the theoretical advantages of Sections~\ref{sec:static design phase} and~\ref{sec:adaptive strategy} are not merely asymptotic curiosities, but manifest concretely. 
We emphasize that the theoretical framework, not this evaluation, is the main contribution of this paper; the experiments here are intended only to establish that it is not vacuous. 

\section{Conclusion}
\label{sec:conclusion}

We introduced \emph{gracefully robust} strategies for GR(1) games, addressing active assumption violation and guarantee abandonment under partial violations.
Since computing such strategies is doubly exponential, we proposed \emph{terminal} graceful robustness and a two-phase synthesis approach: a static phase precomputes an exponential strategy library, while an adaptive phase uses novel liveness monitors to mix them at runtime, provably converging to the optimal response.

Several directions remain for future work.
First, a more rigorous and optimized implementation of \textit{Grace} could make the approach practical for larger specifications.
Second, heuristics to accelerate the convergence of the adaptive strategy---such as more aggressive decay schedules or informed priors based on domain knowledge---could reduce the number of steps needed to identify $\Asetinf$.
Above all, our method demonstrates a seamless integration of reactive synthesis with reinforcement learning-like ideas: the static phase provides a library of policies, while the adaptive phase performs online learning to identify which policy best matches the actual environment.
This paradigm could have broader applications in settings where specifications or environment models are uncertain, enabling online learning and adaptation of controllers in response to observed behavior.

\begin{acks}
  This work is partly funded by the CEFIPRA project SMILeS (Project number
  7302-1) and the grant RYC2024-049116-I funded by
  MICIU/AEI/10.13039/501100011033 and the ESF+.
\end{acks}

\bibliographystyle{ACM-Reference-Format}
\bibliography{references}

%%%
%%% If your work has an appendix, this is the place to put it.
%\appendix
%Lorem ipsum dolor sit amet, consectetur adipiscing elit. Morbi
%malesuada, quam in pulvinar varius, metus nunc fermentum urna, id
%sollicitudin purus odio sit amet enim. Aliquam ullamcorper eu ipsum
%vel mollis. Curabitur quis dictum nisl. Phasellus vel semper risus, et
%lacinia dolor. Integer ultricies commodo sem nec semper.
%Etiam commodo feugiat nisl pulvinar pellentesque. Etiam auctor sodales
%ligula, non varius nibh pulvinar semper. Suspendisse nec lectus non
%ipsum convallis congue hendrerit vitae sapien. Donec at laoreet
%eros. Vivamus non purus placerat, scelerisque diam eu, cursus
%ante. Etiam aliquam tortor auctor efficitur mattis.
%Nam id fermentum dui. Suspendisse sagittis tortor a nulla mollis, in
%pulvinar ex pretium. Sed interdum orci quis metus euismod, et sagittis
%enim maximus. Vestibulum gravida massa ut felis suscipit
%congue. Quisque mattis elit a risus ultrices commodo venenatis eget
%dui. Etiam sagittis eleifend elementum.
%
%Nam interdum magna at lectus dignissim, ac dignissim lorem
%rhoncus. Maecenas eu arcu ac neque placerat aliquam. Nunc pulvinar
%massa et mattis lacinia.

\appendix
\newpage

\centerline{\Large{\bf{Appendix}}}

The appendix collects deferred proofs from the main technical sections,
a generalization of the liveness-monitor construction, background results
in probability, and the benchmark descriptions used in the experimental evaluation.
To assist navigation:

\smallskip
\noindent\textbf{Appendix~\ref{appendix:proofs of the problem statement section}.}
Proofs of the complexity and robustness claims for gracefully robust
strategies: the complexity of robust-strategy synthesis
(Theorem~\ref{thm:complexity of robust strategies}), the graceful-complexity
bound (Corollary~\ref{cor:graceful complexity}), and that global graceful
robustness implies terminal graceful robustness
(Proposition~\ref{prop:global implies terminal}).

\smallskip
\noindent\textbf{Appendix~\ref{appendix:proofs of the static design section}.}
Soundness (Theorem~\ref{thm:static computation soundness}) and complexity
(Theorem~\ref{thm:static complexity}) of the static design phase that
computes an optimal strategy for every assumption subset

\smallskip
\noindent\textbf{Appendix~\ref{appendix:proofs of the adaptive section}.}
Soundness of the liveness monitor (Theorem~\ref{thm:monitor soundness}),
probabilistic guarantees on the adaptive scheduler
(Lemmas~\ref{claim:lower bound on N-step consecutive execution of any one strategy}
and \ref{claim:random scheduling preserves environmentally friendliness}),
almost-sure convergence of the monitor weights
(Lemma~\ref{claim:convergence of w to 1 for infinitely visited assumptions}),
and overall soundness of the adaptive strategy
(Theorem~\ref{thm:overall soundness})

\smallskip
\noindent\textbf{Appendix~\ref{sec:appendix:monitor generalization}.}
Generalization of the liveness-monitor soundness result to recurring events
with sub-exponential, rather than constant, hitting times

\smallskip
\noindent\textbf{Appendix~\ref{appendix:results in probability}.}
Statement of the Borel--Cantelli lemmas, including L\'evy's
filtration-based extension, used throughout
Appendix~\ref{appendix:proofs of the adaptive section}

\noindent\textbf{Appendix~\ref{appendix:experiments}.}
Description of the benchmark suite used in the experimental evaluation

\clearpage

\section{Proofs of Claims in Section~\ref{sec:gracefully robust strategies}}
\label{appendix:proofs of the problem statement section}

\thmComplexityRobustStrategies*

\begin{proof}
	In the work of Bloem et al.~\cite{DBLP:conf/cav/BloemCGHJ10}, they consider
	robustness specifications specified using ``multiple counting requirements,''
	which are sets of pairs from $[1;m-1]\times [1;n]$.
	Each requirement $s=\set{(a_1,b_1),\ldots,(a_l,b_l)}$ has the restriction that
	for every $i\neq j$, either $a_i<a_j$ or $b_i>b_j$ but not both, and $s$
	represents: if at least $a_1$ assumptions are fulfilled then at least $b_1$
	guarantees must be fulfilled \emph{AND} $\ldots$ \emph{AND} if at least $a_l$
	assumptions are fulfilled then at least $b_l$ guarantees are fulfilled.
	We write $l = |s|$ for the size of a single requirement and $S$ for the set of
	\emph{all} requirements; these are distinct quantities.
	To design non-dominated strategies we iterate over $S$ in the given
	total order, solving the synthesis problem for each $s$ until one succeeds.
	Writing $m'=m-1$, a requirement of size $l$ is obtained by choosing $l$ distinct
	elements from $[1;m']$ and $l$ from $[1;n]$ (one valid arrangement each), so the
	\emph{number} of requirements is
	\[
		|S| = \sum_{l=0}^{\min(m',n)}\binom{m'}{l}\binom{n}{l}
		= \binom{m'+n}{n} = \mathcal{O}(2^{m'+n}) = \mathcal{O}(2^{m+n}),
	\]
	while, by the ordering above, the \emph{size} of any single requirement is
	$l \le \min(m-1,n) = \mathcal{O}(k)$.

	For each $s\in S$ the synthesis question is a generalized reactivity game

	\begin{equation}\label{eq:reduction of multiple counting to generalized reactivity}
	\bigwedge_{p\in[1;l]}\big(\bigwedge_{i\in[1;Y]}\gbuchi(A_{p,i}) \implies
	\bigwedge_{j\in[1;Z]}\gbuchi(G_{p,j})\big)
	\end{equation}
	
	 Here each threshold antecedent
	``at least $a_p$ of the $m$ assumptions'' is the disjunction of conjunctions of
	B\"uchi objectives (a DCB)
	$\bigvee_{\Aset'\subseteq\Aset,\,|\Aset'|=a_p}\bigwedge_{A\in\Aset'}\buchi(A)$,
	of size $a_p\binom{m}{a_p}=\mathcal{O}(m\,2^m)$
(each of the $\binom{m}{a_p}$ disjuncts has $a_p$ B\"uchi objectives, ~\cite[Section~3]{DBLP:conf/cav/BloemCGHJ10}). By
	\cite[Corollary~1]{DBLP:conf/cav/BloemCGHJ10} its equivalent
	generalized-B\"uchi form has $Y=\mathcal{O}\big(2^{m\,2^m}\big)=2^{\mathcal{O}(m\,2^m)}$
	conjuncts, and symmetrically the consequents give $Z=2^{\mathcal{O}(n\,2^n)}$.

Each requirement is a generalized reactivity game of rank $l\le k$ with $Y$
	and $Z$ generalized-B\"uchi conjuncts on the antecedent and consequent sides,
	respectively. By \cite[Theorem~3]{DBLP:conf/cav/BloemCGHJ10}, and bounding the
	rank $l\le k$, such a game is solved in
	$\mathcal{O}\big(|\V|^{k}\,|\E|\,(Y Z)^{k(k+1)}\,k!\big)$ time, and by
	\cite[Theorem~5]{DBLP:conf/cav/BloemCGHJ10} the resulting strategy has size
	$\big((Y+1)(Z+1)\big)^{k}\,k!\,|\V|$. There are $|S|=\mathcal{O}(2^{m+n})$
	requirements to try (the factor $h$ in \cite[Theorem~5]{DBLP:conf/cav/BloemCGHJ10}) so
	the total time is $\mathcal{O}\big(2^{m+n}\cdot|\V|^{k}\,|\E|\,(YZ)^{k(k+1)}\,k!\big)$.
	Since $Y=2^{\mathcal{O}(m\,2^m)}$ and $Z=2^{\mathcal{O}(n\,2^n)}$, this is doubly
	exponential in $m$ and $n$.

\end{proof}

\corGracefulComplexity*

\begin{proof}
	As in Theorem~\ref{thm:complexity of robust strategies}, there are
	$|S|=\mathcal{O}(2^{m+n})$ requirements, each of size $l\le k$. Converting each
	threshold condition in each multiple counting requirement to its equivalent generalized-B\"uchi form via
	\cite[Corollary~1]{DBLP:conf/cav/BloemCGHJ10} (as in the proof of
	Theorem~\ref{thm:complexity of robust strategies}) yields, per pair,
	$Y=2^{\mathcal{O}(m\,2^m)}$ antecedent and $Z=2^{\mathcal{O}(n\,2^n)}$ consequent
	B\"uchi conjuncts. We bound the cost of one requirement and multiply by $|S|$.
	
	Fix $s\in S$. Its generalized reactivity game reduces
	\cite[proof of Theorem~3]{DBLP:conf/cav/BloemCGHJ10} to a Streett game with
	$\mathcal{O}(|\V|\,Y^l Z^l)$ vertices, $\mathcal{O}(|\E|\,Y^l Z^l)$ edges, and
	$l$ Streett pairs (the per-side blow-up is $Y^l Z^l$, in the \emph{converted}
	B\"uchi counts $Y,Z$, not in $m,n$). For gracefulness we solve, on this game,
	the obliging game with weak objective $\gbuchi(\Aset)$; by Chatterjee et
	al.~\cite[Lemma~2, Theorem~6]{chatterjee2010obliging} this $l$-Streett/
	$m$-generalized-B\"uchi obliging game becomes a regular Streett game with
	$\mathcal{O}(|\V|\,Y^l Z^l)$ vertices, $\mathcal{O}(|\E|\,Y^l Z^l)$ edges, and
	$(l+m)$ Streett pairs.
	A Streett game with $n'$ vertices, $m'$ edges and $k'$ pairs is solved in
	$\mathcal{O}(m'\,(n')^{k'}\,k'!)$ time and $k'!$ memory~\cite{piterman2006faster}, so this game
	costs
	\[
		\mathcal{O}\big(|\E|\,Y^l Z^l\cdot(|\V|\,Y^l Z^l)^{l+m}\cdot(l+m)!\big)
		=\mathcal{O}\big(|\E|\,|\V|^{l+m}\,Y^{l(l+m+1)}\,Z^{l(l+m+1)}\,(l+m)!\big),
	\]
	and it uses $(l+m)!$ memory with respect to this new regular Streett game's
	state space, i.e., with respect to the original game $\G$, the required
	memory will be $\mathcal{O}(Y^lZ^l(l+m)!)$.
	By $l\le k$ this is at most
	$\mathcal{O}\big(|\E|\,|\V|^{k+m}\,Y^{k(k+m+1)}\,Z^{k(k+m+1)}\,(k+m)!\big)$
	time and at most $\mathcal{O}(Y^kZ^k(k+m)!)$ memory.
	Multiplying the computation time by the $|S|=\mathcal{O}(2^{m+n})$
	requirements gives the stated time bound; for memory, we do not need to
	multiply with $|S|$ as only one of the strategies among all of the $|S|$ cases
	will survive.
	Notice that both computation time and memory are doubly exponential in $m$ and $n$ because $Y,Z$
	are.
\end{proof}

\propGlobalImpliesTerminal*

\begin{proof}
	Suppose, towards a contradiction, that $\polZ$ is gracefully robust from
	$\v^0$ but not terminally gracefully robust.
	Then there is a strategy $\polO'$ of \PO such that, writing $\G[U]$ for
	the limit subgame of $\polZ$ and $\polO'$, some graceful strategy
	$\polZ''$ of \PZ on $\G[U]$ dominates $\polZ$ within $\G[U]$ (uniformly
	over every vertex of $U$, as per Remark~\ref{rem:dominated within subgame}).

	Recall $U$ is exactly the set of vertices reachable from
	$\recurring(\pathZO(\polZ,\polO',\v^0))$: consequently, $U$ is closed
	under every edge of $\G$, since any edge out of a vertex reachable from
	$\recurring(\pathZO(\polZ,\polO',\v^0))$ lands on another vertex that is
	also reachable from it, hence also in $U$. So once a play enters $U$, it
	can never leave. This lets us define the composite strategy $\polZ'$ of
	\PZ that plays $\polZ$ until the history first enters $U$ (say at vertex
	$\v^t$), and plays $\polZ''$ thereafter.

	\emph{$\polZ'$ is graceful.} Fix any strategy $\polO$ of \PO. If
	$\pathZO(\polZ',\polO,\v^0)$ never enters $U$, it coincides entirely with
	$\pathZO(\polZ,\polO,\v^0)$, which is graceful by hypothesis. Otherwise,
	once the history enters $U$ at some vertex $\v^t$,
	Proposition~\ref{prop:prefix-independence of graceful strategies} (with
	$\polZ''$ playing the role of the winning obliging-game strategy from
	$\v^t$, since $\polZ''$ is graceful uniformly over $U$) gives that
	$\polZ'$ is graceful.

	\emph{$\polZ'$ dominates $\polZ$.} Gracefulness, and the comparisons
	$\bettereq$ and $\better$, depend only on the set of vertices visited
	infinitely often by the resulting path, not on any finite prefix. So for
	any $\polO$ under which $\pathZO(\polZ',\polO,\v^0)$ enters $U$ at
	$\v^t$, the tail from $\v^t$ onward agrees with the outcome of $\polZ''$
	against $\polO$ restricted to $\G[U]$, while the tail of
	$\pathZO(\polZ,\polO,\v^0)$ agrees with $\polZ$ against the same
	restricted $\polO$; the shared finite prefix up to $\v^t$ does not affect
	which vertices recur infinitely. Since $\polZ''\bettereq\polZ$ within
	$\G[U]$ for every subgame strategy of \PO, this gives
	$(\polZ',\polO)\bettereq(\polZ,\polO)$ for every $\polO$ (trivially so if
	$U$ is never entered).

	For the strict part of domination, since $\polZ''\better\polZ$ within
	$\G[U]$ for some subgame strategy $\polO''$ of \PO, define $\tilde\polO$
	to be the strategy of \PO that agrees with $\polO'$ until $U$ is entered,
	and with $\polO''$ thereafter, mirroring how $\polZ'$ itself switches
	from $\polZ$ to $\polZ''$ upon entering $U$. By the same tail-invariance
	argument as above, $(\polZ',\tilde\polO)\better(\polZ,\tilde\polO)$.

	So $\polZ'$ is a graceful strategy dominating $\polZ$ from $\v^0$,
	contradicting that $\polZ$ is gracefully robust.
\end{proof}

\section{Proofs of Claims in Section~\ref{sec:static design phase}}
\label{appendix:proofs of the static design section}

\thmStaticComputationSoundness*

\begin{proof}
	\textbf{Observation \#1:}
	First, observe that for every solution $(Z,\polZ)$ of
	$\Synt(\cdot,\cdot,\cdot)$, and for every strategy $\polO$ of \PO, if the initial
	vertex $\vinit$ is in $Z$, then the resulting infinite path always remains
	inside $Z$; here, the set $Z$ will be referred to as a \textit{trap} for \PO.
	This follows from the fact that obliging games reduce to solving $\omega$-regular games on extended graphs, and moreover, for every
	$\omega$-regular game, the winning strategy of \PZ is able to restrict the
	game within its winning region~\cite{piterman2006faster}.

	\smallskip
	\noindent\textbf{Observation \#2:}
	Second, the construction by Chatterjee et
	al.~\cite[Lemma~2]{chatterjee2010obliging} ensures that \PZ has a
	finite-memory strategy $\polZ^*$ to win the obliging game, as well as \PO
	has a finite-memory strategy $\polO^*$ to fulfill the assumptions.
	From this, it follows that for every initial vertex $\vinit$ in the winning
	domain for the obliging game, the resulting path
	$\pathZO(\G,\polZ^*,\polO^*,\vinit)$ contains every $A\in \Aset$ within
	fixed intervals: to see this, take the product $\G'$ of $\G$ and the memory
	automata for $\polZ^*$ and $\polO^*$, and the path
	$\pathZO(\G,\polZ^*,\polO^*,\vinit)$ is eventually periodic in $\G'$.

	\smallskip
	\noindent\textbf{Observation \#3:} Finally, observe the following
	monotonicity of the winning regions of the obliging games:
	For every $\Aset'\subseteq \Aset''$ and every $\Gset'\supseteq \Gset''$, it
	must hold that the winning region of the obliging game
	$(\G,\gr{\Aset'}{\Gset'}\cap \safe(W),\gbuchi(\Aset))$ is a subset of the
	winning region of the obliging game $(\G,\gr{\Aset''}{\Gset''}\cap \safe(W),\gbuchi(\Aset))$. 
	For instance, in Example~\ref{ex:gracefulness} and \ref{ex:graceful robustness}, consider the sets $\Aset'=\set{A_1}\subset
	\set{A_1,A_2}=\Aset''$, $\Gset' = \Gset''= \set{G_1}$, and observe that
	$W=\set{0,\ldots,9}$, and the winning region for
	$(\G,\gr{\Aset'}{\Gset'}\cap \safe(W),\gbuchi(\Aset))$ is $\set{0,\ldots,3}$, which is a subset of the winning region $\set{0,\ldots,
	9}$ for $(\G,\gr{\Aset''}{\Gset''}\cap \safe(W),\gbuchi(\Aset))$.
	Intuitively, the former is harder for \PZ as compared to the latter, because
	$\gr{\Aset'}{\Gset'}\subseteq \gr{\Aset''}{\Gset''}$.
	
	\medskip 
	The proof of the theorem follows a ranking argument as below.
	Define the function $\rank\colon W\to \mathbb{N}$, such that $\rank(v)$ for
	every $\v\in W$ is the iteration index of the outer loop (over guarantees)
	during which $\BR(\v,\Aset^*)\neq \bot$ for the first time.
	For example, if this happened to $\v$ in the first iteration with
	$\tilde{\Gset} = \Gset$, then $\rank(\v)=1$, and if this happened in the
	last iteration with $\tilde{\Gset} = \emptyset$, then $\rank(\v)=2^m$.
	Define $R^{(i)} = \set{\v\in W\mid \rank(\v)=i}$.
	
	Let $\rank_{\min} \coloneqq \min_{\v\in W} \rank(v)$ be the smallest rank of
	any vertex in the original GR(1) winning region $W$, and let $\v$ be a vertex with $\rank(\v)=\rank_{\min}$.
	From Observation~\#1 above, \PZ has a strategy---call it
	$\polZ^{(1)}$---using which the set $R^{(1)}$ is never left, and
	Condition~\eqref{item:def:optimal strategy:a}  of Definition~\ref{def:optimal strategy for assumption subset} is fulfilled by construction
	of $R^{(1)}$ and $\polZ^{(1)}$.
	Condition~\eqref{item:def:optimal strategy:b} is also fulfilled, because
	$\polZ^{(1)}$ fulfills $\safe(W)$ by design.
	On the other hand, Condition~\eqref{item:def:optimal strategy:a2} follows from
	Observation~\#2, and Condition~\eqref{item:def:optimal strategy:c} follows
	since the most number of guarantees are fulfilled (as the minimum rank is attained).
	This is the base case of our inductive argument on ranks.
	For simplicity, let us assume that $\rank_{\min}=1$, i.e., $R^{(1)}\neq
	\emptyset$; all arguments generalize when this is not the case.

	Now for the induction step, suppose \PZ has an optimal strategy---call it
	$\polZ^{(i)}$---against $\Aset^*$ in the region $R^{(i)}$ for some $1\leq i< 2^m$, such that $R^{(i)}$ is a trap for \PO.
	From this, we construct an optimal strategy $\polZ^{(i+1)}$ of \PZ against
	$\Aset^*$ from the region $R^{(i+1)}$.
	Firstly, by construction and from Observation~\#3, it follows that
	$R^{(i)}\subseteq R^{(i+1)}$.
	Consider some $\v\in R^{(i+1)}$, and let $\polZ$ be the winning
	strategy from $\v$ in the obliging game obtained by the $\Synt(\cdot,\cdot,\cdot)$ subroutine.
	The desired strategy $\polZ^{(i+1)}$ is constructed as follows:
	as long as the game stays in $R^{(i+1)}\setminus R^{(i)}$, \PZ uses $\polZ$,
	and as soon as the game enters $R^{(i)}$, \PZ resets the history and
	switches to $\polZ^{(i)}$; here, we invoked
	Proposition~\ref{prop:prefix-independence of graceful strategies} to maintain
	that finite prefixes are irrelevant after $R^{(i)}$ is entered.
	By construction, Conditions~\eqref{item:def:optimal strategy:a} and
	\eqref{item:def:optimal strategy:b} of Definition~\ref{def:optimal strategy for
	assumption subset} hold, regardless of whether the path stays forever in
	$R^{(i+1)}\setminus R^{(i)}$ or it eventually enters $R^{(i)}$.
	Condition~\eqref{item:def:optimal strategy:c} also holds, because of the
	decreasing order over the guarantees in the outer for loop, and how ranks
	are defined.
	Finally, from Observation~\#1, under $\polZ^{(i+1)}$, $R^{(i+1)}$ is a trap for \PO, completing
	the induction step.
\end{proof}

\thmStaticComplexity*

\begin{proof}
	First, let us analyze the complexity of each instance of $\Synt$, in
	Line~\ref{line:alg:static:solve game}.
	First, consider the game $(\G,\gr{\tilde{\Aset}}{\tilde{\Gset}})$, which can
	be turned into a one-pair Streett game with $\mathcal{O}(|V|\cdot m\cdot n)$
	vertices and $\mathcal{O}(|E|\cdot m\cdot n)$ edges (see, e.g.,
	\cite[Section~5]{DBLP:conf/cav/BloemCGHJ10}). 
	Now let us add the safety condition $\safe(W)$, and consider the game
	$(\G,\gr{\tilde{\Aset}}{\tilde{\Gset}}\cap\safe(W))$.
	To solve this game, we can rely on the known compositional feature of safety
	objectives.
	In particular, we first consider the pure safety game $(\G,\safe(W))$, and
	note that the winning region for \PZ is the entire safe set $W$: this
	follows from the fact that $W$ is a trap for \PO. 
	It follows that every strategy of \PZ is safe iff it picks an edge within
	$W$ from vertices in $W$; see for example works on
	shielding~\cite{bloem2015shield}.
	It also follows that~\cite{bloem2015shield}, in order to solve
	$(\G,\gr{\tilde{\Aset}}{\tilde{\Gset}}\cap\safe(W))$, we can simply prune
	all the unsafe edges $(u,v)$ with $u\in W\cap \VZ$ and $v\notin W$, and then
	consider the game $(\G',\gr{\tilde{\Aset}}{\tilde{\Gset}})$ in the pruned
	graph $\G'$; notice that $\G'$ has $\mathcal{O}(|V|)$ vertices and
	$\mathcal{O}(|E|)$ edges.
	Therefore, like the simpler game $(\G',\gr{\tilde{\Aset}}{\tilde{\Gset}})$,
	the game $(\G,\gr{\tilde{\Aset}}{\tilde{\Gset}}\cap\safe(W))$ also reduces
	to a one-pair Streett game with $\mathcal{O}(|V|\cdot m\cdot n)$
	vertices and $\mathcal{O}(|E|\cdot m\cdot n)$ edges.

	Since $\gbuchi(\Aset)$ can be represented as a $m$-pair Streett automaton
	with just $1$ state (transition-labeled), it follows from
	\cite[Chapter~4]{chatterjee2010obliging} that
	$\Synt(\G',\gr{\tilde{\Aset}}{\tilde{\Gset}},\gbuchi(\Aset))$ can be
	reduced to a Streett game with $m+1$ Streett pairs, $\mathcal{O}(|V|\cdot
	m\cdot n)$ vertices, and $\mathcal{O}(|E|\cdot m\cdot n)$ edges.
	Since a Streett game with $x$ vertices, $y$ edges, and $k'$ pairs is
	solved in $\mathcal{O}(y\,(x)^{k'}\,k'!)$ time and $k'!$
	memory~\cite{piterman2006faster}, we conclude that
	$\Synt(\G',\gr{\tilde{\Aset}}{\tilde{\Gset}},\gbuchi(\Aset))$ can be solved
	in $T_{\Synt} \coloneqq \mathcal{O}(|V|^{m+1}\cdot|E|\cdot (mn)^{m+2}\cdot (m+1)!)$
	and the strategy will need $(m+1)!$ memory.

	The two nested loops iterate over $\tilde{\Gset}\in 2^{\Gset}$ and
	$\tilde{\Aset}\in 2^{\Aset}$, performing $2^{n}\cdot 2^{m}=2^{m+n}$ iterations.
	Each iteration makes one call
	$\Synt(\G',\gr{\tilde{\Aset}}{\tilde{\Gset}},\gbuchi(\Aset))$, costing at most
	$T_{\Synt}$, followed by the update of $\BR$ in Line~\ref{line:alg:static:populate map Gamma}, which
	touches at most $|\V|\cdot 2^{m}$ entries (one per vertex $\v\in Z\subseteq\V$
	and per superset $\Aset'\supseteq\tilde{\Aset}$). Summing over the $2^{m+n}$
	iterations gives the time bound. 
	% The final
	% stitching loop (Line~\ref{line:alg:static:stitching strategies}) runs $2^{m}$
	% times and defines each $\static(\Aset')$ by a lookup in $\BR$, costing
	% $\mathcal{O}(|\V|\cdot 2^{m})$ in total.
	%  Since the two main loops
	% already perform $2^{m+n}$ updates of $\BR$, each touching up to $|\V|\cdot 2^m$
	% entries, this $\mathcal{O}(|\V|\cdot 2^m)$ is dominated and does not affect the
	% bound.

	For space, the algorithm retains $\leq 2^{m+n}$ distinct sub-strategies
	computed by $\Synt$ and stored in $\BR$, each of size at most $(m+1)!$,
	giving us the memory complexity bound.
\end{proof}

\section{Proofs of Claims in Section~\ref{sec:adaptive strategy}}
\label{appendix:proofs of the adaptive section}

\thmMonitorSoundness*

\begin{proof}
	Suppose $\rho = v^0v^1\ldots\in \pathsinf(\G)$ is a given path, and let
	$w^t$ and $\alpha^t$ denote the values of these variables at time $t$.
	Initially, $w^0=1$ and $\alpha^0$ is set to some arbitrary value in $(0,1)$.
	At each time step $t\geq 0$ and each $A\in \Aset$, the liveness estimate $w^t$ and the decay rate $\alpha^t$ are updated as:
	\begin{enumerate}[(a)]
		\item when $v^t\in A$, then the liveness estimate is reset to $1$, i.e.,
		$w_A^{t+1}=1$, and moreover, the decay rate gets discounted by the
		attenuation rate, i.e., $\alpha_A^{t+1} = \lambda\alpha_A^t$;\label{list
	item:assumption monitor:current state in A}
		\item otherwise, when $v^t\notin A$, the liveness estimate is discounted by $1$ minus
		the current decay rate, i.e., $w_A^{t+1} = (1-\alpha_A^t) w_A^t$, and
		the decay rate remains unchanged, i.e., $\alpha_A^{t+1}=\alpha_A^t$.\label{list item:assumption monitor:current state not in A}
	\end{enumerate}

	We explain why \eqref{list item:assumption monitor:current
	state in A} and \eqref{list item:assumption monitor:current state not in A}
	from above enforce asymptotic correctness, i.e., the parts
	\eqref{list item:assumption monitor:convergence to 1} and \eqref{list
	item:assumption monitor:convergence to 0} of Definition~\ref{def:liveness monitor}.
	
	For \eqref{list item:assumption monitor:convergence to 1}, observe that
	after the $k$-th visit to $A$, the discount factor gets reduced to
	$\lambda^k\alpha^0$. 
	In addition, since $A$ has a constant hitting time, bounded by an unknown
	constant $n$, between the $k$-th and $(k+1)$-th visits to $A$, the liveness
	estimate $w^t$ remains bounded from below by $(1-\lambda^k\alpha^0)^{n}$.
	Since $0<\lambda<1$, it follows that as $k\to \infty$,
	$(1-\lambda^k\alpha^0)^{n}\to 1$.
	This establishes Part~\eqref{list item:assumption monitor:convergence to 1}.

	On the other hand, Part~\eqref{list item:assumption monitor:convergence to
	0} follows by simply observing that after the last visit to $A$, the weight
	decays constantly, and therefore $w_A^t\to 0$ as $t\to \infty$.
\end{proof}

\lemLowerBoundConsecutiveExecution*

\begin{proof}
	We wish to lower-bound the probability that \emph{some} statically
	computed strategy is selected at $N$ consecutive steps. 	Concretely, for a suitably chosen assumption $B\in\Aset$ we follow the single
	strategy $\static(\set{B})$ associated with the singleton $\set{B}$, and bound
	the probability that it is the one selected at each of the $N$ steps.

	We fix some notation for the rest of the proof. All superscripts count steps
	relative to the present time $\tau$, so step $t$ refers to global time
	$\tau+t$; thus $w_A^t$, $\tilde{w}_{\Bset}^t$, and $w_{\Bset}^t$ denote the
	values of these quantities at time $\tau+t$. We write
	\[
		S^t \;=\; \sum_{\Bset\in 2^{\Aset}} w_{\Bset}^t
	\]
	for the normalizing sum at step $t$, so that the probability of selecting the
	strategy of a set $\Bset$ at step $t$ is $p_{\Bset}^t = w_{\Bset}^t/S^t$; in
	particular, for the tracked singleton, $p_{\set{B}}^t = w_{\set{B}}^t/S^t$.
	We record two facts about a singleton $\set{B}$. Since it is a singleton set, the minimum score  in $\set{B}$ is 	the single score $w_B^t$, and the term subtracted in $\tilde{w}_{\set{B}}^t$ is at
	most $1$, so
	\begin{equation}\label{eq:p2:anchorD}
		1+\tilde{w}_{\set{B}}^t
		= 1+w_B^t-\max\big(\max_{A\neq B}w_A^t,\ \gamma^{\tau+t}\big)\ \geq\ w_B^t,
	\end{equation}
	and, dropping $w_B^t\geq 0$,
	\begin{equation}\label{eq:p2:anchorF}
		1+\tilde{w}_{\set{B}}^t\ \geq\ 1-\max\big(\max_{A\neq B}w_A^t,\ \gamma^{\tau+t}\big).
	\end{equation}

	\smallskip
	\noindent\emph{A ratio estimate.}
	Suppose, for the tracked singleton $\set{B}$ over $N$ consecutive steps, we
	have $1+\tilde{w}_{\set{B}}^t\geq L_t$ with $L_t\in(0,2]$. 
	The condition $L_t\le 2$ holds automatically. Indeed,
	$\tilde{w}_{\set{B}}^t = w_B^t - \max(\max_{A\neq B}w_A^t,\ \gamma^{\tau+t})$,
	where the subtracted term is at least $\gamma^{\tau+t}>0$ and $w_B^t\le 1$, so
	$\tilde{w}_{\set{B}}^t < 1$; and $\tilde{w}_{\set{B}}^t > -1$ because $w_B^t>0$
	and the subtracted term is at most $1$. Hence $1+\tilde{w}_{\set{B}}^t$ lies
	strictly between $0$ and $2$.

	We now derive a lower bound for $p_{\set{B}}^t$, and using the simplification
	\[ p_{\set{B}}^t = \frac{w_{\set{B}}^{t}}{S^{t}}
	=\frac{w_{\set{B}}^{t}}{ \sum_{\mathbf{B}\in 2^{\Aset}} w_{\set{B}}^{t}}
	= \frac{1}{1+ \sum_{\mathbf{B}\neq\set{B}} \frac{
	w_{\mathbf{B}}^{t}}{w_{\set{B}}^{t}}},
	\]
	this amount to finding an upper bound for the ratio $\frac{
	w_{\mathbf{B}}^{t}}{w_{\set{B}}^{t}}$ for every $\Bset$.
	Recall $w_{\Bset}^t=(1+\tilde{w}_{\Bset}^t)/(1-\tilde{w}_{\Bset}^t)$,
	so for any  set $\Bset$ this ratio is:
	\[
		\frac{w_{\Bset}^t}{w_{\set{B}}^t}
		= \frac{1+\tilde{w}_{\Bset}^t}{1-\tilde{w}_{\Bset}^t}\cdot
		  \frac{1-\tilde{w}_{\set{B}}^t}{1+\tilde{w}_{\set{B}}^t}
		= \frac{(1+\tilde{w}_{\Bset}^t)\,(1-\tilde{w}_{\set{B}}^t)}
		       {(1-\tilde{w}_{\Bset}^t)\,(1+\tilde{w}_{\set{B}}^t)} .
	\]
	We bound the four factors separately: In the numerator, $1+\tilde{w}_{\Bset}^t\le 2$ (as
	$\tilde{w}_{\Bset}^t\le 1$) and $1-\tilde{w}_{\set{B}}^t\le 2$ (as
	$\tilde{w}_{\set{B}}^t\ge -1$). In the denominator,
	$1-\tilde{w}_{\Bset}^t\ge\gamma^{\tau+t}$ holds for every $\Bset$ (follows
	from Eqn.~\eqref{eq:definition of w tilde for nonempty sets} and \eqref{eq:definition of w tilde for empty sets}), and
	$1+\tilde{w}_{\set{B}}^t\ge L_t$ by assumption. Combining,
	\[
		\frac{w_{\Bset}^t}{w_{\set{B}}^t}
		\ \le\ \frac{2\cdot 2}{\gamma^{\tau+t}\cdot L_t}
		\ =\ \frac{4}{\gamma^{\tau+t}\,L_t}.
	\]
	There are at most $2^m$ sets $\Bset\neq\set{B}$ (actually, $2^m-1$ such
	sets, but we need an upper bound), so the sum in the denominator
	of $p_{\set{B}}^t$ is bounded by
	\[
		\Sigma\ \coloneqq\ \sum_{\Bset\neq\set{B}}\frac{w_{\Bset}^t}{w_{\set{B}}^t}
		\ \le\ 2^m\cdot\frac{4}{\gamma^{\tau+t}\,L_t}
		\ =\ \frac{2^{m+2}}{\gamma^{\tau+t}\,L_t}.
	\]
	Finally, let $U\coloneqq 2^{m+2}/(\gamma^{\tau+t}L_t)$, so that $\Sigma\le
	U$ from above. Since $\gamma^{\tau+t}<1$ and $L_t\le 2$, we have
	$\gamma^{\tau+t}L_t\le 2<2^{m+2}$, so $U>1$; combined with $\Sigma\le U$,
	this gives $1+\Sigma\le 1+U<2U$, and
	\begin{equation}\label{eq:p2:perstep}
		p_{\set{B}}^t\ =\ \frac{1}{1+\Sigma}\ >\ \frac{1}{2U}
		\ =\ \frac{\gamma^{\tau+t}\,L_t}{2^{m+3}} .
	\end{equation}
	\medskip
	We split according to whether any assumption reaches score $\theta\coloneqq\tfrac12$
	during the first $N$ steps of the window.

	\smallskip
	\noindent\textbf{Case (i): some assumption has score at least $\theta$ at some
	step in $[\tau+1,\tau+N]$.}
	This merges the situation where the top score is already at least $\theta$ at
	the start and the situation where an assumption is visited during the window
	(a visit resets a score to $1$). Let $t^*\in[1,N]$ be the first such
	step and $A'$ an assumption with $w_{A'}^{t^*}\geq\theta$. We track $\set{A'}$
	over the $N$ steps $t^*,\dots,t^*+N-1$, which lie in $[\tau+1,\tau+2N]$ as
	$t^*\leq N$. By the monitor update
	rule, at each step the score either resets to $1$ or is multiplied by
	$(1-\alpha)$ for the current decay rate $\alpha\le\alpha^0$; in both cases
	$w_{A'}^{t}\ge(1-\alpha^0)\,w_{A'}^{t-1}$ (if it resets, $w_{A'}^t=1\ge(1-\alpha^0)w_{A'}^{t-1}$;
	if it decays, $(1-\alpha)\ge(1-\alpha^0)$). So the score loses at most a factor
	$(1-\alpha^0)$ per step, and expanding from time $t^*$, where $w_{A'}^{t^*}\ge\theta$,
	\[
		w_{A'}^t\ \ge\ (1-\alpha^0)^{\,t-t^*}\,w_{A'}^{t^*}\ \ge\ \theta\,(1-\alpha^0)^{\,t-t^*}
		\qquad\text{for } t=t^*,\dots,t^*+N-1.
	\]
	Since $\set{A'}$ is a singleton, Eqn.~\eqref{eq:p2:anchorD} gives
	$1+\tilde{w}_{\set{A'}}^t\ge w_{A'}^t$, and combining with the score bound above,
	\[
		1+\tilde{w}_{\set{A'}}^t\ \ge\ w_{A'}^t\ \ge\ \theta\,(1-\alpha^0)^{\,t-t^*} \eqqcolon L_t.
	\]
	The ratio estimate can be applied :   $\set{A'}$ is the
	tracked singleton, and $L_t$ satisfies $0<L_t\le\theta=\tfrac12\le 2$, so
	Eqn.~\eqref{eq:p2:perstep} applies at each of the $N$ steps $t=t^*,\dots,t^*+N-1$
	and yields
	\[
		p_{\set{A'}}^t\ \ge\ \frac{\gamma^{\tau+t}\,L_t}{2^{m+3}}
		\ =\ \frac{\gamma^{\tau+t}\,\theta\,(1-\alpha^0)^{\,t-t^*}}{2^{m+3}}.
	\]
	
 Substituting into Eqn.~\eqref{eq:p2:perstep} and
	multiplying over these $N$ steps,
	\[
		P_N\ \geq\ \prod_{j=0}^{N-1}\frac{\gamma^{\tau+t^*+j}\,\theta\,(1-\alpha^0)^j}{2^{m+3}}
		\ \geq\ \frac{\theta^N\,(\gamma^{\tau+2N})^N\,(1-\alpha^0)^{N(N-1)/2}}{2^{(m+3)N}}
		\ =\ \frac{(\gamma^{\tau+2N})^N\,(1-\alpha^0)^{N(N-1)/2}}{2^{(m+4)N}},
	\]
	using $\theta=\tfrac12$, $\prod_{j}\gamma^{\tau+t^*+j}\geq(\gamma^{\tau+2N})^N$,
	and $\sum_{j=0}^{N-1}j=N(N-1)/2$. Since $2^{(m+4)N}\leq 2^{(m+5)N}$, this is at
	least the bound in Eqn.~\eqref{eq:probability lower bound of seeing a streak of one strategy}.

	\smallskip
	\noindent\textbf{Case (ii): every assumption has score below $\theta$ at every
	step in $[\tau+1,\tau+N]$.}
	Track $\set{B^*}$ for any fixed assumption $B^*$. Then
	$\max_{A\neq B^*}w_A^t<\theta$ for all $1\leq t\leq N$ ((here step $t$ means global time $\tau+t$), so by
	Eqn.~\eqref{eq:p2:anchorF},
	\[
		1+\tilde{w}_{\set{B^*}}^t\ \geq\ 1-\max\big(\max_{A\neq B^*}w_A^t,\ \gamma^{\tau+t}\big)
		\ \geq\ 1-\max(\theta,\gamma^{\tau+t})\eqqcolon L_t.
	\]
Then $L_t\geq 1-\gamma^1>0$ is bounded below by a constant, uniformly in $t$:
	since $\gamma^k=1/\ln(3+k)$ is decreasing in $k$, we have
	$\gamma^{\tau+t}\le\gamma^{\tau+1}\le\gamma^1$ for all $1\le t\le N$, and as
	$\theta=\tfrac12\le\gamma^1=\tfrac{1}{\ln 4}$, we get $\max(\theta,\gamma^{\tau+t})\le\gamma^1$; hence
	\[
		L_t = 1-\max(\theta,\gamma^{\tau+t})\ \geq\ 1-\gamma^1
		= 1-\tfrac{1}{\ln 4}\ \approx\ 0.28\ >\ 0 .
	\]	
 Substituting into Eqn.~\eqref{eq:p2:perstep} and multiplying over the
	first $N$ steps,
	\[
		P_N\ \geq\ \prod_{t=1}^{N}\frac{\gamma^{\tau+t}\,L_t}{2^{m+3}}
		\ \geq\ \frac{(\gamma^{\tau+2N})^N\,(1-\gamma^1)^N}{2^{(m+3)N}},
	\]
	which exceeds the bound in Eqn.~\eqref{eq:probability lower bound of seeing a streak of one strategy}
	because $(1-\gamma^1)^N\,2^{2N}=\big(4(1-\tfrac{1}{\ln
	4})\big)^N\geq 1\geq(1-\alpha^0)^{N(N-1)/2}$ (using $(m+5)-(m+3)=2$).
	
\noindent\textit{Remark.}
The threshold $\theta=\tfrac12$ is arbitrary: any fixed $\theta\in(0,1/\ln(4)]$
works, with the same two cases and the per-step estimate \eqref{eq:p2:perstep};
only the resulting constant changes (Case (i) contributes $\theta^N$ and Case
(ii) contributes $(1-\max(\theta,\gamma^1))^N$). 
\end{proof}

\lemRandomSchedulingPreservesEnvFriendliness*

\begin{proof}
	From Theorem~\ref{thm:static computation soundness}, we know that the
	statically computed strategy  corresponding to every subset $\Aset^*\in
	2^{\Aset}$ of assumptions is optimal, and therefore graceful (follows from
	Definition~\ref{def:optimal strategy for assumption subset}).
	In detail, for every fixed $\Aset^*\in 2^{\Aset}$ with $\polZ^* =
	\static(\Aset^*)$, and for every initial vertex $\v$, there exists a strategy
	$\polO^*$ of \PO, call it $\mathit{response}(\polZ^*)$, such that \textit{all} the assumptions are fulfilled in
	future, i.e., $\gbuchi(\Aset)$ is fulfilled in $\rho^*\coloneqq\pathZO(\G,\polZ^*,\polO^*,\v)$.
	Moreover, both $\polZ^*$ and $\polO^*$ are the finite-memory strategies produced
	by the same obliging-game-solving construction used in the proof of
	Theorem~\ref{thm:static computation soundness}, so, as observed there,
	$\rho^*$ is eventually periodic; consequently every $A\in\Aset$, being
	visited infinitely often on $\rho^*$, has a constant hitting time $n^*$ on
	$\rho^*$.
	Define
	\begin{multline*}
	N\coloneqq \max\Big[\bigcup_{\Aset^*\in 2^\Aset}\bigcup_{A\in \Aset}\bigcup_{\v\in \V}\big\lbrace \text{ the hitting time } n^* \text{ of } A \text{ in } \rho^* \mid\\ \polZ^* = \static(\Aset^*), \polO^*=\mathit{response}(\polZ^*),\rho^*=\pathZO(\G,\polZ^*,\polO^*,\v)\big\rbrace\Big].
	\end{multline*}
	Intuitively speaking, for any arbitrary $\Aset^*\in 2^{\Aset}$, if we show
	that on any given path produced by the adaptive strategy, there will be
	infinitely many windows of length $N$ where \PZ acts according to $\polZ^*=\static(\Aset^*)$,
	then it is sufficient for \PO to stick to the strategy
	$\mathit{response}(\polZ^*)$ all the time, and it will be guaranteed that
	all assumptions are met, i.e., $\gbuchi(\Aset)$ is fulfilled.
	This is because, even though $\mathit{response}(\polZ^*)$ may not fulfill
	$\gbuchi(\Aset)$ at other times when \PZ is \textit{not} using $\polZ^*$, it
	is guaranteed that \PO will  be ``lucky'' infinitely often with $\polZ^*$
	being executed $N$ steps in a row, and by our definition of $N$, each of
	these infinite windows will allow sufficient time for all the assumptions in
	$\Aset$ to be visited at least once.

	Divide time as an infinite sequence of blocks of $2N$ steps as follows:
	$(0,\ldots,2N-1), (2N,\ldots,4N-1),\ldots$.
	Let $E_i$ be the event that within the $i$-th block $(2iN,\ldots,
	2(i+1)N-1)$, for $i=0,1,\ldots$, any one of the static optimal strategies of \PZ
	is continuously active for $N$ steps.

	Define the filtration $(\mathcal{F}_i)_{t\geq 0}$ where $\mathcal{F}_i$ is
	the smallest $\sigma$-algebra generated by all cylinder sets defined by
	paths ending at the block $i$, i.e., paths of length $2iN$.
	From Lemma~\ref{claim:lower bound on N-step consecutive execution of any one
	strategy}, we have $\P[E_i \mid \mathcal{F}_{i-1}] \geq
	(\gamma^{2iN+2N})^N\cdot\varepsilon$ where $\varepsilon\coloneqq
	\frac{(1-\alpha^0)^{N(N-1)/2}}{2^{(m+5)N}}>0$ is a constant.
	Observe that:
	\begin{align*}
		\sum_{i=1}^\infty \P[E_i \mid \mathcal{F}_{i-1}] \geq
		\sum_{i=1}^\infty (\gamma^{2iN+2N})^N\cdot\varepsilon
		&= \varepsilon\sum_{i=1}^\infty \frac{1}{\ln(3+2iN+2N)^N},
	\end{align*}
	which we show to diverge to $\infty$.
	To see this, we use the integral test for convergence, which states that a
	given series $\sum_{i=0}^\infty f(i)$, with continuous, positive, and
	decreasing $f$, converges iff the integral
	$\int_{x=0}^{\infty}f(x)dx$ is finite~\cite[Theorem~9]{thomas2018calculus}.
	Let $u = \ln(3+2xN+2N)$, so that $e^u = 3+2xN+2N$, so that $e^udu = 2Ndx$.
	With this change of variable, we obtain:
	\[
	\int_{x=0}^{\infty} \frac{1}{\ln(3+2xN+2N)^N} dx = \frac{1}{2N}\int_{u=\ln(3+2N)}^{\infty}\frac{e^u}{u^N}du. 
	\]
	Since $e^u$ grows faster than any polynomial, and in particular $u^N$, it
	follows that the above integral evaluates to $\infty$.	
	It follows that $\sum_{i=1}^\infty \P[E_i \mid \mathcal{F}_{i-1}] =
	 \infty$, so that we can apply the L\'evy's
	extension of Borel-Cantelli lemma (Lemma~\ref{lem:borel-cantelli} in
	Appendix~\ref{appendix:results in probability}) to conclude that, $\P$-almost surely,
	$E_i$ occurs for infinitely many $i$, i.e., in infinitely many blocks some
	static strategy is continuously active for $N$ steps. Since there are only
	$2^m$ static strategies in total, by the pigeonhole principle, at least one
	of them, say $\polZ^*$, must be the one continuously active in infinitely
	many of these blocks.
	By the definition of graceful strategies, we know that \PO has a strategy
	$\polO^*$ to fulfill all assumptions in $\Aset$ against $\polZ^*$.
	Then against the adaptive strategy, \PO can continuously use $\polO^*$ and it
	is guaranteed that every time the window of size $N$ appears when $\polZ^*$
	is scheduled continuously, all the assumptions will be visited at least once.
	This proves that there exists some strategy of \PO using which all the
	assumptions can be $\P$-almost surely fulfilled infinitely often.
\end{proof}

\lemConvergenceToOne*

\begin{proof}
	For every time $t>0$, define $\mathit{visits}_t$ as the random variable
	indicating the number of visits to $A$ up to time $t$, i.e., the number of
	steps $s\le t$ with $\v^s\in A$.
	It follows from the monitor construction that $w_A^t\geq
	(1-\lambda^{\mathit{visits}_t}\alpha^0)^t$.
	Notice that $w_A^t$ is a random variable, as it contains the random variable
$\mathit{visits}_t$.
	We will show that $(1-\lambda^{\mathit{visits}_t}\alpha^0)^t\toas 1$ as
	$t\to \infty$, thereby proving the claim.
	For this, we use $\ln(1-\lambda^{\mathit{visits}_t}\alpha^0)^t =
	t\ln(1-\lambda^{\mathit{visits}_t}\alpha^0) =
	-t\lambda^{\mathit{visits}_t}\alpha^0 +
	\mathcal{O}(t(\lambda^{\mathit{visits}_t}\alpha^0)^2)$.
	Therefore, we need to show that $t\lambda^{\mathit{visits}_t}\alpha^0\toas
	0$ as $t\to \infty$; note that this alone suffices, since
	$t(\lambda^{\mathit{visits}_t}\alpha^0)^2 \le \alpha^0\,(t\lambda^{\mathit{visits}_t}\alpha^0)$
	(as $\lambda^{\mathit{visits}_t}\in(0,1]$), so
	$t(\lambda^{\mathit{visits}_t}\alpha^0)^2\toas 0$ follows automatically
	once $t\lambda^{\mathit{visits}_t}\alpha^0\toas 0$ is established.
	Intuitively, this amounts to showing that, $\P$-almost surely,
	$\mathit{visits}_t$ grows sufficiently fast so that the decay
	$\lambda^{\mathit{visits}_t}$ dominates the increase in $t$.
	We start by deriving a probabilistic lower bound on $\mathit{visits}_t$ for
	any given time $t$.
	We will use the same trick as in the proof of Lemma~\ref{claim:random
	scheduling preserves environmentally friendliness}: from
	Assumption~\ref{assump:finite-memory environment}, we infer that, assuming
	$A$ is visited infinitely often, there is a bound $N$ on the number of steps
	between two consecutive visits to $A$ against any one of the statically produced optimal strategies of \PZ.
	Furthermore, from Lemma~\ref{claim:lower bound on N-step consecutive
	execution of any one strategy}, we know that at the current time $t$ and given
	any history, the probability of any one of the static optimal strategies being
	scheduled for $N$ times in a row, within any time interval of length $2N$,
	is bounded from below by $\varepsilon_t \coloneqq
	(\gamma^{t+2N})^N\cdot\varepsilon'$ where $\varepsilon'\coloneqq
	\frac{(1-\alpha^0)^{N(N-1)/2}}{2^{(m+5)N}}>0$ is a constant.
	By Assumption~\ref{assump:finite-memory environment}, a single static
	strategy run for $N$ steps visits $A$ at least once, so an $N$-consecution
	within a block forces a visit to $A$; this is why the consecutive
	probability lower-bounds the visit probability, and in other words,
	$\varepsilon_t$ serves as the probability of seeing an $A$
	in the next block of size $2N$.
	Since $\gamma^0,\gamma^1,\ldots$ is a decreasing sequence, so is the
	sequence $\varepsilon_0,\varepsilon_1,\ldots$.
	Therefore, $\varepsilon_t$ serves as a valid lower bound on all
	$\varepsilon_i$ for $i\leq t$.
	Coming back to $\mathit{visits}_t$, by time $t$, there have been
	$n\coloneqq\lfloor t/2N\rfloor$ blocks of length $2N$, and in each of them,
	the probability of seeing at least one $A$ is at least $\varepsilon_t$,
	\emph{conditional} on the history before the block, since the blocks are
	not independent: by Lemma~\ref{claim:lower bound on N-step consecutive
	execution of any one strategy}, each block fails with conditional
	probability at most $1-\varepsilon_t$ given the past, so a fixed set of $j$
	blocks all fail with probability at most $(1-\varepsilon_t)^j$ by
	conditioning sequentially. This justifies the union bound below.
	Therefore, for every $m \leq n-2$,
	\[
		\P[\mathit{visits}_t < m+1] \leq \binom{n}{m}\cdot (1-\varepsilon_t)^{n-m}
		\leq \binom{n}{m}\cdot e^{-\varepsilon_t(n-m)},
	\]
	where the first inequality is obtained by taking the union bound over the
	probabilities of not seeing an $A$ (i.e., not seeing a contiguous block of
	length $N$ where only one static strategy is used) in any $n-m$ blocks
	(if at most $m$ blocks succeed then at least $n-m$ fail, and a union bound
	over the $\binom{n}{n-m}=\binom{n}{m}$ choices of which blocks fail gives
	the bound), and the second one uses the standard inequality $(1-x)^{r}\leq e^{-xr}$
	for all $r\geq 0$ and $0<x<1$ (which follows from $1-x\le e^{-x}$).
	For $m+1= (\ln(t))^p \leq n$ for a given $p$ with $n-m\geq t/4N$, and using
	$\binom{n}{m}\leq (en/m)^{m} \leq e^{(\ln(t))^{p+1}}$, where the first
	inequality is the standard bound $\binom{n}{m}\le(en/m)^m$ and the second one holds for \textit{sufficiently large} $t$, we obtain:
	\[
		\P[\mathit{visits}_t < (\ln(t))^p] \leq e^{(\ln(t))^{p+1} - \varepsilon'(n-m)/(\ln(3+t+2N))^N} \leq
		e^{(\ln(t))^{p+1} - \varepsilon't/(4N(\ln(3+t+2N))^N)}.
	\]
	
		Writing $c \coloneqq \varepsilon'/4N$, since as $t\to \infty$, $ct/(\ln(3+t+2N))^N$ dominates $(\ln(t))^{p+1}$, and
	furthermore $ct/(\ln(3+t+2N))^N\to \infty$, it follows that
	$\P[\mathit{visits}_t < (\ln(t))^p]$ decays super-polynomially with $t$, giving us:
	\[
		\sum_{t=1}^\infty\P[\mathit{visits}_t < (\ln(t))^p] < \infty.
	\]
	By the first Borel-Cantelli lemma (see Lemma~\ref{lem:borel-cantelli} in
	Appendix~\ref{appendix:results in probability}), we conclude that for
	sufficiently large $t$, $\P$-almost surely it will always be the case that
	$\mathit{visits}_t > (\ln(t))^p$.
	
	Taking $p=2$, we will have $\mathit{visits}_t\geq (\ln(t))^2$ for all large
	$t$, $\P$-almost surely.
	Let $\beta>0$ be the constant such that $\ln(\lambda) = -\beta$.
	Then,
	\[
		t\lambda^{\mathit{visits}_t}\alpha^0 \leq t\lambda^{(\ln(t))^2}\alpha^0
		= te^{-\beta(\ln(t))^2}\alpha^0
		= e^{\ln(t) - \beta(\ln(t))^2}\alpha^0
		= e^{-\ln(t)(\beta\ln(t)-1)}\alpha^0.
	\]
	Since both $\ln(t)\to \infty$ and $(\beta\ln(t)-1)\to \infty$ as $t\to
	\infty$, while $\alpha^0$ is a constant, we have $t\lambda^{\mathit{visits}_t}\alpha^0\to 0$, $\P$-almost surely.
	Similar argument would show that, as $t\to \infty$,
	$t(\lambda^{\mathit{visits}_t}\alpha^0)^2\to 0$, $\P$-almost surely,
	completing the proof.
\end{proof}

\thmOverallSoundness*

\begin{proof}[Proof of Theorem~\ref{thm:overall soundness}]
	We start by showing Condition~\eqref{item:def:terminally gracefully
	robust:robustness} of Definition~\ref{def:terminally gracefully robust strategies}, i.e., the adaptive strategy is undominated in every
	limit subgame, $\P$-almost surely.
	Since Assumption~\ref{assump:finite-memory environment} holds, we know that
	the assumptions $\Aset$ are partitioned into the sets $\Asetinf$ and
	$\Asetfin$, such that every strategy of \PO will fulfill
	$\grabin(\Asetinf,\Asetfin)$.
	From Lemma~\ref{claim:convergence of w to 0 for finitely visited
	assumptions}, we know that for every $A\in\Asetfin$, $w_A^t\to 0$ as $t\to
	\infty$, and from Lemma~\ref{claim:convergence of w to 1 for infinitely
	visited assumptions}, we know that for every $A\in \Asetinf$, $w_A^t\to 1$
	as $t\to \infty$, $\P$-almost surely.
	Applying this to Algorithm~\ref{alg:adaptive strategy}, as $t\to \infty$, we
	observe that $\P$-almost surely the following hold:
	\begin{itemize}
		\item $ w_{\Asetinf}^t\to \infty$:
		First observe that $\tilde{w}_{\Asetinf}^t\to 1$ as $t\to \infty$,
		because in Eqn.~\eqref{eq:definition of w tilde for nonempty sets},
		$\gamma^t\to 0$, while $\min_{A\in \Asetinf}\set{w_A^t}\to 1$ and
		$\max_{A\notin\Asetinf} \set{w_A^t}\to 0$.
		Secondly, $\tilde{w}_{\Asetinf}^t\to 1 \implies w_{\Asetinf}^t\to
		\infty$ from the definition of $w_{\Aset^*}^t$ for any $\Aset^*\in 2^\Aset$ in Eqn.~\eqref{eq:def:w for sets of assumptions}.

		\item For every $\Aset^*\in 2^{\Aset}$ with $\Aset^*\neq \Asetinf$,
	$w_{\Aset^*}^t$ remains bounded as $t\to \infty$ (in most cases $w_{\Aset^*}^t\to 0$,
	with two exceptions below where $w_{\Aset^*}^t\to 1$ instead; boundedness in
	every case is all that is needed for the argument that follows).

	\textbf{Case (i): $\Aset^*$ is neither a subset nor a superset of $\Asetinf$.}
	First observe that $\tilde{w}_{\Aset^*}^t\to -1$ as $t\to \infty$, because
	in Eqn.~\eqref{eq:definition of w tilde for nonempty sets}, $\gamma^t\to 0$,
	while $\min_{A\in \Aset^*}\set{w_A^t}\to 0$, and $\max_{A\notin\Aset^*}
	\set{w_A^t}\to 1$ (since there exists $A\in \Asetinf\neq \Aset^*$ with
	$w_A^t\to 1$).
	Secondly, $\tilde{w}_{\Aset^*}^t\to -1 \implies w_{\Aset^*}^t\to
	0$ from the definition of $w_{\Aset^*}^t$ for any $\Aset^*\in 2^\Aset$ in Eqn.~\eqref{eq:def:w for sets of assumptions}.

	\textbf{Case (ii): $\Aset^*=B$ is a nonempty proper subset of $\Asetinf$.}
	Every member of $B$ is infinitely often visited, so $\min_{a \in B}w_a^t \to 1$.
	Also, $\max_{a \notin B}w_a^t \to 1$ (since $\Asetinf\setminus B$ is
	nonempty and every one of its members is infinitely often visited), so
	\[
		\tilde{w}_{B}^t
		= \min_{a\in B}w_a^t-\max\!\Big(\max_{a\notin B}w_a^t,\gamma^t\Big)
		\longrightarrow 1-1 = 0,
		\qquad\text{so}\qquad
		w_{B}^t=\frac{1+\tilde{w}_{B}^t}{1-\tilde{w}_{B}^t}\to 1.
	\]

	\textbf{Case (iii): $\Aset^*=B$ is a superset of $\Asetinf$ containing a
	finitely visited assumption.}
	Now $B$ contains some $a'\in\Asetfin$, so $\min_{a\in B }w_a^t\to 0$;
	and everything outside $B$ lies in $\Asetfin$, so
	$\max_{a\notin B }w_a^t\to 0$. Hence
	\[
		\tilde{w}_{B}^t \longrightarrow 0-0 = 0,
		\qquad\text{so}\qquad
		w_{B}^t\to 1.
	\]

	In every case (i)--(iii), $\tilde{w}_{\Aset^*}^t\le 1-\gamma^t<1$, so
	$w_{\Aset^*}^t$ remains bounded, as claimed.
	\end{itemize}
	Combining these two, we conclude that, as $t\to \infty$, $p_{\Asetinf}^t =
	\frac{w_{\Asetinf}^t}{w_{\Asetinf}^t+\sum_{\Aset^*\neq\Asetinf}w_{\Aset^*}^t}\to
	1$ while $p_{\Aset^*}= \frac{w_{\Aset^*}^t}{w_{\Asetinf}^t+\sum_{\Aset^*\neq\Asetinf}w_{\Aset^*}^t}\to 0$ for every $\Aset^*\neq \Asetinf$.
	In other words, in the long run and $\P$-almost surely, the adaptive strategy
	will converge to the statically computed strategy $\static(\Asetinf)$.
	Since $\static(\Asetinf)$ is optimal for $\Asetinf$ and when no other
	assumptions are fulfilled (Theorem~\ref{thm:static computation soundness}), it follows that within the limit subgame that is
	reached, no other strategy would be able to fulfill a larger number of
	guarantees, i.e., it will not be dominated.

	\smallskip
	Now we proceed to show Condition~\eqref{item:def:terminally gracefully robust:gracefulness} of Definition~\ref{def:terminally gracefully robust strategies} holds.
	From Lemma~\ref{claim:random scheduling preserves environmentally
	friendliness}, we know that the adaptive strategy allows \PO to fulfill all the assumptions, $\P$-almost surely.
	Now, there are two possibilities: either \PO actually fulfills all the
	assumptions, i.e., $\Asetinf=\Aset$, or he violates some, i.e.,
	$\Asetinf\subset \Aset$.
	In the latter case, the GR(1) specification is satisfied automatically.
	In the former case, we show that all guarantees will be satisfied, almost
	surely and always eventually.

	First, from our analysis above, observe that $p_{\Aset}^t\to 1$ as $t\to
	\infty$, i.e., $\polZ^*\coloneqq\static(\Aset)$ will be deployed at each
	step with arbitrarily high probability as $t\to \infty$.
	Using the same analysis as in Lemma~\ref{claim:random scheduling preserves
	environmentally friendliness}, it follows that for every arbitrary constant
	$N$, $\polZ^*$ will be deployed continuously for $N$ steps always
	eventually and $\P$-almost surely.

	Second, since all the static optimal strategies guarantee that the winning
	region 	$W$ of the obliging game $(\G,\gr{\Aset}{\Gset},\Aset)$ is never
	left (Condition~\eqref{item:def:optimal strategy:b} of
	Definition~\ref{def:optimal strategy for assumption subset}), therefore the
	adaptive strategy never leaves $W$.

	Third, within $W$, when $\Asetinf=\Aset$, the strategy $\polZ^*$ must
	fulfill all guarantees $\Gset$, because by definition of $W$, we know that
	all guarantees can be fulfilled, and therefore if $\polZ^*$ does not
	fulfill $\Gset$, then it is dominated.

	Finally, from Assumption~\ref{assump:finite-memory environment}, we can
	derive a bound on the length of the path that contains a visit to all $A\in
	\Asetinf$ by \PO, when \PZ continuously uses $\polZ^*$, which is:
	$N=\max_{A\in \Aset} \hittingtime(\polZ^*,A)$.
	Since $\polZ^*$ itself is a finite-memory strategy, it follows that there
	exists $N'>N$ such that all guarantees are also visited within every path
	segment of length $N'$.
	As we established above, $\polZ^*$ will be scheduled for arbitrary number of
	steps always eventually and $\P$-almost surely, it follows that there will
	be infinitely many segments of length larger than $N'$ during which all
	assumptions and all guarantees will be seen at least once.
	This implies, $\P$-almost surely, the adaptive policy is winning for the
	original GR(1) specification $\gr{\Aset}{\Gset}$, and is graceful with
	respect to the entire set of assumptions $\Aset$, establishing Condition~\eqref{item:def:terminally gracefully robust:gracefulness} of Definition~\ref{def:terminally gracefully robust strategies}.
\end{proof}

\begin{remark}\label{rem:bias needed for finiteness}
	The bias $\gamma^t$ ensures that $w_{\Bset}^t$ remains \emph{finite} at
	every finite time $t$, not just in the limit. Without it, consider
	$\Bset=\Aset$ (the full set): the term $\max_{A\notin \Aset} w_A^t$ is a
	maximum over the empty set, so it would not subtract anything, leaving
	$\tilde{w}_{\Aset}^t = \min_{A\in\Aset}w_A^t$, which can equal exactly $1$
	whenever every assumption has been visited simultaneously -- for
	instance, at $t=0$, since every monitor initializes with $w^0=1$. This
	would make $w_{\Aset}^0=(1+1)/(1-1)$ undefined immediately at
	initialization. The floor $\max(\cdot,\gamma^t)$ guarantees
	$1-\tilde{w}_{\Bset}^t\geq \gamma^t>0$ for every $\Bset$ (including
	$\Bset=\Aset$), keeping $w_{\Bset}^t$ bounded by $2/\gamma^t$ at every
	finite $t$, while $\gamma^t\to 0$ still allows $w_{\Bset}^t\to\infty$ in
	the limit for whichever $\Bset$ turns out to be $\Asetinf$ -- exactly the
	divergence used in the proof of Theorem~\ref{thm:overall soundness}.
\end{remark}

\section{Generalization of Asymptotic Liveness Monitors to Recurring Events with Sub-Exponential Hitting Times}
\label{sec:appendix:monitor generalization}

We will say a symbol $X$ has \textit{sub-exponential}
hitting time in the word $\rho$ if the hitting time of $X$ in $\rho$ is
bounded by every exponential function, i.e., $g(i) = 2^{o(i)}$, where $g\colon
\mathbb{N}\to \mathbb{N}$ is a function such that for every $i\in \mathbb{N}$
with $v^i \in X$, there exists $j \leq i+g(i)$ with $v^j\in X$..
Clearly, if $X$ has constant or sub-exponential hitting time in $\rho$, then
$\rho\in \buchi(X)$, but the other direction is not true: e.g., consider the
situation when the $k$-th visit to $X$ happens at $v^{2^k}$, i.e., $X$ is
visited infinitely often but the frequency gets exponentially small over time.

\begin{theorem}[Generalization of Theorem~\ref{thm:monitor soundness} to
sub-exponential hitting times]
	Let $W \subseteq \pathsinf(\G)$ be the set of infinite paths such that for
	every $\rho\in W$, either $\rho\notin \buchi(A)$ or $A$ has a
	sub-exponential hitting time in $\rho$.
	(The exact hitting time need not be available.)
	Then the monitor $\M\colon W\to (0,1)$ implemented by Algorithm~\ref{alg:liveness monitor} is an asymptotic liveness monitor for the specification $\buchi(A)$.
\end{theorem}

\begin{proof}
	Suppose $\rho = v^0v^1\ldots\in \pathsinf(\G)$ is a given path, and let
	$w^t$ and $\alpha^t$ denote the values of these variables at time $t$.
	Initially, $w^0=1$ and $\alpha^0$ is set to some arbitrary value in $(0,1)$.
	At each time step $t\geq 0$ and for every $A\in \Aset$, the liveness estimate $w^t$ and the decay rate $\alpha^t$ are updated as
follows:
	\begin{enumerate}[(a)]
		\item when $v^t\in A$, then the liveness estimate is reset to $1$, i.e.,
		$w_A^{t+1}=1$, and moreover, the decay rate gets discounted by the
		attenuation rate, i.e., $\alpha_A^{t+1} = \lambda\alpha_A^t$;\label{list
	item:assumption monitor:current state in A:appendix}
		\item otherwise, when $v^t\notin A$, the liveness estimate is discounted by $1$ minus
		the current decay rate, i.e., $w_A^{t+1} = (1-\alpha_A^t) w_A^t$, and
		the decay rate remains unchanged, i.e., $\alpha_A^{t+1}=\alpha_A^t$.\label{list item:assumption monitor:current state not in A:appendix}
	\end{enumerate}

	We explain why \eqref{list item:assumption monitor:current
	state in A:appendix} and \eqref{list item:assumption monitor:current state not in A:appendix}
	from above enforce asymptotic correctness, i.e., the parts
	\eqref{list item:assumption monitor:convergence to 1} and \eqref{list
	item:assumption monitor:convergence to 0} of Definition~\ref{def:liveness monitor}.
	
	For \eqref{list item:assumption monitor:convergence to 1}, observe that
	after the $k$-th visit to $A$, the discount factor gets reduced to
	$\lambda^k\alpha^0$. 
	In addition, since $A$ has a sub-exponential hitting time, bounded by
	$2^{o(k)}$, between the $k$-th and $(k+1)$-th visits to $A$, the liveness
	estimate $w^t$ remains bounded from below by $(1-\lambda^k\alpha^0)^{2^{o(k)}}$.
	We show that as $A$ is visited infinitely many times, i.e., as $k\to
	\infty$, we achieve $(1-\lambda^k\alpha^0)^{g(k)}\to 1$ :
	Firstly, for $-1\leq x<1$, using the Maclaurin series, we know that
	$\log(1-x) = -x+O(x^2)$~\cite[Chapter~11.10]{stewart2021calculus}.
	Since $-1\leq \lambda^k\alpha^0<1$, we obtain:
	\begin{align*}
		\log[(1-\lambda^k\alpha^0)^{2^{o(k)}}] 
		&= 2^{o(k)}[-\lambda^k\alpha^0 + O(\lambda^{2k}(\alpha^0)^2)].
	\end{align*}
	Since $\lambda^k$ dominates the sub-exponential term $2^{o(k)}$ as $k$
	grows, and since $0 < \lambda <1$, it follows that, as $k\to \infty$, we
	obtain $2^{o(k)}\lambda^k\alpha^0\to 0$ which implies
	$\log[(1-\lambda^k\alpha^0)^{2^{o(k)}}]\to 0$.
	From this the Part~\eqref{list item:assumption monitor:convergence to 1} is established.

	On the other hand, Part~\eqref{list item:assumption monitor:convergence to
	0} follows by simply observing that after the last visit to $A$, the weight
	decays constantly, and therefore $w_A^t\to 0$ as $t\to \infty$.
\end{proof}

% \section{Illustration of the liveness monitors and the time-varying
% distributions computed by adaptive strategies}\label{appendix:timing diagram}

% \input{FIGURES/matching_pennies_policy_illustration}

\section{Known results in probability theory}\label{appendix:results in probability}
\renewcommand{\Pr}{\mathbf{Pr}}
\newcommand{\Exp}{\mathbf{E}}

\begin{lemma}[The Borel-Cantelli lemmas~{\cite[Chapter~10]{shiryaev1996graduate}}]\label{lem:borel-cantelli}
	Let $E_1,E_2,\ldots$ be an infinite sequence of events in some probability
	space.
  \begin{description}
    \item[First lemma:] If $\sum_{n=1}^\infty \Pr[E_n] < \infty$, then the
    probability that infinitely many events occur is $0$, i.e.,
    $\Pr[\limsup_{n\to \infty} E_n] = 0$.
    \item[Second lemma:] If the events are all independent and
    $\sum_{n=1}^\infty\Pr[E_n]=\infty$, then $\Pr$-almost surely infinitely many
    of them occur, i.e., $\Pr[\limsup_{n\to \infty} E_n]=1$.
    \item[L\'evy's extension:] Suppose there exists a filtration
    $(\mathcal{F}_i)_{i\geq 1}$ such that $E_i\in \mathcal{F}_i$ for each $i$.
    If $\sum_{n=1}^\infty\Pr[E_n|\mathcal{F}_{n-1}]=\infty$, then $\Pr$-almost
    surely infinitely many of them occur, i.e., $\Pr[\limsup_{n\to \infty} E_n]=1$.
  \end{description}
\end{lemma}

\section{Experiment Benchmarks Description}\label{appendix:experiments}

\paragraph{Dining Philosophers} The dining philosophers problem is a classic concurrency problem that illustrates synchronization issues and resource sharing among multiple processes. In this problem, $n$ philosophers sit around a dining table, and each philosopher alternates between thinking and eating. To eat, a philosopher needs to pick up two forks placed between them and their neighbors. We characterize the problem with the following components:
\begin{itemize}
	\item \textbf{Environment Assumptions:} Each philosopher can either be thinking or eating. The environment ensures that a philosopher can only eat if both adjacent forks are available.
	\item \textbf{System Guarantees:} The system guarantees that each philosopher will eventually be able to eat, ensuring that no philosopher starves. Additionally, the system ensures that no two adjacent philosophers eat simultaneously, preventing conflicts over shared forks.
\end{itemize}
In the instance used in the paper, there are 54 states, 114 edges, 3 assumptions, 3 guarantees.

\paragraph{Lift Controller}
The lift controller benchmark models an $n$-floor elevator that serves requests from multiple floors. At each step, the environment may issue requests from any subset of floors, while the system controls the elevator, which may move up one floor, down one floor, or remain at its current position. Requests remain pending until the elevator reaches the corresponding floor, at which point they are cleared. We characterize the problem as a GR(1) specification with the following components:
\begin{itemize}
    \item \textbf{Environment Assumptions:} Each floor generates requests infinitely often, ensuring that every floor continues to issue service requests throughout the execution.
    \item \textbf{System Guarantees:} The system guarantees that every floor is visited infinitely often while a request from that floor is pending, ensuring that every recurring request is eventually served.
\end{itemize}
In the instances used in the paper, for n=2, 16 states, 34 edges, 2 assumptions, 2 guarantees; for n=3, 48 states, 66 edges, 3 assumptions, 3 guarantees.

\paragraph{Network Routing Arena}
The network routing benchmark models a packet-processing controller operating over a communication network. The arena is organized into several interconnected regions representing different stages of packet handling, including exploration, authentication, routing decisions, memory updates, recovery, and idle behavior. At each step, the environment controls packet-related events, such as the arrival of new packets and their authentication status, while the system selects routing decisions, updates its internal memory, and decides whether packets should be forwarded or dropped. The arena is shared across multiple benchmark instances, with different GR(1) specifications obtained by varying the environment assumptions and system guarantees. We characterize the problem as a GR(1) specification with the following components:
\begin{itemize}
\item \textbf{Environment Assumptions:} The environment specifies fairness assumptions over packet arrivals and packet attributes (e.g., authenticated or guest packets), ensuring that relevant events occur infinitely often according to the benchmark instance.
\item \textbf{System Guarantees:} The system specifies liveness objectives over routing decisions and memory updates, requiring that forwarding, dropping, or memory-related actions occur infinitely often or in response to recurring environment events, depending on the benchmark instance.
\end{itemize}
For all the instances used in the experimentation section, arena has 36 vertices, 90 edges, and the number of assumptions and guarantees vary from 2 to 3.

\paragraph{Generalized Buffer}
The generalized buffer benchmark models a collection of $n$ independent producer-consumer channels, each consisting of a single-slot buffer. The environment acts as the producer and may deposit data into any subset of empty buffers, while the system acts as the consumer and may consume data from any subset of occupied buffers. A buffer can only be consumed if it currently contains data, ensuring that invalid consumption actions are disallowed. We characterize the problem as a GR(1) game with the following components:
\begin{itemize}
    \item \textbf{Environment Assumptions:} Each producer generates data infinitely often, ensuring that every buffer is filled infinitely often.
    \item \textbf{System Guarantees:} The system guarantees that each buffer is emptied infinitely often by consuming the data it contains, ensuring that no buffer remains permanently occupied.
\end{itemize}
In the instances used in the paper, for n=2, 12 states, 29 edges, 2 assumptions, 2 guarantees; for n=3, 16 states, 54 edges, 3 assumptions, 3 guarantees.

\paragraph{Process Scheduler}
The process scheduler benchmark models a system responsible for allocating CPU time to one of $n$ processes. At each step, the environment specifies the set of processes requesting execution, while the system selects at most one process to schedule. We characterize the problem as a GR(1) specification with the following components:
\begin{itemize}
    \item \textbf{Environment Assumptions:} Each process requests execution infinitely often, ensuring that every process continues to compete for processor time throughout the execution.
    \item \textbf{System Guarantees:} The system guarantees that every process is scheduled infinitely often while it has a pending request, ensuring that no requesting process is starved.
\end{itemize}
In the instances used in the paper, for n=2, 16 states, 60 edges, 2 assumptions, 2 guarantees; for n=3, 40 states, 288 edges, 3 assumptions, 3 guarantees.

\paragraph{Bus Arbiter}
The bus arbiter benchmark models a shared communication bus accessed by $n$ competing processes. At each step, the environment specifies the set of processes requesting access to the bus, while the system selects at most one process to receive the grant. Mutual exclusion is enforced by ensuring that no two processes are granted access simultaneously. We characterize the problem as a GR(1) specification with the following components:
\begin{itemize}
    \item \textbf{Environment Assumptions:} Each process requests access to the shared bus infinitely often.
    \item \textbf{System Guarantees:} The system guarantees that every process is granted access to the bus infinitely often, ensuring that no process is permanently denied service.
\end{itemize}
In the instance used in the paper, for n=2 16 vertices, 61 edges, 2 assumptions, 2 guarantees.

% \AS{not completed yet}

\end{document}